\documentclass[12pt,technote, onecolumn, draft]{IEEEtran}
\usepackage{latexsym}
\usepackage{mathrsfs}
\usepackage{multirow}
\usepackage{amsfonts,amssymb,amsmath,amsthm,bm}
\usepackage{color}
\usepackage{cite}
\usepackage{comment}
\usepackage{graphicx}
\usepackage{caption}
\newtheorem{theorem}{Theorem}
\newtheorem{example}{Example}
\newtheorem{definition}{Definition}
\newtheorem{remark}{Remark}

\newtheorem{lemma}{Lemma}
\newtheorem{proposition}{Proposition}

\begin{document}
	\title{Self-orthogonal codes from plateaued functions and their applications in quantum codes and LCD codes$^{\dag}$}
	\author{Yadi Wei, Jiaxin Wang, Fang-Wei Fu
		\IEEEcompsocitemizethanks{\IEEEcompsocthanksitem Yadi Wei and Fang-Wei Fu  are with Chern Institute of Mathematics and LPMC, Nankai University, Tianjin 300071, China, Emails:  wydecho@mail.nankai.edu.cn,  fwfu@nankai.edu.cn.
			Jiaxin Wang is with School of Mathematics, Hefei University of Technology, Hefei 230601, China, Email: wjiaxin@hfut.edu.cn.}
		
		\thanks{$^\dag$This research is supported by the National Key Research and Development Program of China (Grant No. 2022YFA1005000), the National Natural Science Foundation of China (Grant Nos. 12141108, 62371259, 12411540221), the Fundamental Research Funds for the Central Universities of China (Nankai University), the Nankai Zhide Foundation.  }
	}

	\maketitle

	\begin{abstract}
		Self-orthogonal codes have received great attention due to their important applications in quantum codes, LCD codes and lattices. Recently, several families of self-orthogonal codes containing the all-$1$ vector were constructed by augmentation technique. In this paper, utilizing plateaued functions, we construct some classes of linear codes which do not contain the all-$1$ vector. We also investigate their punctured codes. The weight distributions of the constructed codes are explicitly determined. Under certain conditions, these codes are proved to be self-orthogonal. Furthermore, some classes of optimal linear codes are obtained from their duals. Using the self-orthogonal punctured codes, we also construct several new families of at least almost optimal quantum codes and optimal LCD codes. 
	\end{abstract}

	
	\section{Introduction}\label{Section1}
	Throughout this paper, let $p$ be an odd prime and  $n$ be a positive integer. Let $\mathbb{F}_{p}^{n}$ be the vector space of the $n$-tuples over $\mathbb{F}_{p}$, $\mathbb{F}_{p^n}$ be the finite field with $p^n$ elements, $V_{n}^{(p)}$ be an $n$-dimensional vector space over $\mathbb{F}_{p}$, and  $\langle \cdot \rangle_{n}$ denote a (non-degenerate) inner product of $V_{n}^{(p)}$. In this paper, when $V_{n}^{(p)}=\mathbb{F}_{p}^{n}$, let $\langle a, b\rangle_{n}=a \cdot b=\sum_{i=1}^{n}a_{i}b_{i}$, where $a=(a_{1}, \dots, a_{n}), b=(b_{1}, \dots, b_{n})\in \mathbb{F}_{p}^{n}$; when $V_{n}^{(p)}=\mathbb{F}_{p^n}$, let $\langle a, b\rangle_{n}=Tr_{1}^{n}(ab)$, where $a, b \in \mathbb{F}_{p^n}$, $Tr_{1}^{n}(\cdot)$ denotes the trace function from $\mathbb{F}_{p^n}$ to $\mathbb{F}_{p}$; when $V_{n}^{(p)}=V_{n_{1}}^{(p)}\times \dots \times V_{n_{t}}^{(p)} (n=\sum_{i=1}^{t}n_{i})$, let $\langle a, b\rangle_{n}=\sum_{i=1}^{t}\langle a_{i}, b_{i}\rangle_{n_{i}}$, where $a=(a_{1}, \dots, a_{t}), b=(b_{1}, \dots, b_{t})\in V_{n}^{(p)}$.
	
	Linear codes have a wide range applications in secret sharing schemes \cite{Anderson,Carlet1,Ding1,Yuan}, authentication codes \cite{Ding2}, association schemes \cite{Calderbank1}, and strongly regular graphs \cite{Calderbank2}. A linear code is called a self-orthogonal code if it is contained in its dual. Self-orthogonal codes have many applications including quantum codes \cite{Ling}, LCD codes \cite{Massey1} and  lattices \cite{Wan}. A linear code $\mathcal{C}$ is called a linear complementary dual code (LCD code for short) if it intersects its dual code
	trivially, i.e., $\mathcal{C}\cap\mathcal{C}^{\bot}=\{0\}$. LCD codes have garnered extensive attention due to their applications
	in communications and cryptography. In \cite{Carlet}, Carlet \emph{et al.} proved that any linear code over $\mathbb{F}_q$ is
	equivalent to a Euclidean LCD code for $q\ge 4$, where $q$ is a prime power. This also motivates researchers to study Euclidean LCD codes over $\mathbb{F}_2$ and $\mathbb{F}_3$. Quantum error-correcting codes (quantum codes for short) are designed to protect quantum information from decoherence. The construction of quantum codes with new and good parameters is an interesting topic. In \cite{Yu}, all the binary quantum codes of distance 3 were given. However, nonbinary quantum codes of distance 3 have not been totally constructed.
	
	There are a number of methods to construct linear codes, one of which is based on functions. Two generic constructions of linear codes from functions, called the first and the second generic constructions, have been distinguished from others and widely used by researchers. The first generic construction is obtained by
	\begin{equation}
		\mathcal{C}_{F}=\{c(a,b)=(Tr_1^m(aF(x))+\langle b,x\rangle_n)_{x\in {V_n^{(p)}}\setminus\{0\}}: a\in\mathbb{F}_{p^m},b\in V_n^{(p)}\},
	\end{equation}
	where $m$ is a positive integer and $F(x)$ is a function from $V_n^{(p)}$ to $\mathbb{F}_{p^m}$. The second generic construction, which is also called defining-set construction, is obtained by
	\begin{equation}
		\mathcal{C}_D=\{c(a)=(\langle a,x_1\rangle_n, \langle a, x_2\rangle_n, \cdots, \langle a,x_m\rangle_n): a\in V_n^{(p)}\},
	\end{equation}
	where $m$ is a positive integer and $D=\{x_1,x_2,\cdots,x_m\}\subseteq V_n^{(p)}$ is called the defining set of $\mathcal{C}_{D}$.
	
	Due to the properties of cryptographic functions, it's a good way to construct linear codes from cryptographic functions. By the first and the second generic constructions, several linear codes have been constructed from cryptographic functions such as bent functions \cite{Mesnager1,Ozbudak,Pelen,Tang1,Wang1,Zhou}, perfect nonlinear (PN) functions \cite{Carlet1,Li,Yuan1} and plateaued functions \cite{Cheng,Mesnager2,Mesnager3,Sinak,Yang}. In \cite{Mesnager1,Mesnager2}, based on the first generic construction, Mesnager \emph{et al.} constructed three-weight linear codes from weakly regular bent and plateaued functions. A function $f$ from $V_n^{(p)}$ to $\mathbb{F}_p$ is called an $l$-form if $f(cx)=c^lf(x)$ for all $c\in\mathbb{F}_p^*, x\in V_n^{(p)}$, where $1\le l\le p-1$ is an integer. In \cite{Mesnager3,Tang1},  based on the second generic construction, some linear codes with two or three weights were constructed from weakly regular bent and plateaued functions of $l$-form. However, the self-orthogonality of the linear codes in \cite{Mesnager1,Mesnager2,Mesnager3,Tang1} were not investigated. 
	Recently, in \cite{Heng1,Heng2}, using augmentation technique,  Heng \emph{et al.} constructed some self-orthogonal codes containing the all-$1$ vector from weakly regular bent functions.
	In \cite{Heng3}, they proved that for a linear code $\mathcal{C}$ over $\mathbb{F}_{p^n}$, which contains the all-$1$ vector, if all its codewords
	have weights divisible by $p$, then $\mathcal{C}$ is self-orthogonal.
	Later, in \cite{Cakmak,Mesnager4,Wang2,Wang3,Xie}, by augmentation technique, some self-orthogonal linear codes containing the all-$1$ vector were also constructed from vectorial dual-bent functions and weakly regular plateaued functions. Therefore, it is an interesting problem to study the self-orthogonality of linear codes obtained by the first and the second generic constructions and to construct some new self-orthogonal codes from these codes by other techniques.
	
	In this paper, based on the first and the second generic construction, using weakly regular and non-weakly regular plateaued functions, we first construct some linear codes and study the punctured codes of them. It is worth to mention that when using the second generic construction to construct linear codes, the non-weakly regular plateaued functions we used do not have to be $l$-form. The parameters and weight distributions of the constructed codes are determined. Under certain conditions, these codes are proved to be self-orthogonal. Meanwhile, we also determine the parameters of their dual codes and obtain some classes of optimal linear codes from the dual codes. Finally, we construct some new families of at least almost optimal  quantum codes and optimal LCD codes from the self-orthogonal punctured codes. The optimal linear codes we obtained are listed below:\\
	$(1)$ A family of $p$-ary $[(p-1)(p^{n-1}-\epsilon_0p^{\frac{n+s}{2}-1}), (p-1)(p^{n-1}-\epsilon_0p^{\frac{n+s}{2}-1})-n-1,3]$ linear codes, where $\epsilon_0=\pm 1$, $0\le s\le n-2$, $n+s\ge 4$, and $n-s\ge 4$ if $p=3$ and $\epsilon_0=1$.\\
	$(2)$ A family of $p$-ary $[(p-1)p^{n-1}, (p-1)p^{n-1}-n-1,3]$ linear codes, where $n\ge 2$.\\
	$(3)$ A family of $p$-ary $[\frac{p^{n-1}+(p-1)p^{\frac{n+s}{2}-1}-1}{p-1}, \frac{p^{n-1}+(p-1)p^{\frac{n+s}{2}-1}-1}{p-1}-n, 3]$ linear codes, where $0\le s \le n-4$. \\
	$(4)$ A family of $p$-ary  $[\frac{p^{n-1}-\epsilon_0p^{\frac{n+s}{2}-1}}{2},\frac{p^{n-1}-\epsilon_0p^{\frac{n+s}{2}-1}}{2}-n,3]$ linear codes, where $\epsilon_0=\pm 1$, $0\le s \le n-4$, and $p>3$ if $\epsilon_0=1$.\\
	$(5)$ A family of $p$-ary $[\frac{p^{n-1}+\epsilon_0p^{\frac{n+s-1}{2}}}{2},\frac{p^{n-1}+\epsilon_0p^{\frac{n+s-1}{2}}}{2}-n,3]$ linear codes, where $\epsilon_0=\pm 1$, $0\le s\le n-3$, and $p>3$ if $\epsilon_0=-1$.\\
	$(6)$ A family of $p$-ary $[(p-1)(p^{n-1}-\epsilon_0p^{\frac{n+s}{2}-1})+n+1,(p-1)(p^{n-1}-\epsilon_0p^{\frac{n+s}{2}-1}),3]$ LCD codes, where $\epsilon_0=\pm 1$, $0\le s\le n-2$, $n+s\ge 6$ for $p=3$, $n+s\ge 4$ for $p\ne 3$ , and $n-s\ge 4$ if $p=3$ and $\epsilon_0=1$.\\
	\noindent$(7)$ A family of $p$-ary $[(p-1)p^{n-1}+n+1, (p-1)p^{n-1},3]$ LCD codes, where $n\ge 3$ for $p=3$ and $n\ge 2$ for $p\ne 3$.\\
	$(8)$ A family of ternary $[\frac{3^{n-1}+2\times 3^{\frac{n+s}{2}-1}-1}{2}+n, \frac{3^{n-1}+2\times 3^{\frac{n+s}{2}-1}-1}{2},3]$ LCD codes, where $0\le s\le n-4$, $n+s\ge 6$.\\
	$(9)$ A family of $p$-ary $[\frac{p^{n-1}-\epsilon_0p^{\frac{n+s}{2}-1}}{2}+n, \frac{p^{n-1}-\epsilon_0p^{\frac{n+s}{2}-1}}{2},3]$ LCD codes, where $\epsilon_0=\pm 1$, $0\le s\le n-4$, $n+s\ge 6$ for $p=3$, and $p>3$ if $\epsilon_0=1$.\\
	$(10)$ A family of $p$-ary $[\frac{p^{n-1}+\epsilon_0p^{\frac{n+s-1}{2}}}{2}+n, \frac{p^{n-1}+\epsilon_0p^{\frac{n+s-1}{2}}}{2},3]$ LCD codes, where $\epsilon_0=\pm 1$, $0\le s\le n-3$, $n+s\ge 5$ for $p=3$, and $p>3$ if $\epsilon_0=-1$.
	
	The rest of the paper is organized as follows. In Section 2, we introduce the needed preliminaries. In Section 3, we present some auxiliary results. In Section 4, we construct several  families of self-orthogonal codes from weakly regular and non-weakly regular plateaued functions. In Section 5, we construct some new quantum codes and LCD codes. In Section 6, we make a conclusion.
	\section{Preliminaries}\label{Section2}
	In this section, we introduce some preliminaries which will be used for the subsequent sections. We begin this section by setting some basic notation. \\
	$\bullet$ $\#D$: The cardinality of a set $D$.\\
	$\bullet$ $\mathbb{F}_p^*$: The multiplicative group of the finite field $\mathbb{F}_p$.\\
		$\bullet$ $SQ$ and $NSQ$: The set of squares and nonsquares in $\mathbb{F}_p^{*}$, respectively.\\
		$\bullet$ $\eta$: The quadratic character of $\mathbb{F}_p^{*}$. For convenience, set $\eta(0)=0$. \\
		$\bullet$ $p^*$: $p^*=\eta(-1)p=(-1)^{\frac{(p-1)}{2}}p$.\\
		$\bullet$ $\xi_p=e^{\frac{2\pi \sqrt{-1}}{p}} $: The $p$-th complex primitive root of unity.\\
		$\bullet$ $\delta_0$: The indicator function, i.e., $\delta_0(0)=1$ and $\delta_0(a)=0$ for any $a\in V_n^{(p)}\setminus\{0\}$.\\
		$\bullet$ $e_i$: The $i$-th canonical basis vector in $\mathbb{F}_p^n$.
	\subsection{Cyclotomic field\ $\mathbb{Q}(\xi_p)$}
	The cyclotomic field $\mathbb{Q}(\xi_p)$ is obtained from the rational field $\mathbb{Q}$ by adjoining $\xi_p$. The ring of algebraic integers in $\mathbb{Q}(\xi_p)$ is $\mathcal{O}_{ {\mathbb{Q}_{(\xi_p)}} }=\mathbb{Z}[\xi_p]$ and  the set $\{\xi_p^j:1\le j\le p-1\}$ is an integral basis of it. The field extension $\mathbb{Q}(\xi_p)/\mathbb{Q}$ is Galois extension of degree $p-1$, and $Gal(\mathbb{Q}(\xi_p)/\mathbb{Q})=\{\sigma_a:\ a\in \mathbb{F}_p^{*}\}$ is the Galois group, where the automorphism $\sigma_a$ of $\mathbb{Q}(\xi_p)$ is defined by $\sigma_a(\xi_p)=\xi_p^a.$ The cyclotomic field $\mathbb{Q}(\xi_p)$ has a unique quadratic subfield $\mathbb{Q}(\sqrt{p^{*}})$, and so $Gal(\mathbb{Q}(\sqrt{p^{*}})/\mathbb{Q})=\{1,\ \sigma_\gamma\}$, where $\gamma$ is a nonsquare in $\mathbb{F}_p^{*}.$ Obviously, for $a\in \mathbb{F}_p^{*}$ and $b \in \mathbb{F}_p$, we have that $\sigma_a(\xi_p^b)=\xi_p^{ab}$ and $\sigma_a(\sqrt{p^*})=\eta(a)\sqrt{p^*}$. For more information on cyclotomic field, the reader is referred to \cite{Ireland}.
	\subsection{Exponential sums}
	\begin{lemma}[\cite{Lidl}]\label{Le 1}With the defined notation, we have\\
		$\mathrm{(1)}$$\sum\limits_{y\in \mathbb{F}_p^{*}}\eta(y)=0$.\\
		$\mathrm{(2)}$ $\sum\limits_{y\in \mathbb{F}_p^{*}}\eta(y)\xi_p^y=\sqrt{p^*}$\\
		$\mathrm{(3)}$ For any $a \in \mathbb{F}_p^{*}$, $\sum\limits_{y\in SQ}\xi_p^{ya}=\frac{\eta(a)\sqrt{p^*}-1}{2}$ and  $\sum\limits_{y\in NSQ}\xi_p^{ya}=\frac{-\eta(a)\sqrt{p^*}-1}{2}.$\\
		$\mathrm{(4)}$ For any $\alpha\in V_n^{(p)}$, $\sum\limits_{x\in V_n^{(p)}}\xi_p^{\langle x, \alpha\rangle_n}=\delta_0(\alpha)p^n.$ 
	\end{lemma}
	\begin{lemma}[\cite{Lidl}]\label{Le 2}Let $f(x)=a_2x^2+a_1x+a_0\in\mathbb{F}_p[x]$ with $a_2\ne 0$, then
		\[\sum_{c \in \mathbb{F}_p}\xi_p^{f(c)}=\xi_p^{a_0-a_1^2(4a_2)^{-1}}\eta(a_2)\sqrt{p^*}.\]
	\end{lemma}
	\subsection{Plateaued functions}
	
	Let $f(x):V_n^{(p)}\longrightarrow\mathbb{F}_p$ be a $p$-ary function. The  \emph{Walsh\  transform} of $f(x)$ is defined by 
	\[\widehat{\chi_f}(\alpha)=\sum\limits_{x\in V_n^{(p)}} \xi_p^{f(x)-\langle \alpha, x\rangle_n},\ \alpha \in V_n^{(p)}.\]  
	The function $f(x)$ is said to be \emph{balanced} over $\mathbb{F}_p$ if $f(x)$ takes every element of $\mathbb{F}_p$ with the same number $p^{n-1}$ of pre-images, i.e., $\widehat{\chi_f}(0)=0$. Otherwise, $f(x)$ is 
	\emph{unbalanced}.
	If $|\widehat{\chi_f} (\alpha)|\in\{0,p^{\frac{n+s}{2} }\} $ for any $\alpha\in V_n^{(p)}$, where $0\le s\le n$, then $f(x)$ is said to be \emph{$s$-plateaued}. Specially, when $s=0$, a $0$-plateaued function is called a \emph{bent} function. The \emph{Walsh support} of an $s$-plateaued function $f(x)$ is defined by $\mathrm{Supp}(\widehat{\chi_f})=\{\alpha\in V_n^{(p)}:\  |\widehat{\chi_f}(\alpha)|=p^{\frac{n+s}{2}} \}$. By the Parseval identity $\sum_{\alpha\in V_n^{(p)}}|\widehat{\chi_f}(\alpha)|^2=p^{2n}$, we have that  $\#\mathrm{Supp}(\widehat{\chi_f})=p^{n-s}$. The Walsh transform of an $s$-plateaued function $f(x)$ at $\alpha\in V_n^{(p)}$ is given as follows \cite{Hyun}.
	
	\[\widehat{\chi_f}(\alpha)=
	\begin{cases}
		\pm p^{\frac{n+s}{2}}\xi_{p}^{f^{*}(\alpha)}, 0, & \text{if } p^{n+s}\equiv 1\ (\mathrm{mod}\ 4);\\ 
		\pm \sqrt{-1}p^{\frac{n+s}{2}}\xi_{p}^{f^{*}(\alpha)}, 0,& \text{ if }  p^{n+s}\equiv 3\ (\mathrm{mod}\ 4),
	\end{cases}\] where $f^*(x)$ is a function from $\mathrm {Supp}(\widehat{\chi_f})$ to $\mathbb{F}_p$ called the \emph{dual} of $f(x)$.
	An $s$-plateaued function $f(x)$ is called \emph{weakly\ regular}, if for all $\alpha \in \mathrm {Supp}(\widehat{\chi_f})$,   $p^{-\frac{n+s}{2}}\widehat{\chi_f}(\alpha)=\epsilon\xi_p^{f^*(x)}$, where $\epsilon\in \{\pm 1,\pm \sqrt{-1}\}$ is independent of $\alpha$, otherwise it is called \emph{non-weakly  regular}. Specially, if $\epsilon=1$, $f(x)$ is called \emph{regular}. The following definition is given in \cite{Ozbudak1}.
	
	\begin{definition}
		Let $S$ be a subset of $V_n^{(p)}$ with $\#S=N$ and $f(x)$ be a function from $S$ to $\mathbb{F}_p$. If $|\widehat{\chi_f}(\alpha)|=N^{\frac{1}{2} }$ for all $\alpha\in V_n^{(p)}$, then $f(x)$ is called bent relative to $S$, where $\widehat{\chi_f}(\alpha)=\sum_{x\in S}\xi_p^{f(x)-\langle \alpha,x\rangle_n}.$
	\end{definition}
	\begin{remark}\label{Re 1}
		$\mathrm{(1)}$ For an $s$-plateaued function $f(x)$ from $V_n^{(p)}$ to $\mathbb{F}_p$, if its dual $f^*(x)$ is bent relative to $\mathrm{Supp}(\widehat{\chi_f})$, then for any $\alpha\in V_n^{(p)}$, the value of $\widehat{\chi_{f^*}}(\alpha)$  is given by
		
		\[\widehat{\chi_{f^*}}(\alpha)=
		\begin{cases}
			\pm p^{\frac{n-s}{2}}\xi_{p}^{f^{**}(\alpha)}, & \text{if } p^{n-s}\equiv 1\ (\mathrm{mod}\ 4);\\ 
			\pm \sqrt{-1}p^{\frac{n-s}{2}}\xi_{p}^{f^{**}(\alpha)}, & \text{ if }  p^{n-s}\equiv 3\ (\mathrm{mod}\ 4),
		\end{cases}\]
		where $f^{**}(x)$ from $V_n^{(p)}$ to $\mathbb{F}_p$ is the dual of $f^*(x)$. Similarly, the dual $f^*(x)$ is called weakly regular bent relative to $\mathrm{Supp}(\widehat{\chi_f})$, if for all $\alpha\in V_n^{(p)}$, $p^{-\frac{n-s}{2}}\widehat{\chi_{f^*}}(\alpha)=\epsilon\xi_p^{f^{**}(\alpha)}$, where $\epsilon\in\{\pm1,\pm\sqrt{-1}\}$ is independent of $
		\alpha$, otherwise it is called non-weakly regular bent relative to $\mathrm{Supp}(\widehat{\chi_f})$. Specially, if $\epsilon=1$, $f^*(x)$ is called regular bent relative to $\mathrm{Supp}(\widehat{\chi_f})$.\\
		$\mathrm{(2)}$	In \cite{Ozbudak1}, {\"O}zbudak et al. 
		proved that the dual $f^*(x)$ of a weakly regular $s$-plateaued function $f(x)$ is weakly regular bent relative to $\mathrm{Supp}(\widehat{\chi_f})$ and $f^{**}(x)=f(-x)$. Moreover, they also proved that if $f(x)$ is a non-weakly regular $s$-plateaued function such that its dual $f^*(x)$ is  bent relative to $\mathrm{Supp}(\widehat{\chi_f})$, then $f^*(x)$ is non-weakly regular bent relative to  $\mathrm{Supp}(\widehat{\chi_f})$ and satisfies $f^{**}(x)=f(-x)$.
	\end{remark}
	For a function $f(x)$ from $V_n^{(p)}$ to $\mathbb{F}_p$, 
	define $D_{f,i}=\{x\in V_n^{(p)}:f(x)=i\}$ for any $i\in\mathbb{F}_p$, $D_{f,sq}=\{x\in V_n^{(p)}:f(x)\in SQ\}$ and $D_{f,nsq}=\{x\in V_n^{(p)}:f(x)\in NSQ\}$.
	Let $\mu=1$ if $p^{n+s}\equiv1$ (mod $4$) and  $\mu=\sqrt{-1}$ if $p^{n+s}\equiv3$ (mod $4$). For an $s$-plateaued function $f(x):V_n^{(p)}\longrightarrow\mathbb{F}_p$, we define $B_+(f)\ \text{and}\ B_{-}(f)$ as follows.
	 \begin{align*}
		B_+(f)&=\{\alpha\in \mathrm{Supp}(\widehat{\chi_f}):\ \widehat{\chi_f}(\alpha)=\mu p^{\frac{n+s}{2}}\xi_p^{f^*(\alpha)}\},\\
		B_-(f)&=\{\alpha\in \mathrm{Supp}(\widehat{\chi_f}):\ \widehat{\chi_f}(\alpha)=-\mu p^{\frac{n+s}{2}}\xi_p^{f^*(\alpha)}\}.
	\end{align*}
	For any $\alpha\in\mathrm{Supp}(\widehat{\chi_f})$, define $\epsilon_{\alpha}=1$ (respectively $-1$) if $\alpha\in B_+(f)$ (respectively $B_-(f)$). If $f(x)$ is unbalanced, we define the type of $f(x)$ as 
	$f(x)\  \text{is of}\  type\  (+) \ \text{if\  }\  \epsilon_0=1\ \text{and of}\  type\ (-)\ \text{if\ }\\  \epsilon_0=-1.$ Besides, if the dual $f^*(x)$ of $f(x)$ is bent relative to $\mathrm{Supp}(\widehat{\chi_f})$, then we define $B_+(f^*)$ and $B_-(f^*)$ as follows.
	
	\begin{align*}
		B_+(f^*)&=\{\alpha\in V_n^{(p)}:\ \widehat{\chi_{f^*}}(\alpha)=\mu p^{\frac{n-s}{2}}\xi_p^{f(-\alpha)}\},\\
		B_-(f^*)&=\{\alpha\in V_n^{(p)}:\ \widehat{\chi_{f^*}}(\alpha)=-\mu p^{\frac{n-s}{2}}\xi_p^{f(-\alpha)}\}.
	\end{align*}
	Meanwhile, for any $\alpha\in V_{n}^{(p)}$, define $\epsilon_\alpha^*=1$ (respectively $-1$) if $\alpha\in B_+(f^*)$ (respectively $B_-(f^*)$). Define  the type of $f^*(x)$ as 
	$f^*(x)\  \text{is of}\  type\  (+)\  \text{if\  }\  \epsilon_0^*=1\ \text{and of}\  type\ (-)\ \text{if\ }\  \epsilon_0^*=-1.$ 
	\begin{remark}\label{Re 2} 
		By the result in \cite[Proposition 4.1]{Ozbudak1}, for an unbalanced weakly regular $s$-plateaued function $f(x)$ from $V_n^{(p)}$ to $\mathbb{F}_p$, we easily have that for $p^{n+s} \equiv1\ (\mathrm{mod}\ 4)$, the types of $f(x)$ and $f^*(x)$ are the same, and for $p^{n+s}\equiv3\ (\mathrm{mod}\ 4)$, the types of $f(x)$ and $f^*(x)$ are different.
	\end{remark}
	
	The following lemma gives the value distributions of unbalanced $p$-ary $s$-plateaued functions.
	\begin{lemma} [\cite{Ozbudak1}]\label{Le 3} Let $f(x):V_n^{(p)}\longrightarrow\mathbb{F}_p$ be an unbalanced $p$-ary $s$-plateaued function with $f^{*}(0)=j_0$. For any $j\in \mathbb{F}_p$, define $N_j(f)=\#D_{f,j}.$ \\                                                                                                                                                                                                                                                                                                                                                                                                                                                                                                                                                                                                                                                                                                                                                                                                                                                                                                                                                                                                                                                                                                                                                                                                                                                                                                                                                                                                                                                                                                                                                                                                                                                                                                                                                                                                                                                                                                                                                                                                                                                                                                                                                                                                                                                                                                                                                                                                                                                                                                                                                                                                                                                                                                                                                                                                                                                                                                                                                                                                                                                                                                                                                                                                                                                                                                                                                                                                                                                                                                                                                                                                                                                                                                                                                                                                                                                                                                                                                                                                                                                                                                                                                                                                                                                                                                                                                                                                                                                                                                                                                                                                                                                                                                                                                                                                                                                                                                                                                                                                                                                                                                                                                                                                                                                                                                                                                                                                                                                                                                                                                                                                                                                                                                                                                                                                                                                                                                                                                                                                                                                                                                                                                                                                                                                                                                                                                                                                                                                                                                                                                                                                                                                                                                                                                                                                                                                                                                                                                                                                                                                                                                                                                                                                                                                                                                                                                                                                                                                                                                                                                                                                                                                                                                                                                                                                                                                                                                                                                                                                                                                                                                                                                                                                                                                                                                                                                                                                                                                                                                                                                                                                                                                                                                                                                                                                                                                                                                                                                                                                                                                                                                                                                                                                                                                                                                                                                                                                                                           
$\bullet$ When $n+s$ is even, we have $N_{j}(f)=p^{n-1}+\epsilon_0(\delta_0(j-j_0)p-1)p^{\frac{n+s}{2}-1}$ for any  $j\in \mathbb{F}_p$.\\
			$\bullet$ When $n+s$ is odd, we have
			$N_{j}(f)=p^{n-1}+\epsilon_0\eta(j-j_0)p^{\frac{n+s-1}{2}}$ for any $j\in\mathbb{F}_p$.
	\end{lemma}	
	When the dual $f^*(x)$ of an $s$-plateaued function $f(x)$ is bent relative to $\mathrm{Supp}(\widehat{\chi_f})$, we give the value distributions of $f^*(x)$ in the following lemma.
	\begin{lemma}\label{Le 4}
		Let $f(x):V_n^{(p)}\longrightarrow\mathbb{F}_p$ be an $s$-plateaued function whose dual $f^*(x)$ is bent relative to $\mathrm{Supp}(\widehat{\chi_f})$ and $f(0)=j_0$. For any $j \in \mathbb{F}_p$, define $N_j(f^*)=\#\{x\in \mathrm{Supp}(\widehat{\chi_f}):\ f^*(x)=j\}.$   \\ 
		$\bullet$ When $n+s$ is even, we have $N_{j}(f^*)=p^{n-s-1}+\epsilon_0^*(\delta_0(j-j_0)p-1) p^{\frac{n-s}{2}-1}$ for any $j\in \mathbb{F}_p$.\\
			$\bullet$ When $n+s$ is odd, we have
			$N_{j}(f^*)=p^{n-s-1}+\epsilon_0^*\eta(j-j_0)p^{\frac{n-s-1}{2}}$ for any $j\in\mathbb{F}_p$.
	\end{lemma}	
	\begin{proof} For any $j\in\mathbb{F}_p$, we have
		{\small 	\begin{align*}
				N_j(f^*)&=p^{-1}\sum\limits_{x\in\mathrm{Supp}(\widehat{\chi_f})}\sum\limits_{y\in\mathbb{F}_p}\xi_p^{y(f^*(x)-j)}\\
				&=p^{-1}\sum\limits_{y\in\mathbb{F}_p^{*}}\sum\limits_{x\in\mathrm{Supp}(\widehat{\chi_f})}\xi_p^{y(f^*(x)-j)}+p^{n-s-1}\\
				&=p^{-1}\sum\limits_{y\in\mathbb{F}_p^{*}}\sigma_y(\widehat{\chi_{f^*}}(0))\xi_p^{-yj}+p^{n-s-1}.
		\end{align*}}
		Since $f^*(x)$ is bent relative to $\mathrm{Supp}(\widehat{\chi_f})$ and $f^{**}(0)=f(0)=j_0$, then by Lemma \ref{Le 1}, we have that if $n-s$ is even, then for any $j\in\mathbb{F}_p$, 
		{\small 	\[
			N_j(f^*)=\sum\limits_{y\in\mathbb{F}_p^{*}}\epsilon_0^*p^{\frac{n-s}{2}-1}\xi_p^{y(j_0-j)}+p^{n-s-1}
			=\epsilon_0^*(\delta_0(j-j_0)p-1)p^{\frac{n-s}{2}-1}+p^{n-s-1}.\]}
		If $n-s$ is odd, then for any $j\in\mathbb{F}_p$, 
		
		{\small \[N_{j}(f^*)=\sum\limits_{y\in\mathbb{F}_p^{*}}\epsilon_0^*p^{\frac{n-s-3}{2}}\sqrt{p^*}\eta(y)\xi_p^{-y(j-j_0)}+p^{n-s-1}\\
			=\epsilon_0^*\eta(j-j_0)p^{\frac{n-s-1}{2}}+p^{n-s-1}.\]}
		The proof is now completed.
	\end{proof}
	
	From the results of Theorem 4.2 in \cite{Ozbudak1}, we have the following lemma.
 	\begin{lemma}\label{Le 5} Let $f(x):V_n^{(p)}\longrightarrow\mathbb{F}_p$ be an $s$-plateaued function whose dual $f^*(x)$ is bent relative to $\mathrm{Supp}(\widehat{\chi_f})$ and $f(0)=j_0$, $\#B_{+}(f)=k$ with $0\le k\le p^{n-s}$. For any $j \in \mathbb{F}_p$, define $c_j(f^*)=\#\{x\in B_+(f):\ f^*(x)=j\}$ and $d_j(f^*)=\#\{x\in B_-(f):\ f^*(x)=j\}$ .  \\
		$\bullet$
			If $n+s$ is even, then for any $j\in\mathbb{F}_p$, 
			$c_j(f^*)=\frac{k}{p}+\frac{\epsilon_0^*+1}{2}(\delta_0(j-j_0)p-1)p^{\frac{n-s}{2}-1}$ and $
			d_j(f^*)=p^{n-s-1}-\frac{k}{p}+\frac{\epsilon_0^*-1}{2}(\delta_0(j-j_0)p-1)p^{\frac{n-s}{2}-1}.$\\
			$\bullet$ If $n+s$ is odd, then for any $j\in\mathbb{F}_p$, $
			c_j(f^*)=\frac{k}{p}+\frac{\epsilon_0^*+\eta(-1)}{2}\eta(j-j_0)p^{\frac{n-s-1}{2}}$ and $
			d_j(f^*)=p^{n-s-1}-\frac{k}{p}+\frac{\epsilon_0^*-\eta(-1)}{2}\eta(j-j_0)p^{\frac{n-s-1}{2}}.$
	\end{lemma}
\begin{proof}
	If $f(x)$ is weakly regular, then by Remarks \ref{Re 1}, \ref{Re 2} and Lemma \ref{Le 4}, we easily get the results. 
	
	If $f(x)$ is non-weakly regular, then for any $j \in \mathbb{F}_p$, according to the definitions of $c_j(f^*)$ and  $d_j(f^*)$, we have that $c_j(f^*)+d_j(f^*)=N_j(f^*)$.
	From the Theorem 4.2 in \cite{Ozbudak1}, we know that when $n+s$ is even,
	$t_j(f^*)=\frac{2k}{p}-p^{n-s-1}+(\delta_0(j-j_0)p-1)p^{\frac{n-s}{2}-1}$ for any $j\in\mathbb{F}_p$; when $p\equiv1\ \mathrm{(mod\  4)}$ and $n+s$ is odd,  $t_{j}(f^*)=\frac{2k}{p}-p^{n-s-1}+\eta
	(j-j_0)p^{\frac{n-s-1}{2}}$ for any $j\in\mathbb{F}_p$; when  $p\equiv3\ \mathrm{(mod\  4)}$ and $n+s$ is odd,  $t_{j}(f^*)=\frac{2k}{p}-p^{n-s-1}-\eta
	(j-j_0)p^{\frac{n-s-1}{2}}$ for any $j\in\mathbb{F}_p$. By Lemma \ref{Le 4}, the proof can be easily completed.
\end{proof}

	Let $\mathscr{F}$ be the set of $p$-ary $s$-plateaued functions satisfying the following conditions:\\
	$(1)$ $f(0)=0$;\\
	$(2)$ $f^*(x)$ is bent relative to $\mathrm{Supp}(\widehat{\chi_f})$;\\
	$(3)$ For any $x\in V_n^{(p)}$, $a\in\mathbb{F}_p^*$, if $x\in B_+(f^*)$ (respectively $B_-(f^*)$), then $ax\in B_+(f^*)$ (respectively $B_-(f^*)$).\\
	$(4)$ There exists a positive integer $t$ with $2\le t\le p-1$ and $\mathrm{gcd}(t-1,p-1)=1$ such that $f(ax)=a^tf(x)$ for any $a\in\mathbb{F}_p^*$ and $x\in B_+(f^*)$, and there exists a positive integer $t'$ with $2\le t'\le p-1$ and $\mathrm{gcd}(t'-1,p-1)=1$ such that $f(ax)=a^{t'}f(x)$  for any $a\in\mathbb{F}_p^*$ and $x\in B_-(f^*)$.
	\begin{remark}\label{Re 3} 	For an $s$-plateaued function $f(x)$ from $V_n^{(p)}$ to $\mathbb{F}_p$ belonging to $\mathscr{F}$, it is easy to see that  $f(x)=f(-x)$ for any $x\in V_n^{(p)}$. Since $f(0)=0$, then $f(x)$ is unbalanced, i.e., $0\in \mathrm{Supp}(\widehat{\chi_f})$. When $f(x)$ is a non-weakly regular plateaued function, if $t=t'$, then $f(x)$ is a $t$-form, otherwise there does not exist a positive integer $l$ such that $f(x)$ is an $l$-form.
	\end{remark}
	
	In the following, we present a construction of plateaued functions that belong to $\mathscr{F}$. Let $f^{(z)}$ be a weakly regular bent function from $V_{n_1}^{(p)}$ to $\mathbb{F}_p$ for any $z\in\mathbb{F}_{p^{n_2}}$. We consider the variant of generalized Maiorana-McFarland (GMMF  \cite{Cesmelioglu}) function  $F:V_{n_1}^{(p)}\times \mathbb{F}_{p^{n_2}}\times\mathbb{F}_{p^{n_2}}\longrightarrow\mathbb{F}_p$ defined by 
	\begin{equation}
		F(x,y,z)=f^{(z)}(x)+Tr_1^{n_2}(yz^{l-1}),
	\end{equation}
	where $l$ is a positive integers and $\mathrm{gcd}(l-1,p^{n_2}-1)=1$.
	Define 
	\begin{align*}
		W_+(F)&=\{z\in\mathbb{F}_{p^{n_2}}: f^{(z^e)}\ \text{is\ of type}(+)\},\ 
		W_-(F)=\{z\in\mathbb{F}_{p^{n_2}}: f^{(z^e)}\ \text{is\ of type}(-)\},\\
		W_+(F^*)&=\{z\in\mathbb{F}_{p^{n_2}}: f^{{(-z)}^*}\ \text{is\ of type}(+)\},\ 
		W_-(F^*)=\{z\in\mathbb{F}_{p^{n_2}}: f^{{(-z)}^*}\ \text{is\ of type}(-)\},
	\end{align*}
	where $e$ is a positive integer and $e(l-1)\equiv 1{(\mathrm{mod}\  p^{n_2}-1)}$. By the definition of Walsh transform and Remark \ref{Re 2}, we easily know that 
	$F^*(x,y,z)=f^{(y^e)^*}(x)-Tr_{1}^{n_2}(y^ez)$,
	and 
	\begin{equation*}
		B_{\pm}(F)=V_{n_1}^{(p)}\times W_{\pm}(F)\times\mathbb{F}_{p^{n_2}},\ 
		B_{\pm}(F^*)=V_{n_1}^{(p)}\times\mathbb{F}_{p^{n_2}}\times W_{\pm}(F^*).
	\end{equation*}
	\begin{remark}\label{Re 4}
		Let $f^{(0)}(0)=0$ and $f^{(az)}=f^{(z)}$ for any $a\in\mathbb{F}_p^*$, $z\in\mathbb{F}_{p^{n_2}}^*$. Let the types of $f^{(z)}$ be the same for all $z\in\mathbb{F}_{p^{n_2}}^*$. If the types of $f^{(0)}$ and $f^{(z)}$ are the same,  then $F(x,y,z)$ is weakly regular. If the types of $f^{(0)}$ and $f^{(z)}$ are different, then $F(x,y,z)$ is  non-weakly regular. Let $f^{(0)}$ be of $t$-form ($t'$-form) and $f^{(z)}$ of $t'$-form ($t$-form) for $z\in\mathbb{F}_{p^{n_2}}^*$, where $t'\equiv l\ (\mathrm{mod}\ p-1)$ ($t\equiv l\ (\mathrm{mod}\ p-1)$). By using known weakly regular bent functions, we can construct infinitely many weakly regular and non-weakly regular bent functions belonging to $\mathscr{F}$ by Equation (3). According to \cite[Examples 3, 4]{Ozbudak1}, we can also obtain infinitely many weakly regular and non-weakly regular $s$-plateaued functions belonging to $\mathscr{F}$ with $s\ne 0$.
	\end{remark}

	In the following lemmas, we give some results about  the plateaued functions belonging to $\mathscr{F}$.
	
	Define $S_0(\alpha)=\sum_{x\in B_+(f^*)}\xi_p^{f(x)+\langle \alpha,x\rangle_n}$ and $S_1(\alpha)=\sum_{x\in B_-(f^*)}\xi_p^{f(x)+\langle \alpha,x\rangle_n}$ for any $\alpha\in V_n^{(p)}$. Let $S_0(\alpha)=0$ (respectively $S_1(\alpha)=0$) if $B_+(f^*)=\emptyset$ (respectively $B_-(f^*)=\emptyset$). In the following lemma, we give the values of $S_0(\alpha)$ and $S_1(\alpha)$.
	\begin{lemma}\label{Le 6}
		Let $f(x):V_n^{(p)}\longrightarrow\mathbb{F}_p$ be an $s$-plateaued function belonging to $\mathscr{F}$ and $f^*(x)$ be its dual, then we have the following.\\
		$\bullet$ When $n+s$ is even, for $\alpha\in \mathrm{Supp}(\widehat{\chi_f})$, $S_0(\alpha)=\frac{\epsilon_\alpha+1}{2}p^{\frac{n+s}{2}}\xi_p^{f^*(\alpha)}$, $S_1(\alpha)=\frac{\epsilon_\alpha-1}{2}p^{\frac{n+s}{2}}\xi_p^{f^*(\alpha)}$, and for $\alpha\notin \mathrm{Supp}(\widehat{\chi_f})$, $S_0(\alpha)=S_1(\alpha)=0$.\\
			$\bullet$ When $n+s$ is odd, for $\alpha\in \mathrm{Supp}(\widehat{\chi_f})$, $S_0(\alpha)=\frac{\epsilon_\alpha+\eta(-1)}{2}\sqrt{p^*}p^{\frac{n+s-1}{2}}\xi_p^{f^*(\alpha)}$, $S_1(\alpha)=\frac{\epsilon_\alpha-\eta(-1)}{2}\sqrt{p^*}\\p^{\frac{n+s-1}{2}}\xi_p^{f^*(\alpha)}$, and for $\alpha\notin \mathrm{Supp}(\widehat{\chi_f})$, $S_0(\alpha)=S_1(\alpha)=0$.	
	\end{lemma}
	\begin{proof}Since $f(x)\in\mathscr{F}$, then $f(x)=f(-x)$ for any $x\in V_n^{(p)}$. According to Remark \ref{Re 2}, when $f(x)$ is weakly regular, the results obviously hold. When $f(x)$ is non-weakly regular, by the definition of Walsh transform, we have that 
		{\small \begin{equation}S_0(\alpha)+S_1(\alpha)=\begin{cases}
					0,&\text{if}\ \alpha\notin \mathrm{Supp}(\widehat{\chi_f});\\
					\epsilon_\alpha p^{\frac{n+s}{2}}\xi_p^{f^*(\alpha)},&\text{if}\ \alpha\in \mathrm{Supp}(\widehat{\chi_f})\ \text{and}\ n+s\ \text{is\ even};\\
					\epsilon_\alpha \sqrt{p^*}p^{\frac{n+s-1}{2}}\xi_p^{f^*(\alpha)},&\text{if}\ \alpha\in \mathrm{Supp}(\widehat{\chi_f})\ \text{and}\ n+s\ \text{is\ odd}.
				\end{cases}		
		\end{equation}}
		According to Lemma \ref{Le 1}, we know that 
		{\small \begin{equation}
				\begin{split}
					\sum\limits_{x\in V_n^{(p)}}\widehat{\chi_{f^*}}(x)\xi_p^{\langle\alpha,x\rangle_n}
					&=\sum\limits_{y\in\mathrm{Supp}(\widehat{\chi_f})}\xi_p^{f^*(y)}\sum\limits_{x\in V_n^{(p)}}\xi_p^{\langle \alpha-y, x\rangle_n}\\
					&=\begin{cases}
						0,&\text{if}\ \alpha\notin \mathrm{Supp}(\widehat{\chi_f});\\
						p^{n}\xi_p^{f^*(\alpha)},&\text{if}\ \alpha\in \mathrm{Supp}(\widehat{\chi_f}).
					\end{cases}
				\end{split}
		\end{equation}}
		
		One the other hand, since $f^{**}(x)=f(-x)=f(x)$ for any $x\in V_{n}^{(p)}$, then we have that
		
		{\small \begin{equation}
				\sum\limits_{x\in V_n^{(p)}}\widehat{\chi_{f^*}}(x)\xi_p^{\langle\alpha,x\rangle_n}=\begin{cases}
					\sum\limits_{x\in V_n^{(p)}}\epsilon_x^*p^{\frac{n-s}{2}}\xi_p^{f(x)+\langle\alpha,x\rangle_n},&\text{if}\ n+s\ \text{is\ even};\\
					\sum\limits_{x\in V_n^{(p)}}\epsilon_x^*\sqrt{p^*}p^{\frac{n-s-1}{2}}\xi_p^{f(x)+\langle\alpha,x\rangle_n},&\text{if}\ n+s\ \text{is\ odd}.
				\end{cases}
		\end{equation}}
		
		Hence, by Equations (5) and (6), we get that 
		{\small \begin{equation}S_0(\alpha)-S_1(\alpha)=\begin{cases}
					0,&\hspace{-0.3cm}\text{if}\ \alpha\notin \mathrm{Supp}(\widehat{\chi_f});\\
					p^{\frac{n+s}{2}}\xi_p^{f^*(\alpha)},&\hspace{-0.3cm}\text{if}\ \alpha\in \mathrm{Supp}(\widehat{\chi_f})\ \text{and}\ n+s\ \text{is\ even}; \\ \eta(-1)\sqrt{p^*}p^{\frac{n+s-1}{2}}\xi_p^{f^*(\alpha)},&\hspace{-0.3cm}\text{if}\ \alpha\in \mathrm{Supp}(\widehat{\chi_f})\ \text{and}\ n+s\ \text{is\ odd}.
				\end{cases}		
		\end{equation}}
		By Equations (4) and (7), we easily obtain the values of $S_0(\alpha)$ and $S_1(\alpha)$.\end{proof}
	
		\begin{lemma}\label{Le 8}
		Let $f(x):V_n^{(p)}\longrightarrow\mathbb{F}_p$ be an $s$-plateaued function  and $f^*(x)$ be its dual. If $f(0)=0$ and $f(x)=f(-x)$ for any $x\in V_n^{(p)}$, then $f^*(0)=0$.
	\end{lemma}
	\begin{proof}
		Since $f(0)=0$ and $f(x)=f(-x)$ for any $x\in V_n^{(p)}$, then $N_0(f)$ is odd and $N_j(f)$ is even for any $j\in\mathbb{F}_p^*$. Assume that $f^*(0)=i$, $i\ne 0$, then by Lemma \ref{Le 3}, we have that $N_i(f)$ is odd, which gives a contradiction. Thus, $f^*(0)=0$.
	\end{proof}
	\begin{lemma}\label{Le 7}
		Let $f(x):V_n^{(p)}\longrightarrow\mathbb{F}_p$ be an $s$-plateaued function belonging to $\mathscr{F}$,  $\#B_+(f)=k$ with $0\le k\le p^{n-s}$, and $f^*(x)$ be its dual, then\\
		$(1)$ For any $a\in\mathbb{F}_p^*$ and $\alpha\in V_n^{(p)}$, if $\alpha\notin \mathrm{Supp}(\widehat{\chi_f})$, then $a\alpha\notin \mathrm{Supp}(\widehat{\chi_f})$, and if $\alpha\in B_+(f)$ (respectively $B_-(f)$), then $a\alpha\in B_+(f)$ (respectively $B_-(f)$).\\
		$(2)$ There exists a positive integer $h$ with $\mathrm{gcd}(h-1,p-1)=1$ such that $f^*(a\alpha)=a^hf^*(\alpha)$ for any $a\in\mathbb{F}_p^*$, $\alpha\in B_+(f)$, and there exists a positive integer $h'$ with $\mathrm{gcd}(h'-1,p-1)=1$ such that $f^*(a\alpha)=a^{h'}f^*(\alpha)$ for any $a\in\mathbb{F}_p^*$, $\alpha\in B_-(f)$.\\
		$(3)$ The types of $f(x)$ and $f^*(x)$ are the same if $p^{n+s}\equiv1\ (\mathrm{mod}\ 4)$, and different if $p^{n+s}\equiv3\ (\mathrm{mod}\ 4)$.
	\end{lemma}
	\begin{proof}
		 	Since $\mathrm{gcd}(t-1,p-1)=1$ and $\mathrm{gcd}(t'-1,p-1)=1$, then there exist two positive integers $m$ and $m'$ such that $m(t-1)\equiv 1(\mathrm{mod}\ p-1)$ and $m'(t'-1)\equiv 1(\mathrm{mod}\ p-1)$. Let $h=m+1$ and $h'=m'+1$. For any $a\in\mathbb{F}_p^*$, since $B_+(f^*)=aB_+(f^*)$ and $B_-(f^*)=aB_-(f^*)$, 
		 by the similar discussion in \cite[Lemma 6]{Mesnager3}, we have that for any  $\alpha\in V_n^{(p)}$, $S_0(a\alpha)=\sigma_{a^h}(S_0(\alpha))$ and $S_1(a\alpha)=\sigma_{a^{h'}}(S_1(\alpha))$. According to  Lemma \ref{Le 6}, we easily get the results of (1) and (2).
		 
		  When $f(x)$ is weakly regular, then the result directly follows from Remark \ref{Re 2}. When $f(x)$ is non-weakly regular, we only prove the case when $f(x)$ is of type $(+)$, as the case when $f(x)$ is of type $(-)$ is similar. Note that for any $x\in V_n^{(p)}$, if $x\in B_+(f)$, then $-x\in B_+(f)$, so  $k$ is odd. 
		When $p^{n+s}\equiv1\ (\mathrm{mod}\ 4)$, 
		assume that $f^*(x)$ is of type $(-)$. By Lemma \ref{Le 5}, for any $j\in\mathbb{F}_p^*$, $c_j(f^*)$ and $d_j(f^*)$ are odd. Note that $f(x)=f(-x)$ for any $x\in V_n^{(p)}$, then $f^*(0)=0$. Since  $\epsilon_x=\epsilon_{-x}$, $f^*(x)=f^*(-x)$ for any $x\in V_n^{(p)}$,  then for any $j\in\mathbb{F}_p^*$, $c_j(f^*)$ and $d_j(f^*)$ should be even, which gives a contradiction. Hence, $f^*(x)$ is of type $(+)$. When $p^{n+s}\equiv3\ (\mathrm{mod}\ 4)$, by the similar discussion, we easily have that $f^*(x)$ is of type $(-)$. The proof of $(3)$ is now completed.
	\end{proof}
	\subsection{Linear codes and quantum codes}
A $p$-ary $[n, k, d]$ linear code $\mathcal{C}$ is a $k$-dimensional subspace of $\mathbb{F}_p^n$ with minimum Hamming distance $d$. A \emph{generator matrix} for $\mathcal{C}$ is any $k\times n$ matrix $G$ whose
rows form a basis for $\mathcal{C}$. Let $T$ be a set of $t$ coordinate positions in $\mathcal{C}$. The \emph{punctured\ code} of $\mathcal{C}$, defined by $\widetilde{\mathcal{C}}$, can be obtained by puncturing $\mathcal{C}$ on $T$. The linear code defined by $\mathcal{C}^{\bot}=\{x\in\mathbb{F}^{n}_{p}:\ x\cdot y=0\ \text{for all}\  y\in \mathcal{C}\}$ is called the \emph{dual\ code} of $\mathcal{C}$. 
Let $A_j$ be the number of codewords with Hamming weight $j$ in $\mathcal{C}$ for $0\le j\le n$. Then $(1,\ A_1,\cdots,\ A_n)$ is called the \emph{weight\  distribution} of $\mathcal{C}$ and $1+A_1z+\cdots+A_nz^n$ is the \emph{weight\  enumerator} of $\mathcal{C}$. The \emph{sphere packing bound} on  an $[n, k, d]$ linear code over $\mathbb{F}_p$ is given by 
	\[p^{n}\ge p^k\sum_{j=0}^{\lfloor\frac{d-1}{2}\rfloor}\binom{n}{j}(p-1)^j.\]
	An $[n, k, d]$ linear code is said to be \emph{optimal} if there does not exist $[n, k +1, d]$ or $[n, k, d +1]$ linear code. An $[n, k, d]$ linear code is said to be \emph{almost optimal} if there exists an $[n, k +1, d]$ or $[n, k, d +1]$ optimal linear code.
	
	For a linear code $\mathcal{C}$ over $\mathbb{F}_p$, if $\mathcal{C}\subseteq \mathcal{C}^{\bot}$, then $\mathcal{C}$ is called \emph{self-orthogonal}. In particular, if $\mathcal{C}=\mathcal{C}^{\bot}$,
	then $\mathcal{C}$ is said to be \emph{self-dual}. When $p=3$,  $\mathcal{C}$ is self-orthogonal if and
	only if every codeword of $\mathcal{C}$ has weight divisible by three \cite{Huffman}. For a 
	general $p$-ary linear code $\mathcal{C}$, the following proposition gives a way to judge whether $\mathcal{C}$ is self-orthogonal or not.
	\begin{proposition}[\cite{Wan}]\label{Po 3}
		Let $\mathcal{C}$ be a $p$-ary linear code, then $\mathcal{C}$ is self-orthogonal if and only if $c\cdot c=0$ for all $c\in \mathcal{C}$.
	\end{proposition}
	If $\mathcal{C}\cap\mathcal{C}^{\bot}=\{0\}$, then $\mathcal{C}$ is called a \emph{linear complementary dual code} (LCD code). Obviously, the dual code $\mathcal{C}^{\bot}$ of an LCD code $\mathcal{C}$ is also an LCD code. The following proposition presents a sufficient and necessary condition for a linear code $\mathcal{C}$ to be an LCD code.
	\begin{proposition}[\cite{Massey}]\label{Po 4}
		Let $\mathcal{C}$ be a linear code with a generator matrix $G$ over $\mathbb{F}_p$. Then $\mathcal{C}$ is
		an LCD code if and only if $GG^{T}$ is nonsingular, where $G^{T}$ denotes the transpose of
		$G$.
	\end{proposition}
	
	A $p$-ary $[[n, k, d]]_p$ quantum code $Q$ with length $n$ and minimum distance $d$ is a $K$-dimensional subspace of the $p^n$-dimensional Hilbert space $\mathbb{C}^{p^n}$, where $k=\mathrm{log}_pK$. For any  two codewords $v_1, v_2 \in Q$ and any quantum error operator 
	$E$ on $\mathbb{C}^{p^n}$ with $1\le w_Q(E)\le  d-1$,  if $\langle v_1|E|v_2\rangle=0$, where $w_Q(E)$ is the  quantum weight of $E$ and  
	$\langle v_1|E|v_2\rangle$ is the Hermitian inner product of $|v_1\rangle$ and $E|v_2\rangle$, then $Q$ is called a \emph{pure quantum code}. The \emph{quantum Hamming bound} on an $[[n,k,d]]_p$ pure quantum code
	is given by 
	\[p^{n-k}\ge \sum_{j=0}^{\lfloor\frac{d-1}{2}\rfloor}\binom{n}{j}(p^2-1)^j.\]
	An $[[n, k, d]]_p$ quantum code is said to be \emph{optimal} if there does not exist $[[n, k, d +1]]_p$ quantum code. An $[[n, k, d]]_p$ quantum code is said to be \emph{almost optimal} if there exists  an $[[n, k, d +1]]_p$ optimal quantum code. For more information on quantum codes, the reader is referred to \cite{Feng}. The following proposition gives the generalized Steane's enlargement 
	construction by which the quantum codes can be obtained by self-orthogonal linear codes.
	\begin{proposition}[\cite{Ling}]\label{Po 5}
		Let $\mathcal{C}_1$  and $\mathcal{C}_2$ respectively be $[n,k_1,d_1]$ and $[n,k_2,d_2]$ linear codes over $\mathbb{F}_p$. If $\mathcal{C}_1^{\bot}\subseteq \mathcal{C}_1\subseteq \mathcal{C}_2$ and $k_2\ge k_1+2$, then a pure $p$-ary quantum code of parameters $\left[\left[n,k_1+k_2-n,\mathrm{min}\{d_1,\lceil(1+\frac{1}{p})d_2\rceil \}\right]\right]_p$ can be constructed.
	\end{proposition}
	\section{Auxiliary results}
	In this section, we give some exponential sums related to plateaued functions, which play a key role in constructing self-orthogonal codes.
	\begin{lemma}\label{Le 10}
		Let $f(x):V_n^{(p)}\longrightarrow \mathbb{F}_p$ be an $s$-plateaued function and $f^*(x)$ be its dual. Define 
		$M_{a,b}=\#\{x\in V_n^{(p)}:af(x)+\langle b,x\rangle_n=0\}$ for any $a\in\mathbb{F}_p$, $b\in V_n^{(p)}$,
		then we have the following.
		
		{\small \[M_{a,b}=\begin{cases}
				p^{n},&\text{if}\ a=b=0;\\
				p^{n-1},&\text{if}\ a=0, b\ne 0\ \text{or}\  a\ne 0, -a^{-1}b\notin\mathrm{Supp}(\widehat{\chi_f});\\
				\epsilon_{(-a^{-1}b)}(\delta_0(f^*(-a^{-1}b))p-1)p^{\frac{n+s}{2}-1}+p^{n-1},&\text{if}\ n+s\ \text{is even},\ a\ne 0,-a^{-1}b\in\mathrm{Supp}(\widehat{\chi_f});\\
				\epsilon_{(-a^{-1}b)}\eta(f^*(-a^{-1}b))p^{\frac{n+s-3}{2}}p^{*}+p^{n-1},&\text{if}\ n+s\ \text{is odd},\ a\ne 0,-a^{-1}b\in\mathrm{Supp}(\widehat{\chi_f}).
			\end{cases}\]}
	\end{lemma}
	\begin{proof}
	According to the definition of $M_{a,b}$, if $a=b=0$, then $M_{a,b}=p^n$. If $a=0$ and $b\ne 0$, since $\langle b,x\rangle_n$ is a balanced function, then $M_{a,b}=p^{n-1}$. If $a\ne 0$, then we have that
	{\small \begin{align*}
			M_{a,b}&=p^{-1}\sum\limits_{x\in V_n^{(p)}}\sum\limits_{y\in\mathbb{F}_p}\xi_p^{y(af(x)+\langle b,x \rangle_n)}\\
			&=p^{-1}\sum\limits_{y\in\mathbb{F}_p^*}\sigma_{ya}(\sum\limits_{x\in V_n^{(p)}}\xi_p^{f(x)+\langle a^{-1}b, x\rangle_n})+p^{n-1}\\
			&=p^{-1}\sum\limits_{y\in \mathbb{F}_p^*}\sigma_{ya}(\widehat{\chi_f}(-a^{-1}b))+p^{n-1}.
	\end{align*}}
	Obviously, if $-a^{-1}b\notin\mathrm{Supp}(\widehat{\chi_f})$, then $M_{a,b}=p^{n-1}$. If $-a^{-1}b\in \mathrm{Supp}(\widehat{\chi_f})$ and $n+s$ is even, then by Lemma \ref{Le 1}, we get that
	{\small \[
			M_{a,b}
			=\epsilon_{(-a^{-1}b)}p^{\frac{n+s}{2}-1}\sum\limits_{y\in\mathbb{F}_p^*}\xi_p^{yaf^*(-a^{-1}b)}+p^{n-1}
			=\epsilon_{(-a^{-1}b)}p^{\frac{n+s}{2}-1}(\delta_0(f^*(-a^{-1}b))p-1)+p^{n-1}.
	\]}
	If $-a^{-1}b\in \mathrm{Supp}(\widehat{\chi_f})$ and $n+s$ is odd, then by Lemma \ref{Le 1}, we get that
	{\small  \[
			M_{a,b}
			=\epsilon_{(-a^{-1}b)}p^{\frac{n+s-3}{2}}\sqrt{p^*}\sum\limits_{y\in\mathbb{F}_p^*}\eta(ya)\xi_p^{yaf^*(-a^{-1}b)}+p^{n-1}
			=\epsilon_{(-a^{-1}b)}p^{\frac{n+s-3}{2}}p^*\eta(f^*(-a^{-1}b))+p^{n-1}.
\]}
	This completes the proof.\end{proof}
	\begin{lemma}\label{Le 11}
		Let $f(x):V_n^{(p)}\longrightarrow \mathbb{F}_p$ be an $s$-plateaued function belonging to $\mathscr{F}$ and $f^*(x)$ be its dual. Define
		$N_{(0,a)}=\#\{x\in D_{f,0}: \langle a,x\rangle_n=0\}$ for any $a\in V_n^{(p)}$,
		then we have the following.
		\begin{itemize}
			\item [$\bullet$] If $n+s$ is even, then
			{\small \[N_{(0,a)}=\begin{cases}
					p^{n-1}+\epsilon_0(p-1)p^{\frac{n+s}{2}-1},&\text{if}\ a=0;\\
					\epsilon_0(p-1)p^{\frac{n+s}{2}-2}+p^{n-2},&\text{if}\ a\ne 0, a\notin \mathrm{Supp}(\widehat{\chi_f});\\
					(p-1)(\epsilon_a(\delta_0(f^*(a))p-1)+\epsilon_0)p^{\frac{n+s}{2}-2}+p^{n-2},&\text{if}\ a\ne 0, a\in \mathrm{Supp}(\widehat{\chi_f}).
				\end{cases}\]}
			\item[$\bullet$] If $n+s$ is odd, then 
			{\small \[N_{(0,a)}=\begin{cases}
					p^{n-1},&\text{if}\ a=0;\\
					p^{n-2},&\text{if}\ a\ne 0, a\notin\mathrm{Supp}(\widehat{\chi_f});\\
					\epsilon_a(p-1)p^{\frac{n+s-5}{2}}p^*\eta{(f^*(a))}+p^{n-2},&\text{if}\ a\ne 0, a\in \mathrm{Supp}(\widehat{\chi_f}).
				\end{cases}\]}
		\end{itemize}
	\end{lemma}
	\begin{proof}
	According to the definition of $N_{(0,a)}$, if $a=0$, then $N_{(0,a)}=N_0(f)$. By Lemma \ref{Le 3}, we easily obtain the value of $N_{(0,a)}$. If $a\ne 0$, then we have that
	{\small \begin{align*}
			N_{(0,a)}&=p^{-2}\sum\limits_{x\in V_n^{(p)}}\sum\limits_{y\in \mathbb{F}_p}\sum\limits_{z\in\mathbb{F}_p}\xi_p^{yf(x)-z\langle a, x\rangle_n}\\
			&=p^{-2}(\sum\limits_{z\in\mathbb{F}_p^*}\sum\limits_{y\in\mathbb{F}_p^*}\sum\limits_{x\in V_n^{(p)}}\xi_p^{yf(x)-\langle za,x\rangle_n}+\sum\limits_{y\in\mathbb{F}_p^*}\sum\limits_{x\in V_n^{(p)}}\xi_p^{yf(x)}+\sum\limits_{z\in\mathbb{F}_p^*}\sum\limits_{x\in V_n^{(p)}}\xi_p^{-\langle za,x\rangle_n})+p^{n-2}.
	\end{align*}}

	Let $N_1=\sum_{z\in\mathbb{F}_p^*}\sum_{y\in\mathbb{F}_p^*}\sum_{x\in V_n^{(p)}}\xi_p^{yf(x)-\langle za,x\rangle_n}$, $N_2=\sum_{y\in\mathbb{F}_p^*}\sum_{x\in V_n^{(p)}}\xi_p^{yf(x)}$ and $N_3=\sum_{z\in\mathbb{F}_p^*}\\\sum_{x\in V_n^{(p)}}\xi_p^{-\langle za,x\rangle_n}$, then $N_{(0,a)}=p^{-2}(N_1+N_2+N_3)+p^{n-2}$. In the following, we give the values of $N_1$, $N_2 $ and $N_3$.
	
	Note that $N_1=\sum_{z\in\mathbb{F}_p^*}\sum_{y\in\mathbb{F}_p^*}\sigma_y(\widehat{\chi_f}(y^{-1}za)).$
	According to Lemma \ref{Le 7}, we know that for any $y,z\in \mathbb{F}_p^*$, if $a\notin\mathrm{Supp}(\widehat{\chi_f})$, then $y^{-1}za\notin\mathrm{Supp}(\widehat{\chi_f})$, and if  $a\in\mathrm{Supp}(\widehat{\chi_f})$, then $\epsilon_a=\epsilon _{y^{-1}za}$ and $f^*(y^{-1}za)=y^{-h_a}z^{h_a}f^*(a)$, where $h_a=h$ if $a\in B_+(f)$ and $h_a=h'$ if $a\in B_-(f)$. Thus, if $a\notin\mathrm{Supp}(\widehat{\chi_f})$, then $N_1=0$. If $a\in \mathrm{Supp}(\widehat{\chi_f})$ and $n+s$ is even, then we have that
	{\small \[N_1=\sum\limits_{z\in\mathbb{F}_p^*}\sum\limits_{y\in\mathbb{F}_p^*}\sigma_y(\epsilon_ap^{\frac{n+s}{2}}\xi_p^{y^{-h_a}z^{h_a}f^*(a)})=\epsilon_ap^{\frac{n+s}{2}}\sum\limits_{z\in\mathbb{F}_p^*}\sum\limits_{y\in\mathbb{F}_p^*}\xi_p^{y^{1-h_a}z^{h_a}f^*(a)}.\]}
	Since $\mathrm{gcd}(h_a-1,p-1)=1$, then when $y$ runs through $\mathbb{F}_p^*$, we know that $y^{1-h_a}$ runs through $\mathbb{F}_p^*$. By Lemma \ref{Le 1}, we easily get $N_1=\epsilon_ap^{\frac{n+s}{2}}(p-1)(\delta_0(f^*(a))p-1)$. If $a\in \mathrm{Supp}(\widehat{\chi_f})$ and $n+s$ is odd, then we have that 
	{\small \[
			N_1=\sum\limits_{z\in\mathbb{F}_p^*}\sum\limits_{y\in\mathbb{F}_p^*}\sigma_y(\epsilon_ap^{\frac{n+s-1}{2}}\sqrt{p^*}\xi_p^{y^{-h_a}z^{h_a}f^*(a)})
			=\epsilon_ap^{\frac{n+s-1}{2}}\sqrt{p^*}\sum\limits_{z\in\mathbb{F}_p^*}\sum\limits_{y\in\mathbb{F}_p^*}\eta(y)\xi_p^{y^{1-h_a}z^{h_a}f^*(a)}.
	\]}
	Since $h_a$ is even, then $\eta(y)=\eta(y^{1-h_a})$ and $\eta(z^{h_a})=1$ for any $y, z\in\mathbb{F}_p^*$. By Lemma \ref{Le 1}, we have that $N_1=\epsilon_a(p-1)p^{\frac{n+s-1}{2}}p^*\eta(f^*(a))$.
	
	For the value of $N_2$, we have that 
	$N_2=\sum_{y\in\mathbb{F}_p^*}\sigma_y(\widehat{\chi_f}(0)).$ Since $f(0)=0$ and $f(x)=f(-x)$
	
	\noindent for any $x\in V_n^{(p)}$, then $f^*(0)=0$. Note that 
 $0\in\mathrm{Supp}(\widehat{\chi_f})$, then by Lemma \ref{Le 1}, when $n+s$ is even, 
	$N_2=\epsilon_0(p-1)p^{\frac{n+s}{2}}$;
	when $n+s$ is odd, $N_2=0$.
	
	Finally, by Lemma \ref{Le 1}, we easily get that if $a\ne 0$, then $N_3=0$.
	The proof is now completed. \end{proof}
	\begin{lemma}\label{Le 12}
		Let $f(x):V_n^{(p)}\longrightarrow \mathbb{F}_p$ be an $s$-plateaued function belonging to $\mathscr{F}$ and $f^*(x)$ be its dual. Define
		$N_{(sq,a)}=\#\{x\in D_{f,sq}: \langle a,x\rangle_n=0\}$ and $N_{(nsq,a)}=\#\{x\in D_{f,nsq}: \langle a,x\rangle_n=0\}$ for any $a\in V_n^{(p)}$,
		then we have the following.
		\begin{itemize}
			\item [$\bullet$] If $n+s$ is even, then
			{\small \[N_{(sq,a)}=\begin{cases}
					\frac{(p-1)}{2}(p^{n-1}-\epsilon_0p^{\frac{n+s}{2}-1}),&\text{if}\ a=0;\\
					\frac{(p-1)}{2}(p^{n-2}-\epsilon_0p^{\frac{n+s}{2}-2}),&\text{if}\ a\ne 0, a\notin\mathrm{Supp}(\widehat{\chi_f});\\
					\frac{(p-1)}{2}(p^{n-2}+(\epsilon_a(p(\eta(f^*(a))-\delta_0(f^*(a)))+1)-\epsilon_0)p^{\frac{n+s}{2}-2}),&\text{if}\ a\ne 0, a\in \mathrm{Supp}(\widehat{\chi_f}).
					
				\end{cases}
				\]}
			{\small \[N_{(nsq,a)}=\begin{cases}
					\frac{(p-1)}{2}(p^{n-1}-\epsilon_0p^{\frac{n+s}{2}-1}),&\text{if}\ a=0;\\
					\frac{(p-1)}{2}(p^{n-2}-\epsilon_0p^{\frac{n+s}{2}-2}),&\text{if}\ a\ne 0, a\notin\mathrm{Supp}(\widehat{\chi_f});\\
					\frac{(p-1)}{2}(p^{n-2}-(\epsilon_a(p(\eta(f^*(a))+\delta_0(f^*(a)))-1)+\epsilon_0)p^{\frac{n+s}{2}-2}),&\text{if}\ a\ne 0, a\in \mathrm{Supp}(\widehat{\chi_f}).
				\end{cases}
				\]}
			\item[$\bullet$] If $n+s$ is odd, then
			{\small \[N_{(sq,a)}=\begin{cases}
					\frac{(p-1)}{2}(p^{n-1}+\epsilon_0p^{\frac{n+s-1}{2}}),&\text{if}\ a=0;\\
					\frac{(p-1)}{2}(p^{n-2}+\epsilon_0p^{\frac{n+s-3}{2}}),&\text{if}\ a\ne 0, a\notin \mathrm{Supp}(\widehat{\chi_f});\\
					\frac{(p-1)}{2}(p^{n-2}+(\epsilon_a(p\delta_0(f^*(a))-\eta(-f^*(a))-1)+\epsilon_0)p^{\frac{n+s-3}{2}}),&\text{if}\ a\ne 0, a\in \mathrm{Supp}(\widehat{\chi_f}).
				\end{cases}
				\]}
			{\small \[N_{(nsq,a)}=\begin{cases}
					\frac{(p-1)}{2}(p^{n-1}-\epsilon_0p^{\frac{n+s-1}{2}}),&\text{if}\ a=0;\\
					\frac{(p-1)}{2}(p^{n-2}-\epsilon_0p^{\frac{n+s-3}{2}}),&\text{if}\ a\ne 0, a\notin \mathrm{Supp}(\widehat{\chi_f});\\
					\frac{(p-1)}{2}(p^{n-2}-(\epsilon_a(p\delta_0(f^*(a))+\eta(-f^*(a))-1)+\epsilon_0)p^{\frac{n+s-3}{2}}),&\text{if}\ a\ne 0, a\in \mathrm{Supp}(\widehat{\chi_f}).
				\end{cases}
				\]}
		\end{itemize}
	\end{lemma}
	\begin{proof} We only give the proof for $N_{(sq,a)}$, the case for $N_{(nsq,a)}$ is similar.
	
	According to the definition of $N_{(sq,a)}$, if $a=0$, then $N_{(sq,a)}=\sum_{i\in SQ}N_i(f)$. By Lemma \ref{Le 3}, we easily obtain the value of $N_{(sq,a)}$. If $a\ne 0$, then we have 
	{\small \begin{align*}
			N_{(sq,a)}&=p^{-2}\sum\limits_{i\in SQ}\sum\limits_{x\in V_n^{(p)}}\sum\limits_{y\in\mathbb{F}_p}\sum\limits_{z\in \mathbb{F}_p}\xi_p^{y(f(x)-i)-z\langle a, x\rangle_n}\\
			&=p^{-2}(\sum\limits_{z\in\mathbb{F}_p^*}\sum\limits_{y\in\mathbb{F}_p^*}\sum\limits_{x\in V_n^{(p)}}\xi_p^{yf(x)-\langle za,x\rangle_n}\sum\limits_{i\in SQ}\xi_p^{-yi}+\sum\limits_{y\in\mathbb{F}_p^*}\sum\limits_{x\in V_n^{(p)}}\xi_p^{yf(x)}\sum\limits_{i\in SQ}\xi_p^{-yi}+\sum\limits_{i\in SQ}\sum\limits_{z\in\mathbb{F}_p^*}\sum\limits_{x\in V_n^{(p)}}\xi_p^{-\langle za, x\rangle_n})\\&+\frac{(p-1)}{2}p^{n-2}.
	\end{align*}}
	Let $N_1=\sum_{z\in\mathbb{F}_p^*}\sum_{y\in\mathbb{F}_p^*}\sum_{x\in V_n^{(p)}}\xi_p^{yf(x)-\langle za,x\rangle_n}\sum_{i\in SQ}\xi_p^{-yi}$, $N_2=\sum_{y\in\mathbb{F}_p^*}\sum_{x\in V_n^{(p)}}\xi_p^{yf(x)}\sum_{i\in SQ}\xi_p^{-yi}$ 
	 and $N_3=\sum_{i\in SQ}\sum_{z\in\mathbb{F}_p^*}\sum_{x\in V_n^{(p)}}\xi_p^{-\langle za, x\rangle_n}$, then $N_{(sq,a)}=p^{-2}(N_1+N_2+N_3)+\frac{p-1}{2}p^{n-2}$. In the following, we give the values of $N_1$, $N_2$ and $N_3$.
	
	By Lemma \ref{Le 1}, we have that $N_1=\sum_{z\in\mathbb{F}_p^*}\sum_{y\in\mathbb{F}_p^*}\sigma_y(\widehat{\chi_f}(y^{-1}za))\frac{\eta(-y)\sqrt{p^*}-1}{2}.$
	According to Lemma \ref{Le 7}, we know that for any $y,z\in \mathbb{F}_p^*$, if $a\notin\mathrm{Supp}(\widehat{\chi_f})$, then $y^{-1}za\notin\mathrm{Supp}(\widehat{\chi_f})$, and if  $a\in\mathrm{Supp}(\widehat{\chi_f})$, then $\epsilon_a=\epsilon _{y^{-1}za}$ and $f^*(y^{-1}za)=y^{-h_a}z^{h_a}f^*(a)$, where $h_a=h$ if $a\in B_+(f)$ and $h_a=h'$ if $a\in B_-(f)$. Thus, if $a\notin\mathrm{Supp}(\widehat{\chi_f})$, then $N_1=0$. If $a\in \mathrm{Supp}(\widehat{\chi_f})$ and $n+s$ is even, then we have that 
	$N_1=\epsilon_ap^{\frac{n+s}{2}}\sum_{z\in\mathbb{F}_p^*}\sum_{y\in\mathbb{F}_p^*}\xi_p^{y^{1-h_a}z^{h_a}f^*(a)}\frac{\eta(-y)\sqrt{p^*}-1}{2}.$
	Note that $\mathrm{gcd}(h_a-1,p-1)=1$ and $h_a$ is even, then $\eta(y)=\eta(y^{1-h_a})$, $\eta(z^{h_a})=1$, and when 
	$y$ runs 
	through $\mathbb{F}_p^*$, 
	$y^{1-h_a}$ runs through $\mathbb{F}_p^*$ for any $y,z\in \mathbb{F}_p^*$. By Lemma \ref{Le 1}, we get that $N_1=\epsilon_a\frac{(p-1)}{2}p^{\frac{n+s}{2}}(p(\eta(f^*(a))-\delta_0(f^*(a)))+1)$. If $a\in \mathrm{Supp}(\widehat{\chi_f})$ and $n+s$ is odd, then we have that 
	$N_1
	=\epsilon_ap^{\frac{n+s-1}{2}}\sqrt{p^*}\sum_{z\in\mathbb{F}_p^*}\sum_{y\in\mathbb{F}_p^*}\eta(y)\xi_p^{y^{1-h_a}z^{h_a}f^*(a)}\frac{\eta(-y)\sqrt{p^*}-1}{2}.$
	Again by Lemma \ref{Le 1}, we get that
	$N_1=\epsilon_a\frac{(p-1)}{2}p^{\frac{n+s+1}{2}}(p\delta_0(f^*(a))-1-\eta(-f^*(a)))$.
	
	For the value of $N_2$, we have that 
	$N_2=\sum_{y\in\mathbb{F}_p^*}\sigma_y(\widehat{\chi_f}(0))\frac{\eta(-y)\sqrt{p^*}-1}{2}.$
	Since $f(0)=0$ and $f(x)=f(-x)$ for any $x\in V_n^{(p)}$, then $f^*(0)=0$. Note that
	$0\in\mathrm{Supp}(\widehat{\chi_f})$, then by Lemma \ref{Le 1}, when $n+s$ is even, $N_2=-\epsilon_0\frac{(p-1)}{2}p^{\frac{n+s}{2}}$; when $n+s$ is odd, $N_2=\epsilon_0\frac{(p-1)}{2}p^{\frac{n+s+1}{2}}$.
	
	Finally, by Lemma \ref{Le 1}, we easily get that
	if $a\ne 0$, then $N_3=0$. 
	The proof is now completed.\end{proof}
	
	\begin{lemma}\label{Le 13}
		Let $f(x):V_n^{(p)}\longrightarrow \mathbb{F}_p$ be an $s$-plateaued function belonging to $\mathscr{F}$ with $t=t'=2$ and $f^*(x)$ be its dual. Define
		$N_{(i,a,b)}=\#\{x\in D_{f,i}: \langle a,x\rangle_n=b\}$  for any $i\in \mathbb{F}_p^*$, $a\in V_n^{(p)}$ and $b\in\mathbb{F}_p^*$,
		then we have the following.
		\begin{itemize}
			\item [$\bullet$] If $n+s$ is even, then\\
			{\small \[N_{(i,a,b)}=\begin{cases}
					0,&\text{if}\ a=0;\\
					p^{n-2}-\epsilon_0p^{\frac{n+s}{2}-2},&\text{if}\ a\ne 0, a\notin\mathrm{Supp}(\widehat{\chi_f});\\
					(\epsilon_a-\epsilon_0)p^{\frac{n+s}{2}-2}+p^{n-2},&\text{if}\ a\ne 0, a\in \mathrm{Supp}(\widehat{\chi_f}), f^*(a)=0;\\
					(\epsilon_ap\eta(b^24^{-1}+if^*(a))+(\epsilon_a-\epsilon_0))p^{\frac{n+s}{2}-2}+p^{n-2}, &\text{if}\ a\ne 0, a\in \mathrm{Supp}(\widehat{\chi_f}),
					 f^*(a)\ne 0.
				\end{cases}\]}
			\item[$\bullet$] If $n+s$ is odd, then
			{\small \[N_{(i,a,b)}=\begin{cases}
					0,&\text{if}\ a=0;\\
					\epsilon_0\eta(i)p^{\frac{n+s-3}{2}}+p^{n-2},&\text{if}\ a\ne 0, a\notin \mathrm{Supp}(\widehat{\chi_f});\\
					(\epsilon_0-\epsilon_a)\eta(i)p^{\frac{n+s-3}{2}}+p^{n-2},&\text{if}\ a\ne 0, a\in \mathrm{Supp}(\widehat{\chi_f}),f^*(a)=0;\\
					(\epsilon_a\eta(-f^*(a))(\delta_0(if^*(a)+b^24^{-1})p-1)+(\epsilon_0-\epsilon_a)\eta(i))&\text{if}\ a\ne 0, a\in \mathrm{Supp}(\widehat{\chi_f}),f^*(a)\ne 0.\\p^{\frac{n+s-3}{2}}+p^{n-2},	\end{cases}
				\]}
		\end{itemize}
	\end{lemma}
	\begin{proof}
	According to the definition of $N_{(i,a,b)}$, if $a=0$, we easily know that $N_{(i,a,b)}=0$. If $a\ne 0$, then we have that
	{\small \begin{align*}
			N_{(i,a,b)}&=p^{-2}\sum\limits_{x\in V_n^{(p)}}\sum\limits_{y\in\mathbb{F}_p}\sum\limits_{z\in\mathbb{F}_p}\xi_p^{y(f(x)-i)-z(\langle a,x\rangle_n -b)}\\
			&=p^{-2}(\sum\limits_{z\in\mathbb{F}_p^*}\sum\limits_{y\in\mathbb{F}_p^*}\sum\limits_{x\in V_n^{(p)}}\xi_p^{y(f(x)-i)-z(\langle a,x \rangle  _n-b)}+\sum\limits_{y\in\mathbb{F}_p^*}\sum\limits_{x\in V_n^{(p)}}\xi_p^{y(f(x)-i)}+\sum\limits_{z\in\mathbb{F}_p^*}\sum\limits_{x\in V_n^{(p)}}\xi_p^{z(\langle a,x\rangle_n-b)})+p^{n-2}.
	\end{align*}}
	Let $N_1=\sum_{z\in\mathbb{F}_p^*}\sum_{y\in\mathbb{F}_p^*}\sum_{x\in V_n^{(p)}}\xi_p^{y(f(x)-i)-z(\langle a,x \rangle  _n-b)}$, $N_2=\sum_{y\in\mathbb{F}_p^*}\sum_{x\in V_n^{(p)}}\xi_p^{y(f(x)-i)}$ and
	 $N_3=\sum_{z\in\mathbb{F}_p^*}\sum_{x\in V_n^{(p)}}\xi_p^{z(\langle a,x\rangle_n-b)}$, then $N_{(i,a,b)}=p^{-2}(N_1+N_2+N_3)+p^{n-2}$. In the following, we give

	  \noindent the values of 
	  $N_1$,
	   $N_2$ and $N_3$.
	   
	   	Note that 
	$N_1
	=\sum_{z\in\mathbb{F}_p^*}\sum_{y\in\mathbb{F}_p^*}\sigma_y(\widehat{\chi_f}(y^{-1}za))\xi_p^{-yi+zb}.$ 
	According to Lemma \ref{Le 7}, we know that for any $y,z\in \mathbb{F}_p^*$, if $a\notin\mathrm{Supp}(\widehat{\chi_f})$, then $y^{-1}za\notin\mathrm{Supp}(\widehat{\chi_f})$, and if  $a\in\mathrm{Supp}(\widehat{\chi_f})$, then $\epsilon_a=\epsilon _{y^{-1}za}$ and $f^*(y^{-1}za)=y^{-2}z^{2}f^*(a)$. Thus, if $a\notin\mathrm{Supp}(\widehat{\chi_f})$, then $N_1=0$. If $a\in \mathrm{Supp}(\widehat{\chi_f})$ and $n+s$ is even, we have that
	$N_1=\epsilon_ap^{\frac{n+s}{2}}\sum_{z\in\mathbb{F}_p^*}\sum_{y\in\mathbb{F}_p^*}\xi_p^{y^{-1}z^2f^*(a)-yi+zb}.$ 
	If $f^*(a)=0$, then by Lemma \ref{Le 1}, $N_1=\epsilon_ap^{\frac{n+s}{2}}$. If $f^*(a)\ne 0$, by Lemmas \ref{Le 1} and \ref{Le 2}, we get that 
	{\small\[
			N_1=\epsilon_ap^{\frac{n+s}{2}}\sum\limits_{y\in\mathbb{F}_p^*}\xi_p^{-yi}\sum\limits_{z\in\mathbb{F}_p^*}\xi_p^{y^{-1}f^*(a)z^2+zb}
			=\epsilon_ap^{\frac{n+s}{2}}(p\eta(b^24^{-1}+if^*(a))+1).
	\]}
	If $a\in \mathrm{Supp}(\widehat{\chi_f})$ and $n+s$ is odd, we have $
	N_1
	=\epsilon_ap^{\frac{n+s-1}{2}}\sqrt{p^*}\sum_{z\in\mathbb{F}_p^*}\sum_{y\in\mathbb{F}_p^*}\eta(y)\xi_p^{y^{-1}z^2f^*(a)-yi+zb}.$ 
	If $f^*(a)=0$, then by Lemma \ref{Le 1}, $N_1=-\epsilon_a\eta(i)p^{\frac{n+s+1}{2}}$. If $f^*(a)\ne 0$, by Lemmas \ref{Le 1} and \ref{Le 2}, we get that 
	{\small \[
			N_1  =\epsilon_ap^{\frac{n+s-1}{2}}\sqrt{p^*}\sum\limits_{y\in\mathbb{F}_p^*}\eta(y)\xi_p^{-yi}\sum\limits_{z\in\mathbb{F}_p^*}\xi_p^{y^{-1}f^*(a)z^2+zb} =\epsilon_ap^{\frac{n+s+1}{2}}(\eta(-f^*(a))(\delta_0(b^24^{-1}+if^*(a))p-1)-\eta(i)).
	\]}
	
	For the value of $N_2$, we have that
	$N_2=\sum_{y\in\mathbb{F}_p^*}\sigma_y(\widehat{\chi_f}(0))\xi_p^{-yi}.$
	Since $f^*(0)=0$ and
	$0\in\mathrm{Supp}(\widehat{\chi_f})$, then by Lemma \ref{Le 1}, when $n+s$ is even, $N_2=-\epsilon_0p^{\frac{n+s}{2}}$; when $n+s$ is odd, $N_2=\epsilon_0\eta(i)p^{\frac{n+s+1}{2}}$.	Finally, by Lemma \ref{Le 1}, we easily get that $N_3=0$.
	The proof is now completed. \end{proof}
	\section{Self-orthogonal codes from plateaued functions}
	In this section, based on the first and the second generic constructions, we construct some linear codes from plateaued functions and determine their weight distributions. Moreover, we prove that those codes are self-orthogonal under certain conditions. 
	\subsection{Self-orthogonal codes $\mathcal{C}_f$ and $\widetilde{\mathcal{C}}_f$}
	
	Let $f(x)$ be a function from $V_n^{(p)}$ to $\mathbb{F}_p$ with $f(0)=0$. In this subsection, we study the linear code $\mathcal{C}_f$ defined by 
	\begin{equation}
		\mathcal{C}_{f}=\{c(a,b)=(af(x)+\langle b,x\rangle_n)_{x\in {V_n^{(p)}}\setminus\{0\}}: a\in\mathbb{F}_p, b\in V_n^{(p)}\},
	\end{equation}
	and its punctured code $\widetilde{\mathcal{C}}_f$ defined by
	\begin{equation}
		\widetilde{\mathcal{C}}_{f}=\{\tilde{c}(a,b)=(af(x)+\langle b,x\rangle_n)_{x\in V_n^{(p)}\setminus D_{f,0}}: a\in\mathbb{F}_p,b\in V_n^{(p)}\}.
	\end{equation}
	Firstly, we give some results on the linear code $\mathcal{C}_f$ defined by (8).
	\begin{lemma}\label{Le 14}
		Let $f(x)$ be a function from $V_n^{(p)}$ to $\mathbb{F}_p$ with $f(0)=0$, where $p>3$ is an odd prime. If $f(x)=f(-x)$ for any $x\in V_n^{(p)}$, and $\sum_{i\in\mathbb{F}_p^*}i^2N_i(f)=0$, then the linear code $\mathcal{C}_f$ defined by $(8)$ is self-orthogonal.
	\end{lemma}
	\begin{proof}
	For any $a\in\mathbb{F}_p$ and $b\in V_n^{(p)}$, since $f(0)=0$, we have that
	{\small 	\begin{align*}
			\langle c(a,b),c(a,b)\rangle_m&=\sum\limits_{x\in V_n^{(p)}}(af(x)+\langle b,x\rangle_n)(af(x)+\langle b,x\rangle_n)\\
			&=\sum\limits_{x\in V_n^{(p)}}(\langle b,x\rangle_n)^2+\sum\limits_{i\in\mathbb{F}_p^*}\sum\limits_{x\in D_{f,i}}(a^2i^2+2ai\langle b, x\rangle_n)\\
			&=\sum\limits_{x\in V_n^{(p)}}(\langle b,x\rangle_n)^2+\sum\limits_{i\in\mathbb{F}_p^*}a^2i^2N_i(f)+\sum\limits_{i\in\mathbb{F}_p^*}\sum\limits_{x\in D_{f,i}}2ai\langle b,x\rangle_n,
	\end{align*}}
	where $m=p^n-1$. 
	
	If $b=0$, then $\sum_{x\in V_n^{(p)}}(\langle b,x\rangle_n)^2=0$. If $b\ne 0$, since $\langle b,x\rangle_n$ is a balanced function, then
	$\sum_{x\in V_n^{(p)}}(\langle b,x\rangle_n)^2
	=p^{n-1}\sum_{i\in\mathbb{F}_p^*}i^2.$ Since $1^2+2^2+\cdots+(p-1)^2=\frac{(p-1)(2p-1)}{6}p\equiv 0\ (\mathrm{mod}\ p)$ for $p>3$, then $\sum_{i\in\mathbb{F}_p^*}i^2=0$. Thus, $
	\sum_{x\in V_n^{(p)}}(\langle b,x\rangle_n)^2=0.$ Since $f(x)=f(-x)
	$ for any $x\in V_n^{(p)}$, then $D_{f,i}=-D_{f,i}$ for any $i\in\mathbb{F}_p^*$. Thus, $\sum_{i\in\mathbb{F}_p^*}\sum_{x\in D_{f,i}}2ai\langle b,x\rangle_n=\sum_{i\in\mathbb{F}_p^*}\sum_{x\in D_{f,i}}ai(\langle b,x\rangle_n+\langle b,-x\rangle_n)=0$. Since $\sum_{i\in\mathbb{F}_p^*}i^2N_i(f)=0$, then we have that $\langle c(a,b),c(a,b)\rangle_m=0$. According to Proposition \ref{Po 3}, we easily get that $\mathcal{C}_f$ is self-orthogonal.\end{proof}
	\begin{theorem}\label{Th 1}
		Let $n,s$ be integers with $0\le s\le n-2$ for even $n+s$, and $0\le s\le n-1$ for odd $n+s$. Let $f(x):V_n^{(p)}\longrightarrow \mathbb{F}_p$ be an $s$-plateaued function whose dual $f^*(x)$ is bent relative to $\mathrm{Supp}(\widehat{\chi_f})$, $f(0)=0$ and $\#B_+(f)=k$. Then $\mathcal{C}_f$ defined by (8) is a $p$-ary $[p^n-1,n+1]$  linear code and its weight distribution is given by Table \uppercase\expandafter{\romannumeral1} in the Appendix for even $n+s$, and Table \uppercase\expandafter{\romannumeral2} in the Appendix for odd $n+s$, respectively. Moreover, if $n+s\ge 3$ when $p=3$, and $(n,s)\ne (1,0)$, $f(x)=f(-x)$ for any $x\in V_n^{(p)}$  when $p\ne3$, then $\mathcal{C}_f$ is self-orthogonal. 
	\end{theorem}
	\begin{proof} Note that $\mathrm{wt}(c(a,b))=p^{n}-M_{a,b}$, where $M_{a,b}=\#\{x\in V_n^{(p)}: af(x)+\langle b,x\rangle_n=0\}$, then by Lemmas \ref{Le 5} and \ref{Le 10}, we can get the weight distribution of $\mathcal{C}_f$. We also know that $\mathrm{wt}(c(a,b))=0$ if and only if $a=0$ and $b=0$, thus the dimension of $\mathcal{C}_{f}$ is $n+1$. 
	
	Now, we prove that $\mathcal{C}_f$ is self-orthogonal. When $p=3$, if $n+s\ge 3$, we have that  $3|\mathrm{wt}(c(a,b))$ for any $a\in\mathbb{F}_p$ and $b\in V_n^{(p)}$. Thus, $\mathcal{C}_f$ is self-orthogonal. When $p\ne 3$, if $f(x)=f(-x)$ for any $x\in V_n^{(p)}$, by Lemma \ref{Le 8}, we have that $f^*(0)=0$. According to Lemma \ref{Le 3}, when $n+s$ is even, we have that $\sum_{i\in\mathbb{F}_p^*}i^2N_i(f)=(p^{n-1}-\epsilon_0p^{\frac{n+s}{2}-1})\sum_{i\in\mathbb{F}_p^*}i^2$. By the proof of Lemma \ref{Le 14}, we know that $\sum_{i\in\mathbb{F}_p^*}i^2=0$, so $\sum_{i\in\mathbb{F}_p^*}i^2N_i(f)=0$. When $n+s$ is odd, if $n+s\ge 3$, then $N_i(f)$ is divisible by $p$ for any $i\in\mathbb{F}_p^*$. Thus, we also get that $\sum_{i\in\mathbb{F}_p^*}i^2N_i(f)=0$. According to Lemma \ref{Le 14}, we have that $\mathcal{C}_f$ is self-orthogonal.
	\end{proof}
	
	\begin{remark}$(1)$	If $k=0$ or $k=p^{n-s}$, then $f(x)$ is weakly regular, so the linear codes in \cite{Mesnager2} can be obtained by Theorem \ref{Th 1}.\\
		$(2)$ According to \cite[Lemma 2.1]{Ashikhmin}, we have that if $0\le s\le n-4$ for even $n+s$, and $0\le s\le n-3$ for odd $n+s$, then the linear code $\mathcal{C}_f$ constructed in Theorem \ref{Th 1} is minimal.\\	$(3)$ Let $d^{\bot}$ denote the minimum distance of $\mathcal{C}_f^{\bot}$. By the first four Pless power moments given in \cite[pp. 259-260]{Huffman} and the weight distribution of $\mathcal{C}_f$, we have that when $n+s$ is even, $d^{\bot}=3$ if $s=0$, $n=2$ and $0\in B_-(f^*)$, otherwise $d^{\bot}=2$. When $n+s$ is odd, $d^{\bot}=3$ if $s=0,n=1$ and $p\ne 3$ which corresponds to quadratic bent functions \cite[Theorem 4.6]{Hou}, or else $d^{\bot}=2$ if $n>1$.
	\end{remark}
	
	Let $f(x):V_n^{(p)}\longrightarrow \mathbb{F}_p$ be an $s$-plateaued function belonging to $\mathscr{F}$. We easily have that for any $a\in\mathbb{F}_p^*$, $\alpha\in V_n^{(p)}$, $\alpha\in D_{f,0}$ if and only if $a\alpha\in D_{f,0}$. Thus, we can select the subset  $\widetilde{M}$ of $V_n^{(p)}\setminus D_{f,0}$, such that  $\bigcup_{a\in\mathbb{F}_p^*}a\widetilde{M}$ is a partition of  $V_n^{(p)}\setminus D_{f,0}$. In the following theorems, we give some results on the linear code $\widetilde{\mathcal{C}}_f$ defined by (9).
	\begin{theorem}\label{Th 2}
		Let $n,s$ be integers, for even $n+s$, $0\le s\le n-2$, $n+s\ge 4$, and for odd $n+s$, $0\le s\le n-1$, $n+s\ge 3$. Let  $f(x):V_n^{(p)}\longrightarrow \mathbb{F}_p$ be an $s$-plateaued function belonging to $\mathscr{F}$ and $\#B_+(f)=k$. Then $\widetilde{\mathcal{C}}_f$ defined by (9) is a $p$-ary linear code with parameters  $[(p-1)(p^{n-1}-\epsilon_0p^{\frac{n+s}{2}-1}),n+1]$ when $n+s$ is even, and $[(p-1)p^{n-1},n+1]$ when $n+s$ is odd. The weight distribution of $\widetilde{\mathcal{C}}_f$ is given by Table \uppercase\expandafter{\romannumeral3} in the Appendix for even $n+s$, and Table \uppercase\expandafter{\romannumeral4} in the Appendix for odd $n+s$, respectively. The dual code ${\widetilde{\mathcal{C}}_f}^{\bot}$ is a $p$-ary linear code with parameters $[(p-1)(p^{n-1}-\epsilon_0p^{\frac{n+s}{2}-1}),(p-1)(p^{n-1}-\epsilon_0p^{\frac{n+s}{2}-1})-n-1,3]$ when $n+s$ is even, and $[(p-1)p^{n-1},(p-1)p^{n-1}-n-1,3]$ when $n+s$ is odd.  Moreover, if $p\ne 3$ or $n+s\ge 5$ for $p=3$, then $\widetilde{\mathcal{C}}_f$ is self-orthogonal. 
	\end{theorem}
	\begin{proof}
	According to Lemma \ref{Le 8}, we have that $f^*(0)=0$. By Lemma \ref{Le 3}, we get that  the length of $\widetilde{\mathcal{C}}_f$ is $(p-1)(p^{n-1}-\epsilon_0p^{\frac{n+s}{2}-1})$ when $n+s$ is even, and $(p-1)p^{n-1}$ when $n+s$ is odd. Let $\widetilde{N}=\#\{x\in V_n^{(p)}\setminus D_{f,0}: af(x)+\langle b,x\rangle_n=0\}$. Since $\widetilde{N}=M_{a,b}-N_{(0,b)}$, where $M_{a,b}=\#\{x\in V_n^{(p)}:af(x)+\langle b,x\rangle_n=0\}$ and $N_{(0,b)}=\#\{x\in D_{f,0}:\langle b,x\rangle_n=0\}$, then by Lemmas \ref{Le 5}, \ref{Le 7}, \ref{Le 10} and \ref{Le 11}, we can get the weight distribution of $\widetilde{\mathcal{C}}_f$. We also know that $\mathrm{wt}(\tilde{c}(a,b))=0$ if and only if $a=0$ and $b=0$, then the dimension of $\widetilde{\mathcal{C}}_f$ is $n+1$. By the weight distribution of $\widetilde{\mathcal{C}}_f$ and the first four Pless power moments  given in \cite[pp. 259-260]{Huffman}, when $n+s$ is even, we have that $A_1^{\bot}=A_2^{\bot}=0$ and  $A_3^{\bot}=\frac{(p-1)^2(p-2)}{6}(p^{2n-3}(p^2-2p+1)+p^{n+s-2}(2p-3)+p^{\frac{n+s}{2}-1}-(2p-4)kp^{\frac{n+3s}{2}-3}-(3p^2-7p+5)p^{\frac{3n+s}{2}-3}-p^{n-1})>0$ for $0\in B_+(f)$, $A_3^{\bot}=\frac{(p-1)^2(p-2)}{6}(p^{2n-3}(p^2-2p+1)+p^{n+s-2}(2p-3)+(3p^2-5p+1)p^{\frac{3n+s}{2}-3}-(2p-4)kp^{\frac{n+3s}{2}-3}-p^{\frac{n+s}{2}-1}-p^{n-1})>0$ for $0\in B_-(f)$, so the minimal distance of ${\widetilde{\mathcal{C}}_f}^{\bot}$ is $3$. When $n+s$ is odd, we have that $A_1^{\bot}=A_2^{\bot}=0$ and $A_3^{\bot}=\frac{(p-1)^2(p-2)}{6}(p^{2n-3}+p^{2n-1}-2p^{2n-2}-p^{n+s-2}-p^{n-1})>0$, so the minimal distance of ${\widetilde{\mathcal{C}}_f}^{\bot}$ is also $3$.

	Now, we prove that $\widetilde{\mathcal{C}}_f$ is self-orthogonal. When $p=3$, by the weight distribution of $\widetilde{\mathcal{C}}_f$, we easily know that if $n+s\ge 5$, then  $3|\mathrm{wt}(\tilde{c}(a,b))$ for any $a\in\mathbb{F}_p$ and $b\in V_n^{(p)}$, thus $\widetilde{\mathcal{C}}_f$ is self-orthogonal. When $p\ne 3$,  we know that 
	{\small \[\langle \tilde{c}(a,b),\tilde{c}(a,b)\rangle_m=\sum\limits_{x\in V_n^{(p)}\setminus D_{f,0}}(\langle b,x\rangle_n)^2+\sum\limits_{i\in\mathbb{F}_p^*}a^2i^2N_i(f)+\sum\limits_{i\in\mathbb{F}_p^*}\sum\limits_{x\in D_{f,i}}2ai\langle b,x\rangle_n,\]}
	where $m=(p-1)(p^{n-1}-\epsilon_0p^{\frac{n+s}{2}-1})$ for even $n+s$ and $m=(p-1)p^{n-1}$ for odd $n+s$. 
	Since $f(x)$ belongs to $\mathscr{F}$, then $f(x)=f(-x)$ for any $x\in V_n^{(p)}$, by the proofs of Lemma \ref{Le 14} and Theorem \ref{Th 1}, we know that $\sum_{i\in\mathbb{F}_p^*}i^2N_i(f)=0$ and  $\sum_{i\in\mathbb{F}_p^*}\sum_{x\in D_{f,i}}2ai\langle b,x\rangle_n=0$ for any $a\in\mathbb{F}_p$, $b\in V_n^{(p)}$.  Note that for any $b\in V_n^{(p)}$,
	{\small \begin{equation*}\sum\limits_{x\in V_n^{(p)}\setminus D_{f,0}}(\langle b,x\rangle_n)^2=\sum\limits_{i\in\mathbb{F}_p^*}i^2\sum\limits_{x\in \widetilde{M}}(\langle b,x\rangle_n)^2.
	\end{equation*}}
	By the proof of Lemma \ref{Le 14}, we have  $\sum_{i\in\mathbb{F}_p^*}i^2=0$, so $\sum_{x\in V_n^{(p)}\setminus D_{f,0}}(\langle b,x\rangle_n)^2=0.$ Hence, $\langle \tilde{c}(a,b),\tilde{c}(a,b)\rangle_m=0$ for any $a\in\mathbb{F}_p$ and $b\in V_n^{(p)}$. According to Proposition \ref{Po 3}, we have that ${\widetilde{\mathcal{C}}}_f$ is self-orthogonal.
	This completes the proof.\end{proof}
	\begin{remark}
		According to the sphere packing bound, when $n+s$ is even, if $p=3$, $0\in B_+(f)$ and $n-s=2$, for fixed length and dimension,  ${\widetilde{\mathcal{C}}_f}^{\bot}$ is at least almost optimal; in other cases, for fixed length and minimum distance, ${\widetilde{\mathcal{C}}_f}^{\bot}$ is optimal.
		When $n+s$ is odd, for fixed length and minimum distance,  ${\widetilde{\mathcal{C}}_f}^{\bot}$ is optimal.
	\end{remark}
	\subsection{Self-orthogonal codes $\mathcal{C}_D$ and $\widetilde{\mathcal{C}}_D$}
	Let $f(x)$ be a function from $V_n^{(p)}$ to $\mathbb{F}_p$ with $f(0)=0$. Let $D$ be $D_{f,0}\setminus\{0\}$, $D_{f,sq}$ and $D_{f,nsq}$, respectively.  In the subsection, we study the linear codes $\mathcal{C}_D$ defined by (2) and their punctured codes $\widetilde{\mathcal{C}}_D$. We begin this section with the following lemma.
	\begin{lemma}\label{Le 15}
		Let $f(x)$ be a function from $V_n^{(p)}$ to $\mathbb{F}_p$ with $f(0)=0$, where $p>3$ is an odd prime. For any $a\in\mathbb{F}_p^*$, if there exists an even integer $l_a$ with $2\le l_a\le p-1$ such that $f(ax)=a^{l_a}f(x)$ for any $x\in V_n^{(p)}$, then the linear codes $\mathcal{C}_{D_{f,0}\setminus\{0\}}$, $\mathcal{C}_{D_{f,sq}}$ and $\mathcal{C}_{D_{f,nsq}}$ are self-orthogonal.
	\end{lemma}
	\begin{proof} Since $l_a$ is an even integer for any $a\in\mathbb{F}_p^*$, then we have that $x\in D_{f,0}$ (respectively $D_{f,sq}$, $D_{f,nsq}$) if and only if $ax\in D_{f,0}$ (respectively $D_{f,sq}$, $D_{f,nsq}$) for any $x\in V_n^{(p)}$. Thus, we can select the subsets $\widetilde{D}_{f,0}$, $\widetilde{D}_{f,sq}$ and $\widetilde{D}_{f,nsq}$ such that $\bigcup_{a\in\mathbb{F}_p^*}a\widetilde{D}_{f,0}$, $\bigcup_{a\in\mathbb{F}_p^*}a\widetilde{D}_{f,sq}$ and $\bigcup_{a\in\mathbb{F}_p^*}a\widetilde{D}_{f,nsq}$ are partitions of $D_{f,0}$, $D_{f,sq}$ and $D_{f,nsq}$, respectively. For linear code $\mathcal{C}_{D_{f,0}\setminus\{0\}}$, we have that
	{\small \[\langle c(b),c(b)\rangle_m=\sum\limits_{x\in D_{f,0}}(\langle b,x\rangle_n)^2=\sum\limits_{i\in\mathbb{F}_p^*}i^2\sum\limits_{x\in \widetilde{D}_{f,0}}(\langle b,x\rangle_n)^2,\]}
	where $m=\#D_{f,0}-1$ and $b\in V_n^{(p)}$. 
	By the proof of Lemma \ref{Le 14}, we know that $\sum_{i\in\mathbb{F}_p^*}i^2=0$, so $\langle c(b),c(b)\rangle_m=0$ for any $b\in V_n^{(p)}$. According to Proposition \ref{Po 3}, we have that the linear code $\mathcal{C}_{D_{f,0}\setminus\{0\}}$ is self-orthogonal. Similarly, we can get that the linear codes $\mathcal{C}_{D_{f,sq}}$ and $\mathcal{C}_{D_{f,nsq}}$ are both self-orthogonal.\end{proof}
	
	\begin{theorem}\label{Th 3}
		Let $n,s$ be  integers with $0\le s\le n-4$ for even $n+s$, and $0\le s\le n-3$ for odd $n+s$. Let  $f(x):V_n^{(p)}\longrightarrow \mathbb{F}_p$ be an $s$-plateaued function belonging to $\mathscr{F}$ and $\#B_+(f)=k$. Then $\mathcal{C}_{D_{f,0}\setminus\{0\}}$ is a $p$-ary linear code with parameters $[p^{n-1}+\epsilon_0(p-1)p^{\frac{n+s}{2}-1}-1,n]$ when $n+s$ is even, and  $[p^{n-1}-1,n]$ when $n+s$ is odd. The weight distribution of $\mathcal{C}_{D_{f,0}\setminus\{0\}}$ is given by Table \uppercase\expandafter{\romannumeral5} in the Appendix for even $n+s$, and Table \uppercase\expandafter{\romannumeral6} in the Appendix for odd $n+s$, respectively. 
		Moreover, if $p\ne 3$ or $n+s\ge 5
		$ for $p=3$, then $\mathcal{C}_{D_{f,0}\setminus\{0\}}$ is self-orthogonal.
	\end{theorem}
	\begin{proof} According to Lemma \ref{Le 8}, we have that $f^*(0)=0$. By Lemma \ref{Le 3}, we get that the length of  $\mathcal{C}_{D_{f,0}\setminus\{0\}}$ is $p^{n-1}+\epsilon_0(p-1)p^{\frac{n+s}{2}-1}-1$ when $n+s$ is even, and $p^{n-1}-1$ when $n+s$ is odd. Thus, by Lemmas \ref{Le 5}, \ref{Le 7} and \ref{Le 11}, we can get the weight distribution of $\mathcal{C}_{D_{f,0}\setminus\{0\}}$. We also know that $\mathrm{wt}(c(a))=0$ if and only if $a=0$, then the dimension of $\mathcal{C}_{D_{f,0}\setminus\{0\}}$ is $n$. 
	
	Now, we prove that $\mathcal{C}_{D_{f,0}\setminus\{0\}}$ is self-orthogonal. When $p=3$, by the weight distribution of $\mathcal{C}_{D_{f,0}\setminus\{0\}}$, we know that if $n+s\ge 5$, then $3|\mathrm{wt}(c(a))$ for any $a\in V_n^{(p)}$, thus $\mathcal{C}_{D_{f,0}\setminus\{0\}}$ is self-orthogonal. When $p\ne 3$, since $f(x)\in\mathscr{F}$, by Lemma \ref{Le 15}, we easily get that $\mathcal{C}_{D_{f,0}\setminus\{0\}}$ is self-orthogonal. \end{proof}
	
	\begin{theorem}\label{Th 4}
		Let $n,s$ be integers with $0\le s\le n-4$ for even $n+s$, and $0\le s\le n-3$ for odd $n+s$. Let  $f(x):V_n^{(p)}\longrightarrow \mathbb{F}_p$ be an $s$-plateaued function belonging to $\mathscr{F}$ and $\#B_+(f)=k$. When $n+s$ is even,  $\mathcal{C}_{D_{f,sq}}$ and $\mathcal{C}_{D_{f,nsq}}$ are $p$-ary $[\frac{(p-1)}{2}(p^{n-1}-\epsilon_0p^{\frac{n+s}{2}-1}),n]$  linear codes. When $n+s$ is odd,  $\mathcal{C}_{D_{f,sq}}$ is a $p$-ary $[\frac{(p-1)}{2}(p^{n-1}+\epsilon_0p^{\frac{n+s-1}{2}}), n]$ linear code, and 
		$\mathcal{C}_{D_{f,nsq}}$ is a $p$-ary $[\frac{(p-1)}{2}(p^{n-1}-\epsilon_0p^{\frac{n+s-1}{2}}),n]$ linear code.  The weight distributions of $\mathcal{C}_{D_{f,sq}}$ and $\mathcal{C}_{D_{f,nsq}}$ are given by Tables \uppercase\expandafter{\romannumeral7}, \uppercase\expandafter{\romannumeral8} and \uppercase\expandafter{\romannumeral9} in the Appendix. 
		Moreover, if $p\ne 3$ or $n+s\ge 5$ for $p=3$, then $\mathcal{C}_{D_{f,sq}}$ and $\mathcal{C}_{D_{f,nsq}}$ are self-orthogonal.
	\end{theorem}
	\begin{proof}
	We only prove the case of $\mathcal{C}_{D_{f,sq}}$ and the case of $\mathcal{C}_{D_{f,nsq}}$ is similar.
	
	According to Lemma \ref{Le 8}, we have that $f^*(0)=0$. By Lemma \ref{Le 3}, we get that the length of $\mathcal{C}_{D_{f,sq}}$ is $\frac{(p-1)}{2}(p^{n-1}-\epsilon_0p^{\frac{n+s}{2}-1})$ when $n+s$ is even, and $\frac{(p-1)}{2}(p^{n-1}+\epsilon_0p^{\frac{n+s-1}{2}})$ when $n+s$ is odd. Thus, by Lemmas \ref{Le 5}, \ref{Le 7} and \ref{Le 12},  we can get the weight distribution of $\mathcal{C}_{D_{f,sq}}$. We also know that $\mathrm{wt}(c(a))=0$ if and only if $a=0$, then the dimension of $\mathcal{C}_{D_{f,sq}}$ is $n$. 
	
	Now, we prove that $\mathcal{C}_{D_{f,sq}}$ is self-orthogonal. When $p=3$, by the weight distribution of $\mathcal{C}_{D_{f,sq}}$, we know that if $n+s\ge 5$, then $3|\mathrm{wt}(c(a))$ for any $a\in V_n^{(p)}$, thus $\mathcal{C}_{D_{f,sq}}$ is self-orthogonal. When $p\ne 3$, since $f(x)\in\mathscr{F}$, by Lemma \ref{Le 15}, we easily get that $\mathcal{C}_{D_{f,sq}}$ is self-orthogonal.\end{proof}
	\begin{remark}
		Let $d^{\bot}$ be the minimum distance of $\mathcal{C}_{D_{f,0}\setminus\{0\}}^{\bot}$. Obviously, $d^{\bot} >1$. Due to $f(x)\in\mathscr{F}$, then $f(x)=f(-x)$ for any $x\in V_n^{(p)}$. Thus, if $x\in D_{f,0}\setminus\{0\}$, then $-x\in D_{f,0}\setminus\{0\}$.   Note that $\langle x, a\rangle_n+\langle (-x),a\rangle_n=0$ for any $a\in V_n^{(p)}$, which implies that $d^{\bot}=2$. By the similar discussion, we also have that the minimum distances of $\mathcal{C}_{D_{f,sq}}^{\bot}$ and $\mathcal{C}_{D_{f,nsq}}^{\bot}$ are both $2$.
	\end{remark}
	
	In Equation $(2)$, let $D$ be $\widetilde{D}_{f,0}\setminus\{0\}$, $\widetilde{D}_{f,sq}$ and $\widetilde{D}_{f,nsq}$ (given in the proof of Lemma \ref{Le 15}). In the following, we study the punctured codes $\widetilde{\mathcal{C}}_{D_{f,0}\setminus\{0\}}$, $\widetilde{\mathcal{C}}_{D_{f,sq}}$ and $\widetilde{\mathcal{C}}_{D_{f,nsq}}$ of $\mathcal{C}_{D_{f,0}\setminus\{0\}}$, $\mathcal{C}_{D_{f,sq}}$ and $\mathcal{C}_{D_{f,nsq}}$, respectively, and give some results on them.
	\begin{theorem}\label{Th 5}
		Let $n,s$ be integers with $0\le s\le n-4$ for even $n+s$, and $0\le s\le n-3$ for odd $n+s$. Let  $f(x):V_n^{(p)}\longrightarrow \mathbb{F}_p$ be an $s$-plateaued function belonging to $\mathscr{F}$ and $\#B_+(f)=k$. Then $\widetilde{\mathcal{C}}_{D_{f,0}\setminus\{0\}}$  is a $p$-ary linear code with parameters $[\frac{p^{n-1}+\epsilon_0(p-1)p^{\frac{n+s}{2}-1}-1}{p-1}, n]$ when $n+s$ is even, and $[\frac{p^{n-1}-1}{p-1}, n]$ when $n+s$ is odd.  The weight distribution of $\widetilde{\mathcal{C}}_{D_{f,0}\setminus\{0\}}$ is given by Table \uppercase\expandafter{\romannumeral10} in the Appendix for even $n+s$, and Table \uppercase\expandafter{\romannumeral11} in the Appendix for odd $n+s$, respectively. 
		Moreover, if $p=3$ and $n+s\ge 5$, then $\widetilde{\mathcal{C}}_{D_{f,0}\setminus\{0\}}$ is self-orthogonal.
	\end{theorem}
	\begin{proof}
	The weight distribution of $\widetilde{\mathcal{C}}_{D_{f,0}\setminus\{0\}}$ can be easily obtained by Tables \uppercase\expandafter{\romannumeral5} and \uppercase\expandafter{\romannumeral6} in the Appendix. When $p=3$, by the weight distribution of $\widetilde{\mathcal{C}}_{D_{f,0}\setminus\{0\}}$, we easily know that if $n+s\ge 5$, then  $\widetilde{\mathcal{C}}_{D_{f,0}\setminus\{0\}}$ is self-orthogonal.\end{proof}
	\begin{proposition}\label{Po 6}
		Let $\widetilde{\mathcal{C}}_{D_{f,0}\setminus\{0\}}$ be the linear code given by Theorem \ref{Th 5}, and $d^{\bot}$ denote the minimum distance of the dual code ${\widetilde{\mathcal{C}}_{D_{f,0}\setminus\{0\}}}^{\bot}$. When $n+s$ is even, ${\widetilde{\mathcal{C}}_{D_{f,0}\setminus\{0\}}}^{\bot}$ is a $p$-ary   $[\frac{p^{n-1}+\epsilon_0(p-1)p^{\frac{n+s}{2}-1}-1}{p-1},\\ \frac{p^{n-1}+\epsilon_0(p-1)p^{\frac{n+s}{2}-1}-1}{p-1}-n]$ linear code and $d^{\bot}=4$ if $s=0$, $n=4$ and $k=0$, otherwise $d^{\bot}=3$. When $n+s$ is odd, the dual code $\widetilde{\mathcal{C}}_{D_{f,0}\setminus\{0\}}^{\bot}$ is a $p$-ary $[\frac{p^{n-1}-1}{p-1}, \frac{p^{n-1}-1}{p-1}-n]$ linear code and $d^{\bot}=4$ if  $s=0$ and $n=3$, otherwise $d^{\bot}=3$.
	\end{proposition}
	\begin{proof}
	The length and dimension of ${\widetilde{\mathcal{C}}_{D_{f,0}\setminus\{0\}}}^{\bot}$ are easily obtained by Theorem \ref{Th 5}. When $n+s$ is even,  by the weight distribution of $\widetilde{\mathcal{C}}_{D_{f,0}\setminus\{0\}}$ and  the first five Pless power moments  given in \cite[pp. 259-260]{Huffman}, we have that 
	when $0\in B_+(f)$, $A_1^{\bot}=A_2^{\bot}=0$ and $A_3^{\bot}=\frac{1}{6p^3}(p^{2n}-p^{n+2}(p+1)+p^4-p^{\frac{3n+s}{2}}(p^2-6p+5)+2kp^{\frac{n+3s}{2}}(p^2-3p+2)-p^{\frac{n+s}{2}+2}(p^2-1)+p^{n+s+2}(p-1))>0$, so $d^{\bot}=3$. When $0\in B_-(f)$, if $s=0$, $n=4$ and $k\ne 0,$ then  $A_1^{\bot}=A_2^{\bot}=0$ and $A_3^{\bot}=\frac{2kp^2(p^2-3p+2)}{6p^3}>0$, so $d^{\bot}=3$; if $s=0$, $n=4$ and $k=0$, then $A_1^{\bot}=A_2^{\bot}=A_{3}^{\bot}=0$ and $A_{4}^{\bot}=\frac{p^2(p-1)^2(p^4-p^3-p^2-p-2)}{24}>0$, so $d^{\bot}=4$; if  $n>4$, then $A_1^{\bot}=A_2^{\bot}=0$ and $A_3^{\bot}=\frac{1}{6p^3}(p^{2n}-p^{n+2}(p+1)+p^4
	-p^{\frac{3n+s}{2}}(p^2-1)+2kp^{\frac{n+3s}{2}}(p^2-3p+2)+p^{\frac{n+s}{2}}(p^4-p^2)+p^{n+s+2}(p-1))>0$, so $d^{\bot}=3$.
	
	When $n+s$ is odd,  by  the weight distribution of $\widetilde{\mathcal{C}}_{D_{f,0}\setminus\{0\}}$ and the first five Pless power moments  given in \cite[pp. 259-260]{Huffman}, we have that if $s=0$ and $n=3$, then $A_1^{\bot}=A_2^{\bot}=
	A_3^{\bot}=0$ and $A_4^{\bot}=\frac{p(p-1)^2(p^2-p-2)}{24}>0$, so $d^{\bot}=4$; if $n>3$, we have $A_1^{\bot}=A_2^{\bot}=0$ and $A_3^{\bot}=\frac{p^{2n}-p^{n+2}-p^{n+3}+p^{n+s+2}-p^{n+s+1}+p^4}{6p^3}>0$, so $d^{\bot}=3$. \end{proof}
	\begin{remark}
		$(1)$ According to the sphere packing bound, when $n+s$ is even, if $0\in B_+(f)$, then for fixed length and minimum distance, $\widetilde{\mathcal{C}}_{D_{f,0}\setminus\{0\}}^{\bot}$ is optimal; if $s=k=0$ and $n=4$, then for fixed length and dimension, $\widetilde{\mathcal{C}}_{D_{f,0}\setminus\{0\}}^{\bot}$ is optimal; in other cases, for fixed length and dimension, $\widetilde{\mathcal{C}}_{D_{f,0}\setminus\{0\}}^{\bot}$ is at least almost optimal. When $n+s$ is odd, for fixed length and dimension, $\widetilde{\mathcal{C}}_{D_{f,0}\setminus\{0\}}^{\bot}$ is at least almost optimal.\\
		$(2)$ For an $s$-plateaued function $f(x)$ belonging to $\mathscr{F}$, it is known that the set $\widetilde{D}_{f,0}\setminus\{0\}$ is not unique. When $p\ge 5$, for different sets $\widetilde{D}_{f,0}\setminus\{0\}$, the self-orthogonality of linear codes $\widetilde{\mathcal{C}}_{D_{f,0}\setminus\{0\}}$ may be different. By Magma program, when $p=5$, $n=3$ and $s=0$, let $f(x_1,x_2,x_3)=x_1^2+x_2^2+x_3^2$ and $\widetilde{D}_{f,0}\setminus\{0\}=\{(0,1,2),(0,1,3),(2,0,4),(2,0,1),(1,2,0),(1,3,0)\}$, we have that the linear code $\widetilde{\mathcal{C}}_{D_{f,0}\setminus\{0\}}$ is self-orthogonal. When $p=7$, $n=3$ and $s=0$, let $f(x_1,x_2,x_3)=x_1^2+x_2^2+x_3^2$, we can't find a set $\widetilde{D}_{f,0}\setminus\{0\}$ such that the linear code $\widetilde{\mathcal{C}}_{D_{f,0}\setminus\{0\}}$ is self-orthogonal. Then, it is open to study the relationship between the function $f(x)$, the set $\widetilde{D}_{f,0}\setminus\{0\}$ and the self-orthogonality of linear code $\widetilde{\mathcal{C}}_{D_{f,0}\setminus\{0\}}$ for $p\ge 5$
		. 
	\end{remark}
	
	Let $f(x):V_n^{(p)}\longrightarrow \mathbb{F}_p$ be an $s$-plateaued function belonging to $\mathscr{F}$ with $t=t'=2$.  For any $i\in\mathbb{F}_p^*$, $x\in V_n^{(p)}$, we have that if $x\in D_{f,i}$, then $-x\in D_{f,i}$. Thus, there exists a subset $\widetilde{D}_{f,i}$ of $D_{f,i}$ such that $\widetilde{D}_{f,i}\bigcup-\widetilde{D}_{f,i}$ is a partition of $D_{f,i}$. Note that for any $x\in \widetilde{D}_{f,i}$, there doesn't exist $a\in\mathbb{F}_p^*$ with $a\ne 1$ such that $ax\in\widetilde{D}_{f,i}$. Hence, $D_{f,sq}=\bigcup _{a\in\mathbb{F}_p^*}a\widetilde{D}_{f,i}$, where $i\in SQ$, and $D_{f,nsq}=\bigcup _{a\in\mathbb{F}_p^*}a\widetilde{D}_{f,i}$, where $i\in NSQ$. Then, we have the following theorem.
	\begin{theorem}\label{Th 6}
		Let $n,s$ be integers with $0\le s\le n-4$ for even $n+s$, and $0\le s\le n-3$ for odd $n+s$. Let  $f(x):V_n^{(p)}\longrightarrow \mathbb{F}_p$ be an $s$-plateaued function belonging to $\mathscr{F}$ and $\#B_+(f)=k$. When $n+s$ is even, $\widetilde{\mathcal{C}}_{D_{f,sq}}$ and $\widetilde{\mathcal{C}}_{D_{f,nsq}}$ are $p$-ary $[\frac{p^{n-1}-\epsilon_0p^{\frac{n+s}{2}-1}}{2}, n]$ linear codes. When $n+s$ is odd, $\widetilde{\mathcal{C}}_{D_{f,sq}}$ is a $p$-ary $[\frac{p^{n-1}+\epsilon_0p^{\frac{n+s-1}{2}}}{2},n]$ linear code, and 
		$\widetilde{\mathcal{C}}_{D_{f,nsq}}$ is a  $p$-ary $[\frac{p^{n-1}-\epsilon_0p^{\frac{n+s-1}{2}}}{2}, n]$ linear code. The weight distributions of $\widetilde{\mathcal{C}}_{D_{f,sq}}$ and $\widetilde{\mathcal{C}}_{D_{f,nsq}}$ are given by Tables \uppercase\expandafter{\romannumeral12}, \uppercase\expandafter{\romannumeral13} and \uppercase\expandafter{\romannumeral14} in the Appendix.
		Moreover, if $n+s\ge 5$ for $p=3$, and $t=t'=2$, $\widetilde{D}_{f,sq}=\widetilde{D}_{f,i}\ (i\in SQ)$ and $\widetilde{D}_{f,nsq}=\widetilde{D}_{f,i}\ (i\in NSQ)$ for $p\ne3$, then $\widetilde{\mathcal{C}}_{D_{f,sq}}$ and $\widetilde{\mathcal{C}}_{D_{f,nsq}}$ are self-orthogonal.
	\end{theorem}
	\begin{proof}
	The weight distributions of $\widetilde{\mathcal{C}}_{D_{f,sq}}$ and $\widetilde{\mathcal{C}}_{D_{f,nsq}}$ can be easily obtained by Tables \uppercase\expandafter{\romannumeral7}, \uppercase\expandafter{\romannumeral8} and \uppercase\expandafter{\romannumeral9} in the Appendix. Now, we prove that $\widetilde{\mathcal{C}}_{D_{f,sq}}$ and $\widetilde{\mathcal{C}}_{D_{f,nsq}}$ are self-orthogonal.
	When $p=3$, by the weight distributions of $\widetilde{\mathcal{C}}_{D_{f,sq}}$ and $\widetilde{\mathcal{C}}_{D_{f,nsq}}$, we easily know that if $n+s\ge 5$, then  $\widetilde{\mathcal{C}}_{D_{f,sq}}$ and $\widetilde{\mathcal{C}}_{D_{f,nsq}}$ are self-orthogonal. When $p\ne3$, let 
	\[c'(a)=(\langle a,x_1\rangle_n, \langle a, x_2\rangle_n, \cdots, \langle a,x_r\rangle_n),\] where $a\in V_n^{(p)}$ and $\{x_1,x_2,\cdots,x_r\}=D_{f,i}\ (i\in \mathbb{F}_p^*)$. Then, we have that for any $a\in V_n^{(p)}$, $\langle c'(a),c'(a)\rangle_r=\sum_{b\in\mathbb{F}_p^*}b^2N_{i,a,b}$, where $N_{i,a,b}=\#\{x \in D_{f,i}: \langle a ,x \rangle_n=b\}$. By the proof of Lemma \ref {Le 14}, we have that $\sum_{b\in\mathbb{F}_p^*}b^2=0$. Thus, when $t=t'=2$, by Lemma \ref{Le 13}, we have that $\langle c'(a),c'(a)\rangle_r=0$. On the other hand, $\langle c'(a),c'(a)\rangle_r=2\sum_{x\in \widetilde{D}_{f,i}}(\langle a ,x \rangle _n)^2,$ which implies that  $\sum_{x\in\widetilde{D}_{f,i}}(\langle a, x \rangle_n)^2=0$. Therefore, by Proposition \ref{Po 3}, if $\widetilde{D}_{f,sq}=\widetilde{D}_{f,i}\ (i\in SQ)$ and $\widetilde{D}_{f,nsq}=\widetilde{D}_{f,i}\ (i\in NSQ)$, the linear codes  $\widetilde{\mathcal{C}}_{D_{f,sq}}$ and $\widetilde{\mathcal{C}}_{D_{f,nsq}}$ are self-orthogonal. \end{proof}
	\begin{proposition}\label{Po 7}
		Let $\widetilde{\mathcal{C}}_{D_{f,sq}}$ and $\widetilde{\mathcal{C}}_{D_{f,nsq}}$ be the linear codes given by Theorem \ref{Th 6}. When $n+s$ is even, the dual codes ${\widetilde{\mathcal{C}}_{D_{f,sq}}}^{\bot}$ and ${\widetilde{\mathcal{C}}_{D_{f,nsq}}}^{\bot}$ are $p$-ary $[\frac{p^{n-1}-\epsilon_0p^{\frac{n+s}{2}-1}}{2}, \frac{p^{n-1}-\epsilon_0p^{\frac{n+s}{2}-1}}{2}-n, 3]$ linear codes. When $n+s$ is odd, the dual code ${\widetilde{\mathcal{C}}_{D_{f,sq}}}^{\bot}$ is a $p$-ary $[\frac{p^{n-1}+\epsilon_0p^{\frac{n+s-1}{2}} }{2},\ \frac{p^{n-1}+\epsilon_0p^{\frac{n+s-1}{2}} }{2}-n,\ 3]$ linear code except for $p=3$, $n=3$ and $0\in B_-(f)$, and the dual code ${\widetilde{\mathcal{C}}_{D_{f,nsq}}}^{\bot}$ is a $p$-ary  $[\frac{p^{n-1}-\epsilon_0p^{\frac{n+s-1}{2}}}{2},\ \frac{p^{n-1}-\epsilon_0p^{\frac{n+s-1}{2}}}{2}-n,\ 3]$ linear code except for $p= 3$, $n=3$ and $0\in B_+(f)$.
	\end{proposition}
	\begin{proof}
	The lengths and dimensions of $\widetilde{\mathcal{C}}_{D_{f,sq}}^{\bot}$ and $\widetilde{\mathcal{C}}_{D_{f,nsq}}^{\bot}$ are easily obtained by Theorem \ref{Th 6}. When $n+s$ is even, by the weight distributions of $\widetilde{\mathcal{C}}_{D_{f,sq}}$ and $\widetilde{\mathcal{C}}_{D_{f,nsq}}$, and the first four Pless power moments  given in \cite[pp. 259-260]{Huffman}, for linear codes $\widetilde{\mathcal{C}}_{D_{f,sq}}$ and $\widetilde{\mathcal{C}}_{D_{f,nsq}}$, we have $A_1^{\bot}=A_2^{\bot}=0$, $A_3^{\bot}=\frac{1}{48p^3} (p-1)(4kp^{\frac{3s+n}{2}}(p+1)-(3p^2-4p+5)p^{\frac{n+s}{2}+n}+(p-1)^2p^{2n}-(4p^3-8p^2)(p^n-p^{\frac{n+s}{2}})+(2p^2-6p)p^{n+s})>0$ for $0\in B_+(f)$ and $A_3^{\bot}=\frac{1}{48p^3}(p-1)(p^{\frac{3n+s}{2}}(3p^2-8p+1)+4kp^{\frac{3s+n}{2}}(p+1)+(p-1)^2p^{2n}-(4p^3-8p^2)(p^n+p^{\frac{n+s}{2}})+(2p^2-6p)p^{s+n})>0$ for $0\in B_-(f)$, so the minimum distances of ${\widetilde{\mathcal{C}}_{D_{f,sq}}}^{\bot}$ and ${\widetilde{\mathcal{C}}_{D_{f,nsq}}}^{\bot}$ are 3. 
	
	When $n+s$ is odd, by the weight distributions of $\widetilde{\mathcal{C}}_{D_{f,sq}}$ and $\widetilde{\mathcal{C}}_{D_{f,nsq}}$, and the first four Pless power moments  given in \cite[pp. 259-260]{Huffman}, for linear code ${\widetilde{\mathcal{C}}_{D_{f,sq}}}^{\bot}$, when $0\in B_+(f)$, we have $A_1^{\bot}=A_2^{\bot}=0$ and $A_3^{\bot}=\frac{1}{48}(p-1)((3p^2-6p-1)p^{n+s-2}+(4p-8)(p^{\frac{3n+s-3}{2}}-p^{\frac{n+s-1}{2}}-p^{n-1})+(p-1)^2p^{2n-3}-2(p^{n-s}-k)p^{\frac{n+3s-5}{2}}(p^2-2p-3))>0$; when $0\in B_-(f)$, we have $A_1^{\bot}=A_2^{\bot}=0$ and $A_3^{\bot}=\frac{1}{48}(p-1)((3p^2-6p-1)p^{n+s-2}-(4p-8)(p^{\frac{3n+s-3}{2}}+p^{n-1}-p^{\frac{n+s-1}{2} })+(p-1)^2p^{2n-3}+2k(p^2-2p-3)p^{\frac{n+3s-5}{2}})>0$  except for $p=3$ and $n=3$, so we get that the minimum distance of ${\widetilde{\mathcal{C}}_{D_{f,sq}}}^{\bot}$ is $3$. For linear code ${\widetilde{\mathcal{C}}_{D_{f,nsq}}}^{\bot}$, when $0\in B_+(f)$, we have $A_1^{\bot}=A_2^{\bot}=0$ and $A_3^{\bot}=\frac{1}{48}(p-1)((3p^2-6p-1)p^{n+s-2}+(p-1)^2p^{2n-3}-(4p-8)(p^{n-1}-p^{\frac{n+s-1}{2}})-2k(p^2-2p-3)p^{\frac{n+3s-5}{2}}-(2p^2-4p+6)p^{\frac{3n+s-5}{2}})>0$ except for $p=3$ and $n=3$; when $0\in B_-(f)$, we have $A_1^{\bot}=A_2^{\bot}=0$ and $A_3^{\bot}=\frac{1}{48}(p-1)((3p^2-6p-1)p^{n+s-2}+(4p-8)(p^{\frac{3n+s-3}{2}}-p^{\frac{n+s-1}{2}}-p^{n-1})+(p-1)^2p^{2n-3}-2kp^{\frac{n+3s-5}{2}}(p^2-2p-3))>0$, so we get that the minimum distance of ${\widetilde{\mathcal{C}}_{D_{f,nsq}}}^{\bot}$is $3$. \end{proof}
	\begin{remark}
		$(1)$	According to the sphere packing bound, when $n+s$ is even, if $0\in B_+(f)$ and $p>3$, or $0\in B_-(f)$, then for fixed lengths and minimum distances, ${\widetilde{\mathcal{C}}_{D_{f,sq}}}^{\bot}$ and ${\widetilde{\mathcal{C}}_{D_{f,nsq}}}^{\bot}$  are optimal; if $0\in B_+(f)$ and $p=3$, then for fixed lengths and dimensions, ${\widetilde{\mathcal{C}}_{D_{f,sq}}}^{\bot}$ and ${\widetilde{\mathcal{C}}_{D_{f,nsq}}}^{\bot}$  are at least almost optimal. When $n+s$ is odd, if $0\in B_-(f)$ (respectively $0\in B_+(f)$) and $p>3$, or $0\in B_+(f)$ (respectively $0\in B_-(f)$), then for fixed length and minimum distance, ${\widetilde{\mathcal{C}}_{D_{f,sq}}}^{\bot}$ (respectively ${\widetilde{\mathcal{C}}_{D_{f,nsq}}}^{\bot}$) is optimal; if $0\in B_-(f)$ (respectively $0\in B_+(f)$) and $p=3$, then for fixed length and dimension, ${\widetilde{\mathcal{C}}_{D_{f,sq}}}^{\bot}$ (respectively ${\widetilde{\mathcal{C}}_{D_{f,nsq}}}^{\bot}$) is at least almost optimal. \\
		$(2)$ If $k=0$ or $k=p^{n-s}$, then $f(x)$ is weakly regular, so the linear codes in \cite{Mesnager3} can be obtained by Theorems \ref{Th 3}, \ref{Th 4}, \ref{Th 5} and \ref{Th 6}.\\
		$(3)$  According to \cite[Lemma 2.1]{Ashikhmin}, we have that if $n+s$ is even with $0\le s\le n-6$ or $n+s$ is odd with $0\le s\le n-5$, then the linear codes $\mathcal{C}_{D_{f,0}\setminus\{0\}}$, $\mathcal{C}_{D_{f,sq}}$, $\mathcal{C}_{D_{f,nsq}}$, $\widetilde{\mathcal{C}}_{D_{f,0}\setminus\{0\}}$, $\widetilde{\mathcal{C}}_{D_{f,sq}}$ and $\widetilde{\mathcal{C}}_{D_{f,nsq}}$ constructed in Theorems \ref{Th 3}, \ref{Th 4}, \ref{Th 5} and \ref{Th 6} are minimal.
	\end{remark}
	
	In the following, we give two examples by Magma program to  verify the results given in this section.
	\begin{example}
		Let $f(x):\mathbb{F}_5^6\longrightarrow \mathbb{F}_5$, $f(x_1,x_2,x_3,x_4,x_5,x_6)=x_1^2x_4^4+4x_1x_2^3x_4^4+x_1x_2^3+2x_2^2x_4^4+x_3x_4$.
		Then
		\begin{itemize}
			\item[$\bullet$] $f(x)$ is a non-weakly regular $2$-plateaued function of type $(+)$.
			\item[$\bullet$] $B_+(f)=\mathbb{F}_5^2\times\{0\}\times \mathbb{F}_5\times (0,0)$ and  $B_-(f)=\mathbb{F}_5^2\times \mathbb{F}_5^*\times \mathbb{F}_5\times (0,0).$
			\item[$\bullet$] $f^*(x_1,x_2,x_3,x_4,0,0)=x_1^3x_2x_3^4+4x_1^3x_2+x_1^2x_3^4+3x_2^2x_3^
			4+4x_3x_4$, $B_+(f^*)=\mathbb{F}_5^2\times\mathbb{F}_5\times\{0\}\times\mathbb{F}_5^2$, $B_-(f^*)=\mathbb{F}_5^2\times\mathbb{F}_5\times\mathbb{F}_5^*\times\mathbb{F}_5^2$ and $t=4$, $t'=2$.
			\item[$\bullet$] $\mathcal{C}_f$ is a $[15624,7,12000]$ linear code with weight enumerator $1+180z^{12000}+1600z^{12375}+75624z^{12500}+320z^{12625}+400z^{13000}$ and $\mathcal{C}_f^{\bot}$ is a $[15624,15617,2]$ linear code. $\widetilde{\mathcal{C}}_f$ is a $[12000,7,9200]$ linear code with weight enumerator $1+100z^{9200}+256z^{9500}+1600z^{9575}+75000z^{9600}+320z^{9625}+800z^{9700}+44z^{10000}+4z^{12000}$ and ${\widetilde{\mathcal{C}}_f}^{\bot}$ is a $[12000,11993,3]$ linear code. $\mathcal{C}_{D_{f,0}\setminus\{0\}}$ is a $[3624,6,2500]$ linear code with weight enumerator $1+44z^{2500}+400z^{2800}+15000z^{2900}+80z^{3000}+100z^{3300}$ and $\mathcal{C}_{D_{f,0}\setminus\{0\}}^\bot$ is a $[3624,3618,2]$ linear code. $\mathcal{C}_{D_{f,sq}}$ and $\mathcal{C}_{D_{f,nsq}}$ are $[6000,6,4500]$ linear codes with weight enumerators $1+40z^{4500}+300z^{4600}+15000z^{4800}+84z^{5000}+200z^{5100}$, and $\mathcal{C}_{D_{f,sq}}^{\bot}$ and $\mathcal{C}_{D_{f,nsq}}^{\bot}$ are $[6000,5994,2]$ linear codes. $\widetilde{\mathcal{C}}_{D_{f,0}\setminus\{0\}}$ is a $[906,6,625]$ linear code with weight enumerator $1+44z^{625}+400z^{700}+15000z^{725}+80z^{750}+100z^{825}$ and ${\widetilde{\mathcal{C}}_{D_{f,0}\setminus\{0\}}}^{\bot}$ is a $[906,900,3]$ linear code. $\widetilde{\mathcal{C}}_{D_{f,sq}}$ and $\widetilde{\mathcal{C}}_{D_{f,nsq}}$ are $[1500,6,1125]$ linear codes with weight enumerators $1+40z^{1125}+300z^{1150}+15000z^{1200}+84z^{1250}+200z^{1275}$, and ${\widetilde{\mathcal{C}}_{D_{f,sq}}}^{\bot}$ and ${\widetilde{\mathcal{C}}_{D_{f,nsq}}}^{\bot}$ are $[1500,1494,3]$ linear codes.   
		\end{itemize}
	\end{example}
	\begin{example}
		Let $f(x):\mathbb{F}_3^6\longrightarrow \mathbb{F}_3$, $f(x_1,x_2,x_3,x_4,x_5,x_6)=x_1^2x_5^2+x_1^2+x_2^2+x_3^2+x_4x_5$. Then
		\begin{itemize}
			\item [$\bullet$] $f(x)$ is a non-weakly regular $1$-plateaued  function of type $(-)$.
			\item[$\bullet$] $B_+(f)=\mathbb{F}_3^3\times\mathbb{F}_3^*\times\mathbb{F}_3\times\{0\}$ and $B_-(f)=\mathbb{F}_3^3\times\{0\}\times\mathbb{F}_3\times\{0\}$.
			\item[$\bullet$] $f^*(x_1,x_2,x_3,x_4,x_5,0)=2x_1^2x_4^2+2x_1^2+2x_2^2+2x_3^2+2x_4x_5$, $B_+(f^*)=\mathbb{F}_3^3\times\mathbb{F}_3\times\{0\}\times\mathbb{F}_3$, $B_-(f^*)=\mathbb{F}_3^3\times\mathbb{F}_3\times\mathbb{F}_3^*\times\mathbb{F}_3$ and $t=t'=2$.
			\item[$\bullet$] $\mathcal{C}_f$ is a $[728,7,459]$ linear code with weight enumerator $1+180z^{459}+1862z^{486}+144z^{513}$ and $\mathcal{C}_f^{\bot}$ is a $[728,721,2]$ linear code. $\widetilde{\mathcal{C}}_f$ is a $[486,7,306]$ linear code with weight enumerator $1+72z^{306}+180z^{315}+1698z^{324}+144z^{333}+90z^{342}+2z^{486}$ and ${\widetilde{\mathcal{C}}_f}^{\bot}$ is a $[486,479,3]$ linear code. $\mathcal{C}_{D_{f,0}\setminus\{0\}}$ is a $[242, 6, 144]$ linear code with weight enumerator $1+90z^{144}+566z^{162}+72z^{180}$ and  $\mathcal{C}_{D_{f,0}\setminus\{0\}}^{\bot}$ is a $[242,236,2]$ linear code. $\mathcal{C}_{D_{f,sq}}$ is a $[216,6,126]$ linear code with weight enumerator $1+72z^{126}+576z^{144}+80z^{162}$ and $\mathcal{C}_{D_{f,sq}}^{\bot}$ is a $[216,210,2]$ linear code. $\mathcal{C}_{D_{f,nsq}}$ is a $[270,6,162]$ linear code with weight enumerator $1+80z^{162}+558z^{180}+90z^{198}$ and  $\mathcal{C}_{D_{f,nsq}}^{\bot}$ is a $[270,264,2]$ linear code. $\widetilde{\mathcal{C}}_{D_{f,0}\setminus\{0\}}$ is a $[121,6,72]$ linear code with weight enumerator $1+90z^{72}+566z^{81}+72z^{90}$ and ${\widetilde{\mathcal{C}}_{D_{f,0}\setminus\{0\}}}^{\bot}$ is a $[121,115,3]$ linear code, which is optimal according to the Code Table at http://www.codetables.de/. $\widetilde{\mathcal{C}}_{D_{f,sq}}$ is a $[108,6,63]$ linear code with weight enumerator $1+72z^{63}+576
			z^{72}+80z^{81}$ and ${\widetilde{\mathcal{C}}_{D_{f,sq}}}^{\bot}$ is a $[108,102,3]$ linear code,  which is optimal according to the Code Table at http://www.codetables.de/. $\widetilde{\mathcal{C}}_{D_{f,nsq}}$ is a $[135,6,81]$ linear code with weight enumerator $1+80z^{81}+558z^{90}+90z^{99}$ and ${\widetilde{\mathcal{C}}_{D_{f,nsq}}}^{\bot}$ is a $[135,129,3]$ linear code,  which is optimal according to the Code Table at http://www.codetables.de/. 
			\begin{remark}
				In Table \uppercase\expandafter{\romannumeral15} in the Appendix, we give some linear codes constructed by Theorems \ref{Th 2}, \ref{Th 5} and \ref{Th 6}, which are optimal according to the Code Table at http://www.codetables.de/.
			\end{remark}

		\end{itemize}
	\end{example}
	
	\section{Quantum codes and LCD codes from self-orthogonal codes}
	
	In this section, we will construct some quantum codes and LCD codes from the self-orthogonal codes given in Section 4.
	\subsection{Quantum codes from self-orthogonal codes}
	In the following theorem, we give some pure quantum codes constructed by $\widetilde{\mathcal{C}}_f$. 
	\begin{theorem}\label{Th 7}
		Let $f(x): V_n^{(p)}\longrightarrow \mathbb{F}_p$ be an $s$-plateaued function belonging to $\mathscr{F}$. When $n+s$ is an even integer with $0\le s\le n-2$, $n+s\ge 6$ for $p=3$ and $n+s\ge 4$ for $p\ne 3$, there exists a family of $[[(p-1)(p^{n-1}-\epsilon_0 p^{\frac{n+s}{2}-1}), (p-1)(p^{n-1}-\epsilon_0 p^{\frac{n+s}{2}-1})-n-2, 3]]_p$ pure quantum codes, which are at least almost optimal according to the quantum Hamming bound. When $n+s$ is an odd integer with $0\le s\le n-1$, $n+s\ge 5$ for $p=3$ and $n+s\ge 3$ for $p\ne 3$, there exists a family of $[[(p-1)p^{n-1}, (p-1)p^{n-1}-n-2, 3]]_p$ pure quantum codes, which are at least almost optimal according to the quantum Hamming bound.
	\end{theorem}
	\begin{proof}
	When $n+s$ is even, by Theorem \ref{Th 2}, we have that ${\widetilde{\mathcal{C}}_f}^{\bot}$ is a $[(p-1)(p^{n-1}-\epsilon_0 p^{\frac{n+s}{2}-1}), (p-1)(p^{n-1}-\epsilon_0 p^{\frac{n+s}{2}-1})-n-1, 3]$ linear code and $\widetilde{\mathcal{C}}_f\subseteq {\widetilde{\mathcal{C}}_f}^{\bot}$. Note that the vector {\small \[\alpha=(\underbrace{1,\cdots,1,}_{N_1(f)\ \text{times}}\underbrace{2, \cdots, 2,}_{N_2(f)\ \text{times}}\underbrace{3,\cdots,3,}_{N_3(f)\ \text{times}}\cdots,\underbrace{p-1,\cdots,p-1}_{N_{p-1}(f)\ \text{times}})\in \widetilde{\mathcal{C}}_f.\]} Let $\mathcal{C}_1={\widetilde{\mathcal{C}}_f}^{\bot}$ and  $\mathcal{C}_2$ be the dual of the code $\{c\alpha:c\in \mathbb{F}_p\}$. We easily get that $\mathcal{C}_2$ is a $[(p-1)(p^{n-1}-\epsilon_0 p^{\frac{n+s}{2}-1}), (p-1)(p^{n-1}-\epsilon_0 p^{\frac{n+s}{2}-1})-1,2]$ linear code. Note that $\mathcal{C}_1^{\bot}\subseteq \mathcal{C}_1\subseteq\mathcal{C}_2$, according to Proposition \ref{Po 5}, there exists a family of $[[(p-1)(p^{n-1}-\epsilon_0 p^{\frac{n+s}{2}-1}), (p-1)(p^{n-1}-\epsilon_0 p^{\frac{n+s}{2}-1})-n-2, 3]]_p$ pure quantum codes. By the quantum Hamming bound, for fixed lengths and dimensions, we have that those quantum codes are at least almost optimal. 
	
	When $n+s$ is odd, the proof is similar, so we omit it.\end{proof}
	
	In the following, we consider the $s$-plateaued function $f:V_{n}^{(p)}=\mathbb{F}_p^{n_1}\times\mathbb{F}_p^{n_2}\times\mathbb{F}_p^{n_2}\times\mathbb{F}_p^s\longrightarrow\mathbb{F}_p$ defined by 
	\begin{equation}
		f(x,y,z,u)=f^{(z)}(x)+\sum\limits_{i=1}^{n_2}y_iz_i,
	\end{equation}
	where $y=(y_1,y_2,\dots, y_{n_2})$, $z=(z_1,z_2,\dots, z_{n_2})$, and $f^{(z)}(x)$ is a $2$-form weakly regular bent function satisfying $f^{{(cz)}}(x)=f^{(z)}(x)$ and $f^{(z)}(0)=0$ for any $z\in\mathbb{F}_p^{n_2}$, $c\in\mathbb{F}_p^*$. Thus, $f(x,y,z,u)\in\mathscr{F}$ and $t=t'=2$. From the punctured codes $\widetilde{\mathcal{C}}_{D_{f,0}\setminus\{0\}}$, $\widetilde{\mathcal{C}}_{D_{f,sq}}$ and $\widetilde{\mathcal{C}}_{D_{f,nsq}}$, we construct some pure quantum codes using the plateaued functions defined by (10).
	\begin{theorem}\label{Th 8}
		When $n+s$ is an even integer with $0\le s\le n-4$ and  $n+s\ge 6$, there exists a family of $[[\frac{3^{n-1}+2\epsilon_0\times3^{\frac{n+s}{2}-1}-1}{2}, \frac{3^{n-1}+2\epsilon_0\times3^{\frac{n+s}{2}-1}-1}{2}-2n+2, 3]]_3$ pure quantum codes, which are at least almost optimal according to the quantum Hamming bound.
		When $n+s$ is an odd integer with $0\le s\le n-3$ and $n+s\ge 5$, there exists a family of $[[\frac{3^{n-1}-1}{2}, \frac{3^{n-1}-1}{2}-2n+2, 3]]_3$ pure quantum codes, which are at least almost optimal according to the quantum Hamming bound.
		
	\end{theorem}
	\begin{proof} When $n+s$ is even, according to Theorem \ref{Th 5} and Proposition \ref{Po 6}, we have that for $p=3$, ${\widetilde{\mathcal{C}}_{D_{f,0}\setminus\{0\}}}^{\bot}$ is a $[\frac{3^{n-1}+2\epsilon_0\times 3^{\frac{n+s}{2}-1}-1}{2},\frac{3^{n-1}+2\epsilon_0\times 3^{\frac{n+s}{2}-1}-1}{2}-n,3]$ linear code and $\widetilde{\mathcal{C}}_{D_{f,0}\setminus\{0\}}\subseteq {\widetilde{\mathcal{C}}_{D_{f,0}\setminus\{0\}}}^{\bot}$.
	Now, we construct a subcode $\mathcal{C}$ of $\widetilde{\mathcal{C}}_{D_{f,0}\setminus\{0\}}$ such that the minimum distance of $\mathcal{C}^{\bot}$, denoted by $d^{\bot}$, is at least 2.
	
	  According to the Walsh transform of the quadratic bent function given in \cite[Theorem 1]{Ozbudak}, let $f^{(0)}(x)=vx_1^2+x_2^2+\cdots+x_{n_1}^2$,  $f^{(z)}(x)=wx_1^2+x_2^2+\cdots+x_{n_1}^2$ for any $z\in\mathbb{F}_3^{n_2}\setminus\{0\}$,  $\alpha=e_1$ and $\beta=e_{n_1+1}-cve_{n_1+n_2+1}$, where $v,w,c \in \mathbb{F}_3^*$.
	Let $H$ be the linear space over $\mathbb{F}_3$ spanned by $\alpha$ and $\beta$. If $vw\in SQ$, let $c\in NSQ$; if $vw\in NSQ$, let $c\in SQ$. For any nonzero vector $\gamma=k_1\alpha+k_2\beta$, where  $k_1,k_2\in\mathbb{F}_3$, we have that $f(\gamma)\ne 0$. Define $\mathcal{C}=\{(\langle a,x_1\rangle_n, \langle a,x_2\rangle_n,\cdots, \langle a,x_m\rangle_n):  a \in H^{\bot}\},$ where $\{x_1,x_2,\cdots,x_m\}=\widetilde{D}_{f,0}\setminus\{0\}$.  Obviously, $\mathcal{C}$ is a $[\frac{3^{n-1}+2\epsilon_0\times 3^{\frac{n+s}{2}-1}-1}{2},n-2]$ linear code. Assume that the minimum distance of $\mathcal{C}^{\bot}$ is $1$, then there exist $\tau\in\mathbb{F}_3^*$ and $x\in\widetilde{D}_{f,0}\setminus\{0\}$ such that $\langle a,\tau x\rangle_n=0$ for any $a\in H^{\bot}$. Then $x\in (H^{\bot})^{\bot}=H$, which implies a contradiction. Thus, $d^{\bot}\ge 2$. Note that $\mathcal{C}^{\bot}$ is a $[\frac{3^{n-1}+2\epsilon_0\times 3^{\frac{n+s}{2}-1}-1}{2},\frac{3^{n-1}+2\epsilon_0\times 3^{\frac{n+s}{2}-1}-1}{2}-n+2]$ linear code, and $\widetilde{\mathcal{C}}_{D_{f,0}\setminus\{0\}}\subseteq {\widetilde{\mathcal{C}}_{D_{f,0}\setminus\{0\}}}^{\bot}\subseteq \mathcal{C}^{\bot}$. By Proposition \ref{Po 5}, there exists a family of $[[\frac{3^{n-1}+2\epsilon_0\times 3^{\frac{n+s}{2}-1}-1}{2}, \frac{3^{n-1}+2\epsilon_0\times 3^{\frac{n+s}{2}-1}-1}{2}-2n+2, 3]]_3$ pure quantum codes. According to the quantum Hamming bound, for fixed lengths and dimensions, those quantum codes are at least almost optimal. 
	
	When $n+s$ is odd, the proof is similar, so we omit it.\end{proof}
	
	\begin{theorem}\label{Th 9}
		Let $n$, $s$ be non-negative integers with $s\ge 1$ for $p=3$.
		When $n=4+s$, there exists a family of $[[\frac{p^{n-1}-\epsilon_0p^{n-3}}{2}, \frac{p^{n-1}-\epsilon_0p^{n-3}}{2}-n-2,3]]_p$ pure quantum codes, which are at least almost optimal according to the quantum Hamming bound. When $n=3+s$, there exists a family of $[[\frac{p^{n-1}+\epsilon_0p^{n-2}}{2}, \frac{p^{n-1}+\epsilon_0p^{n-2}}{2}-n-2,3]]_p$ pure quantum codes, which are at least almost optimal according to the quantum Hamming bound.
	\end{theorem}
	\begin{proof}
 According to the Walsh transform of the quadratic bent function given in \cite[Theorem 1]{Ozbudak}, when $n=4+s$, let $n_1=2$, $n_2=1$, $f^{(0)}(x)=v_1x_1^2+v_2x_2^2$ and $f^{(z)}(x)=w_1x_1^2+w_2x_2^2$ for any $z\in\mathbb{F}_{p}^{n_2}\setminus\{0\}$, where $v_1,w_1\in NSQ$ and $v_2,w_2\in\mathbb{F}_p^*$. Let $\alpha_1=e_1, \alpha_2=e_4, \alpha_3=e_5, \alpha_4=e_6, \dots,\alpha_{2+s}=e_n$, $H$ be the linear space over $\mathbb{F}_p$ spanned by $\alpha_1, \dots, \alpha_{2+s}$. For any nonzero vector $\gamma=k_1\alpha_1+k_2\alpha_2+\cdots+k_{2+s}\alpha_{2+s}$, where $k_1,k_2,\dots, k_{2+s}\in\mathbb{F}_p$, we have that $f(\gamma)\notin SQ$. Define $\mathcal{C}=\{\langle a,x\rangle_n: x\in \widetilde{D}_{f,i}, a\in H^{\bot}\}$, where $i\in SQ$. By Theorem \ref{Th 6}, Proposition \ref{Po 7} and the proof of Theorem \ref{Th 8}, we easily obtain the result.
	
	When $n=3+s$, let $n_1=1$, $n_2=1$. Let $\alpha_1=e_3, \alpha_2=e_4, \alpha_3=e_5, \dots,\alpha_{1+s}=e_n$, $H$ be the linear space over $\mathbb{F}_p$ spanned by $\alpha_1, \dots, \alpha_{1+s}$. For any nonzero vector $\gamma=k_1\alpha_1+k_2\alpha_2+\cdots+k_{1+s}\alpha_{1+s}$, where $k_1,k_2,\dots, k_{1+s}\in\mathbb{F}_p$, we have that $f(\gamma)=0$. Define $\mathcal{C}=\{(\langle a,x_1\rangle_n, \langle a,x_2\rangle_n,\cdots, \langle a,x_m\rangle_n):  a \in H^{\bot}\},$ where $\{x_1,x_2,\cdots,x_m\}=\widetilde{D}_{f,i},\  i\in SQ$. By Theorem \ref{Th 6}, Proposition \ref{Po 7} and  the proof of Theorem \ref{Th 8}, we easily obtain the result. \end{proof}
	\begin{remark}
 $(1)$ When $n=4+s$, we can also consider using  $\mathcal{C}=\{(\langle a,x_1\rangle_n, \langle a,x_2\rangle_n,\cdots, \langle a,x_m\rangle_n):  a \in H^{\bot}\},$ where $\{x_1,x_2,\cdots,x_m\}=\widetilde{D}_{f,i},\  i\in NSQ$, to construct quantum codes. In this case, we just need to make $v_1,w_1\in SQ$. When $n=3+s$, the linear code $\mathcal{C}$ can also be defined as $\mathcal{C}=\{(\langle a,x_1\rangle_n, \langle a,x_2\rangle_n,\cdots, \langle a,x_m\rangle_n):  a \in H^{\bot}\},$ where $\{x_1,x_2,\cdots,x_m\}=\widetilde{D}_{f,i},\  i\in NSQ$.\\
		$(2)$ When $p=3$ and $p=5$, the parameters of our quantum codes are contained in \cite{Chen,Liang}. In Table \uppercase\expandafter{\romannumeral16} in the Appendix, we present some $7$-ary quantum codes constructed by Theorems \ref{Th 7} and \ref{Th 9}, which have higher rates than the known ones in \cite{Edel}.
	\end{remark}
	\subsection{LCD codes from self-orthogonal codes}
	Let $I_{n+1,n+1}$ be the identity matrix of size $(n+1)\times (n+1)$. In the following, we give some LCD codes constructed by $\widetilde{\mathcal{C}}_f$.
	\begin{theorem}\label{Th 10}
		Let $f(x): V_n^{(p)}\longrightarrow \mathbb{F}_p$ be an $s$-plateaued function belonging to $\mathscr{F}$ and $G$ be a generator matrix of $\widetilde{\mathcal{C}}_f
		$. When $n+s$ is an even integer with $0\le s\le n-2$, $n+s\ge 6$ for $p=3$ and $n+s\ge 4$ for $p\ne 3$, $G'=[I_{n+1,n+1}, G]$ generates a $p$-ary  $[(p-1)(p^{n-1}-\epsilon_0p^{\frac{n+s}{2}-1})+n+1, n+1]$ LCD code $\mathcal{C}$, and its dual code $\mathcal{C}^{\bot}$ is a $p$-ary  $[(p-1)(p^{n-1}-\epsilon_0p^{\frac{n+s}{2}-1})+n+1, (p-1)(p^{n-1}-\epsilon_0p^{\frac{n+s}{2}-1}),3]$ LCD code. When  $n+s$ is an odd integer with $0\le s\le n-1$, $n+s\ge 5$ for $p=3$ and $n+s\ge 3$ for $p\ne 3$, $G'=[I_{n+1,n+1}, G]$ generates a $p$-ary $[(p-1)p^{n-1}+n+1, n+1]$ LCD code $\mathcal{C}$, and its dual code $\mathcal{C}^{\bot}$ is a $p$-ary  $[(p-1)p^{n-1}+n+1, (p-1)p^{n-1},3]$ LCD code.
	\end{theorem}
	\begin{proof} When $n+s$ is even, by Theorem \ref{Th 2}, we have that $\widetilde{\mathcal{C}}_f$ is self-orthogonal. Thus, $G'{G'}^{T}=I_{n+1,n+1}$. According to Proposition \ref{Po 4}, we have that the linear code $\mathcal{C}$ generated by $G'$ is an LCD code, and its dual code $\mathcal{C}^{\bot}$ is also an LCD code. The lengths and dimensions of $\mathcal{C}$ and $\mathcal{C}^{\bot}$ directly follow from Theorem \ref{Th 2}. Now, we prove that the minimum distance of $\mathcal{C}^{\bot}$, denoted by $d^{\bot}$, is $3$. Let $\alpha_1, \alpha_2,\cdots,\alpha_n$ be a basis of $V_{n}^{(p)}$, 
	{\small 	\[\beta=(\underbrace{1,\cdots,1,}_{N_1(f)\ \text{times}}\underbrace{2, \cdots, 2,}_{N_2(f)\ \text{times}}\underbrace{3,\cdots,3,}_{N_3(f)\ \text{times}}\cdots,\underbrace{p-1,\cdots,p-1}_{N_{p-1}(f)\ \text{times}}),\]}
	then {\small $G'=\begin{bmatrix}
			1&0&\cdots&0&\beta\\
			0&1&\cdots&0&(\langle \alpha_1,x\rangle_n)_{x\in V_{n}^{(p)}\setminus D_{f,0}}\\
			\vdots&\vdots&\ddots&0&	\vdots\\
			0&0&\cdots&1&	(\langle \alpha_n,x\rangle_n)_{x\in V_{n}^{(p)}\setminus D_{f,0}}
		\end{bmatrix}.$}
	
	Since the minimum distance of ${\widetilde{\mathcal{C}}_f}^{\bot}$ is 3, then $d^{\bot}\le 3$. Thus, we can prove that $d^{\bot}\ge 3$, i.e., any two columns of $G'$ are linearly independent. As easily seen, we need to prove that for any $i\in \mathbb{F}_p^*$ and $x\in D_{f,i}$, $(1,0,\cdots,0)$ and $(i,\langle \alpha_1, x\rangle_n,\cdots,\langle\alpha_n, x\rangle _n)$  are linearly independent. If there exist $k_1,$ $ k_2\in\mathbb{F}_p$ such that $k_1(1,0,\cdots,0)+k_2(i,\langle \alpha_1, x\rangle_n, \cdots, \langle\alpha_n, x\rangle _n)=0$, then $k_1+k_2 i=0$ and $k_2\langle\alpha_j,x\rangle_n=0$ for any $1\le j\le n$. Since $x\ne 0$, then $k_1=k_2=0$. Thus, $d^{\bot}=3$.
	
	When $n+s$ is odd, the proof is similar, so we omit it.
	\end{proof}
	\begin{remark}
		According to the sphere packing bound, when $n+s$ is even, if $p=3$, $0\in B_+(f)$ and $n-s=2$, for fixed length and dimension, $\mathcal{C}^{\bot}$ is at least almost optimal; in other cases, for fixed length and minimum distance, $\mathcal{C}^{\bot}$ is optimal. When $n+s$ is odd, for fixed length and minimum distance, $\mathcal{C}^{\bot}$ is optimal.
	\end{remark}
	
	In the following, we give some LCD codes constructed by $\widetilde{\mathcal{C}}_{D_{f,0}\setminus\{0\}}$, $\widetilde{\mathcal{C}}_{D_{f,sq}}$ and $\widetilde{\mathcal{C}}_{D_{f,nsq}}$. Firstly, we present a lemma.
	\begin{lemma}\label{Le 16}
		Let $\alpha_1, \alpha_2,\cdots, \alpha_n$ be a basis of $V_{n}^{(p)}$, then the matrix
		{\small \begin{equation}H=\begin{bmatrix}
					H_1\\
					H_2\\
					\vdots\\
					H_n
				\end{bmatrix}=\begin{bmatrix}
					\langle \alpha_1,\alpha_1\rangle_n& \langle \alpha_1,\alpha_2\rangle _n&\cdots& \langle\alpha_1,\alpha_n\rangle_n\\
					\langle \alpha_2,\alpha_1\rangle_n& \langle \alpha_2,\alpha_2\rangle _n&\cdots& \langle\alpha_2,\alpha_n\rangle_n\\
					\vdots&\vdots&\ddots&\vdots\\
					\langle \alpha_n,\alpha_1\rangle_n& \langle \alpha_n,\alpha_2\rangle _n&\cdots& \langle\alpha_n,\alpha_n\rangle_n
				\end{bmatrix}
		\end{equation}}
		is nonsingular.
	\end{lemma}
	\begin{proof} We only need to prove that the row vectors of $H$ are linearly independent. Assume that there exist $k_1,k_2,\cdots, k_n\in\mathbb{F}_p$ such that $\sum_{i=1}^nk_i H_i=0$, i.e., $\sum_{i=1}^nk_i\langle \alpha_i,\alpha_j\rangle_n=0$ for any $1\le j\le n$. Thus, $\sum_{i=1}^n{k_i}\alpha_i=0$, which implies that $k_i=0$ for any $1\le i\le n$. This completes the proof. \end{proof}
	
	In the sequel,  according to the Walsh transform of the quadratic bent function given in \cite[Theorem 1]{Ozbudak}, let $f(x,y,z,u)$ be an $s$-plateaued function defined by (10) satisfying $f^{(0)}(x)=u_1x_1^2+u_2x_2^2+\cdots+u_{n_1}x_{n_1}^2$ and $f^{(z)}(x)=v_1x_1^2+v_2x_2^2+\cdots+v_{n_1}x_{n_1}^2$ for any $z\in\mathbb{F}_p^{n_2}\setminus\{0\}$, where $u_1,\dots, u_{n_1}, v_1,\dots,v_{n_1}\in\mathbb{F}_p^*$. \hspace{0.5cm}   $(*)$

	Let $H^{(1)}$ be the matrix obtained by Equation (11), where $\alpha_1=e_1, \alpha_2=e_2, \dots, \alpha_{n_1}=e_{n_1}, \alpha_{n_1+1}=e_1+e_{n_1+1}, \alpha_{n_1+2}=e_1+e_{n_1+2},\dots, \alpha_{n}=e_1+e_n$. In the following theorem, we construct some LCD codes from 
	$\widetilde{\mathcal{C}}_{D_{f,0}\setminus\{0\}}$.
	\begin{theorem}\label{Th 11}
		Let $p=3$, $f(x,y,z,u)$ be an $s$-plateaued function given by $(*)$. Let $G$ be the generator matrix of $\widetilde{\mathcal{C}}_{D_{f,0}\setminus\{0\}}$ given by
		
		{\small \[G=\begin{bmatrix}
				(\langle \alpha_1,x\rangle_n)_{x\in\widetilde{D}_{f,0}\setminus\{0\}}\\
				(\langle \alpha_2,x\rangle_n)_{x\in\widetilde{D}_{f,0}\setminus\{0\}}\\
				\vdots\\	
				(\langle \alpha_n,x\rangle_n)_{x\in\widetilde{D}_{f,0}\setminus\{0\}}
			\end{bmatrix},\]}
		where $\alpha_1,\alpha_2,\dots,\alpha_n$ are given above. When $n+s$ is even with $0\le s\le n-4$ and $n+s\ge 6$, $G'=[H^{(1)},G]$ generates a ternary $[\frac{3^{n-1}+2\epsilon_0\times 3^{\frac{n+s}{2}-1}-1}{2}+n, n]$ LCD code $\mathcal{C}$, and its  dual code $\mathcal{C}^{\bot}$ is a ternary $[\frac{3^{n-1}+2\epsilon_0\times 3^{\frac{n+s}{2}-1}-1}{2}+n, \frac{3^{n-1}+2\epsilon_0\times 3^{\frac{n+s}{2}-1}-1}{2},3]$ LCD code. When $n+s$ is odd with $0\le s\le n-3$ and $n+s\ge 5$, $G'=[H^{(1)},G]$ generates a ternary $[\frac{3^{n-1}-1}{2}+n, n]$ LCD  code $\mathcal{C}$, and its dual code $\mathcal{C}^{\bot}$ is a ternary $[\frac{3^{n-1}-1}{2}+n, \frac{3^{n-1}-1}{2},3]$ LCD code.
	\end{theorem}
	\begin{proof} When $n+s$ is even, by Theorem \ref{Th 5}, for $p=3$, we have that  $\widetilde{\mathcal{C}}_{D_{f,0}\setminus\{0\}}$ is self-orthogonal. By Lemma \ref{Le 16}, we know that $H^{(1)}$ is nonsingular. Thus, $G'G'^{T}=H^{(1)}(H^{(1)})^T$ is nonsingular. According to Proposition \ref{Po 4}, we have that the linear code $\mathcal{C}$ generated by $G'$ is an LCD code, and its dual code  $\mathcal{C}^{\bot}$ is also an LCD code. The lengths and dimensions of $\mathcal{C}$ and $\mathcal{C}^{\bot}$ directly follow from Theorem \ref{Th 5}. Now, we prove that the minimum distance of $\mathcal{C}^{\bot}$, denoted by $d^{\bot}$, is 3.
	
	Since the minimum distance of ${\widetilde{\mathcal{C}}_{D_{f,0}\setminus\{0\}}}^{\bot}$ is 3, then $d^{\bot}\le 3$. Thus, we can prove that $d^{\bot}\ge 3$, i.e., any two columns of $G'$ are linearly independent. As easily seen, we need to prove that for any $1\le j\le n$, $x\in\widetilde{D}_{f,0}\setminus\{0\}$, $(\langle \alpha_1,\alpha_j\rangle_n, \langle\alpha_2,\alpha_j\rangle_n, \dots, \langle \alpha_n, \alpha_j\rangle_n)$ and $(\langle \alpha_1,x\rangle_n, \langle\alpha_2,x\rangle_n, \dots, \langle \alpha_n, x\rangle_n)$ are linearly independent. If there exist $k_1, k_2\in\mathbb{F}_p$ such that $k_1(\langle \alpha_1,\alpha_j\rangle_n,  \dots, \langle \alpha_n, \alpha_j\rangle_n)+k_2(\langle \alpha_1,x\rangle_n,  \dots, \langle \alpha_n, x\rangle_n)=0$, then $\langle \alpha_i, k_1\alpha_j+k_2x\rangle_n=0$ for any $1\le i\le n$. Thus, $k_1\alpha_j+k_2x=0$, which implies that $k_1^2f(\alpha_j)=k_2^2f(x)$. Since for any $1\le j\le n$, $f(\alpha_j)\ne 0$, then $k_1=k_2=0$. Hence, $d^{\bot}=3$. 
	
	When $n+s$ is odd, the proof is similar, so we omit it.\end{proof}
	\begin{remark}
		According to the sphere packing bound, when $n+s$ is even, if $0\in B_+(f)$, for fixed length and minimum distance, $\mathcal{C}^{\bot}$ is optimal; if $0\in B_-(f)$, for fixed length and dimension, $\mathcal{C}^{\bot}$ is at least almost optimal. When $n+s$ is odd, for fixed length and dimension, $\mathcal{C}^{\bot}$ is at least almost optimal.
	\end{remark}
	Consider the $s$-plateaued function $f(x,y,z,u)$ given by $(*)$, if $n_1=1$, let $\alpha_1=e_1+e_2-v_1e_{2+n_2}, \alpha_2=e_2, \alpha_3=e_3, \dots, \alpha_n=e_n$. If $n_1\ge 2$, let $\alpha_1=w_1e_1, \alpha_2=w_2e_2, \dots,\alpha_{n_1}=w_{n_1}e_{n_1},\alpha_{n_1+1}=e_{n_1+1},\dots,\alpha_n=e_n$, where $w_1,w_2,\dots,w_{n_1}\in \mathbb{F}_p^*$, $w_iu_i\in NSQ$ for $1\le i\le n_1$. Let $H^{(2)}$ be the matrix obtained by Equation (11), where $\alpha_1, \alpha_2, \dots, \alpha_n$ are given above. In the following theorem, we construct some LCD codes from $\widetilde{\mathcal{C}}_{D_{f,sq}}$.
	\begin{theorem}\label{Th 12}
		Let $f(x,y,z,u)$ be an $s$-plateaued function given by $(*)$. Let $G$ be the generator matrix of $\widetilde{\mathcal{C}}_{D_{f,sq}}$ given by
		{\small \[G=\begin{bmatrix}
				(\langle \alpha_1,x\rangle_n)_{x\in\widetilde{D}_{f,sq}}\\
				(\langle \alpha_2,x\rangle_n)_{x\in\widetilde{D}_{f,sq}}\\
				\vdots\\	
				(\langle \alpha_n,x\rangle_n)_{x\in\widetilde{D}_{f,sq}}
			\end{bmatrix},\]}
		where $\alpha_1, \alpha_2, \dots, \alpha_n$ are given above.
		When $n+s$ is an even integer with $0\le s\le n-4$, let $n+s\ge 6$ for $p=3$ and $\widetilde{D}_{f,sq}=\widetilde{D}_{f,i}\ (i\in SQ)$ for $p\ne 3$, then $G'=[H^{(2)},G]$ generates a $p$-ary $[\frac{p^{n-1}-\epsilon_0p^{\frac{n+s}{2}-1}}{2}+n, n]$ LCD code $\mathcal{C}$, and its dual code $\mathcal{C}^{\bot}$ is a $p$-ary $[\frac{p^{n-1}-\epsilon_0p^{\frac{n+s}{2}-1}}{2}+n, \frac{p^{n-1}-\epsilon_0p^{\frac{n+s}{2}-1}}{2},3]$ LCD code. When $n+s$ is an odd integer with $0\le s\le n-3$, let $n+s\ge 5$ for $p=3$ and $\widetilde{D}_{f,sq}=\widetilde{D}_{f,i}\ (i\in SQ)$ for $p\ne 3$, then $G'=[H^{(2)},G]$ generates a $p$-ary $[\frac{p^{n-1}+\epsilon_0p^{\frac{n+s-1}{2}}}{2}+n, n]$ LCD  code $\mathcal{C}$, and its dual code $\mathcal{C}^{\bot}$ is a $p$-ary $[\frac{p^{n-1}+\epsilon_0p^{\frac{n+s-1}{2}}}{2}+n, \frac{p^{n-1}+\epsilon_0p^{\frac{n+s-1}{2}}}{2},3]$ LCD code.
	\end{theorem}
	\begin{proof}
	Note that if $n_1=1$, then $f(\alpha_i)=0$ for $1\le i\le n$, and if $n_1\ge 2$, then $f(\alpha_i)\notin SQ$ for $1\le i\le n$. Similar to the proof of Theorem \ref{Th 11}, we easily get the result.\end{proof}
	\begin{remark}
		$(1)$ According to the sphere packing bound,  when $n+s$ is even, if $0\in B_+(f)$ and $p>3$, or $0\in B_-(f)$, then for fixed length and minimum distance, $\mathcal{C}^{\bot}$ is optimal; if $0\in B_+(f)$ and $p=3$, then for fixed length and dimension, $\mathcal{C}^{\bot}$ is at least almost optimal. When $n+s$ is odd, if $0\in B_-(f)$ and $p>3$, or $0\in B_+(f)$, then for fixed length and minimum distance, $\mathcal{C}^{\bot}$ is optimal; if $0\in B_-(f)$ and $p=3$, for fixed length and dimension, $\mathcal{C}^{\bot}$ is at least almost optimal.\\
		$(2)$ We can also consider using $\widetilde{\mathcal{C}}_{D_{f,nsq}}$ to construct LCD codes, in this case, we just need to make $w_iu_i\in SQ$ for $1\le i\le n_1$ when $n_1\ge 2$.
	\end{remark}
	\begin{remark}
		In Table \uppercase\expandafter{\romannumeral17} in the Appendix, we present some ternary LCD codes constructed by Theorems \ref{Th 10}, \ref{Th 11} and \ref{Th 12}, which are optimal according to the Code Table at http://www.codetables.de/.
	\end{remark}
	\section{Conclusion}
	
	In this paper, we constructed several families of linear codes from plateaued functions. The parameters and weight distributions of them were completely determined. Under certain conditions, those codes were proved to be self-orthogonal. Furthermore, using the constructed self-orthogonal codes, we obtained some families of at least almost optimal quantum codes and optimal LCD codes. Compared with the self-orthogonal codes, quantum codes and LCD codes constructed from functions in \cite{Cakmak,Heng1,Heng2,Heng3,Mesnager4,Wang2,Wang3,Heng4,Xie}, ours have different parameters. Besides, we also obtained some minimal codes which can be used to construct secret sharing schemes.
To conclude this paper, some open problems and potential directions for future research are listed.

\noindent$\bullet$	We gave a construction of plateaued functions belonging to $\mathscr{F}$ in Section 2. A natural question is whether other constructions exist.\\
\noindent$\bullet$ In this paper, using $p$-ary $s$-plateaued functions, we constructed some $p$-ary LCD codes and quantum codes from self-orthogonal codes. Thus, it is worth to study whether our construction methods can be extended to vectorial plateaued functions.\\
\noindent $\bullet$  Locally recoverable codes (LRCs) play a significant role in distributed and cloud storage systems and are widely studied. In \cite{Heng2}, Heng \emph{et al.} obtained several infinite families of optimal or almost optimal LRCs from some functions over finite fields.  It is also worthwhile to investigate the locality of the constructed linear codes to determine whether they are optimal LRCs.
	
	\section*{Appendix}
	\vspace{-0.7cm}
	\begin{table}[h]
		\centering
		\caption{The weight distribution of $\mathcal{C}_f$ in Theorem 1 when $n+s$ is even}
		\renewcommand\arraystretch{0.74}	
		\resizebox{\textwidth}{!}{
			\begin{tabular}{|c|c|c|}
				\hline
				\scriptsize{Weight} & \scriptsize{Multiplicity  ($0\in B_+(f^*)$)}&\scriptsize{Multiplicity ($0\in B_-(f^*)$)}\\ \hline
				\scriptsize{0} & \scriptsize{1}&\scriptsize{1} \\ \hline
				\scriptsize{	$(p-1)p^{n-1}$}&\scriptsize{$p^{n+1}-(p-1)p^{n-s}-1$}&\scriptsize{$p^{n+1}-(p-1)p^{n-s}-1$}\\ \hline
				\scriptsize{	$(p-1)(p^{n-1}-p^{\frac{n+s}{2}-1})$}&\scriptsize{$(p-1)(\frac{k}{p}+(p-1)p^{\frac{n-s}{2}-1})$}&\scriptsize{$(p-1)\frac{k}{p}$}\\ \hline
				\scriptsize{	$(p-1)(p^{n-1}+p^{\frac{n+s}{2}-1})$}&\scriptsize{$(p-1)(p^{n-s-1}-\frac{k}{p})$}&\scriptsize{$(p-1)(p^{n-s-1}-\frac{k}{p}-(p-1)p^{\frac{n-s}{2}-1})$}\\\hline
				\scriptsize{$(p-1)p^{n-1}+p^{\frac{n+s}{2}-1}$}&\scriptsize{$(p-1)^2(\frac{k}{p}-p^{\frac{n-s}{2}-1})$}&\scriptsize{$(p-1)^2\ \frac{k}{p}$}\\ \hline
				\scriptsize{	$(p-1)p^{n-1}-p^{\frac{n+s}{2}-1}$}&\scriptsize{$(p-1)^2(p^{n-s-1}-\frac{k}{p})$}&\scriptsize{$(p-1)^2(p^{n-s-1}-\frac{k}{p}
					+p^{\frac{n-s}{2}-1})$}\\\hline
		\end{tabular}}
	\end{table}
	\vspace{-1cm}
	\begin{table}[h]
		\centering
		\caption{The weight distribution of $\mathcal{C}_f$ in Theorem 1 when $n+s$ is odd}
		\renewcommand\arraystretch{1}	
		\begin{tabular}{|c|c|}
			\hline
			\scriptsize{Weight} & \scriptsize{Multiplicity}\\ \hline
			\scriptsize{0}&\scriptsize{1}\\ \hline
			\scriptsize{	$(p-1)p^{n-1}$}&\scriptsize{$p^{n+1}-(p-1)^2p^{n-s-1}-1$}\\ \hline
			\scriptsize{	$(p-1)p^{n-1}-p^{\frac{n+s-1}{2}}$}&\scriptsize{$\frac{(p-1)^2}{2}(p^{\frac{n-s-1}{2}}+p^{n-s-1})$}\\ \hline
			\scriptsize{	$(p-1)p^{n-1}+p^{\frac{n+s-1}{2}}$}&\scriptsize{$\frac{(p-1)^2}{2}(p^{n-s-1}-p^{\frac{n-s-1}{2}})$}\\ \hline
		\end{tabular}
	\end{table}
	\vspace{-1cm}
	\begin{table}[h]
		\centering
		\caption{The weight distribution of $\widetilde{\mathcal{C}}_{f}$ in Theorem 2 when $n+s$ is even}
		\renewcommand\arraystretch{2.5}	
		\resizebox{\textwidth}{!}{
			\begin{tabular}{|c|c|c|}
				\hline
				\Huge{Weight}&	\Huge{Multiplicity($0\in B_+(f)$)}&	\Huge{Multiplicity($0\in B_-(f)$)}\\\hline
				\Huge{	0}&\Huge{1}&\Huge{1}\\\hline
				\Huge{	$(p-1)^2(p^{n-2}-\epsilon_0p^{\frac{n+s}{2}-2})$}&	\Huge{$p^{n+1}-p^{n-s+1}$}&\Huge{$p^{n+1}-p^{n-s-1}$}\\\hline
				\Huge{	$(p-1)^2p^{n-2}$}&	\Huge{$\frac{k}{p}+(p-1)p^{\frac{n-s}{2}-1}-1$}&	\Huge{$p^{n-s-1}-(p-1)p^{\frac{n-s}{2}-1}-\frac{k}{p}-1$}\\\hline
				\Huge{	$(p-1)^2(p^{n-2}-2\epsilon_0p^{\frac{n+s}{2}-2})$}&	\Huge{$p^{n-s-1}-\frac{k}{p}$}&	\Huge{$\frac{k}{p}$}\\\hline
				\Huge{$(p-1)(p^{n-1}-p^{n-2}-\epsilon_0p^{\frac{n+s}{2}-1})$}&	\Huge{$(p-1)(\frac{2k}{p}+(p-2)p^{\frac{n-s}{2}-1}-1)$}&	\Huge{$(p-1)(2p^{n-s-1}-(p-2)p^{\frac{n-s}{2}-1}-\frac{2k}{p}-1)$}\\\hline
				\Huge{	$(p-1)((p-1)p^{n-2}-\epsilon_0(p-2)p^{\frac{n+s}{2}-2})$}&	\Huge{$2(p-1)(p^{n-s-1}-\frac{k}{p})$}&	\Huge{$2(p-1)\frac{k}{p}$}\\\hline
				\Huge{	$(p-1)(p^{n-1}-\epsilon_0p^{\frac{n+s}{2}-1})$}&	\Huge{$p-1$}&	\Huge{$p-1$}\\\hline
				\Huge{	$(p-1)^2p^{n-2}-\epsilon_0(p-2)p^{\frac{n+s}{2}-1}$}&	\Huge{$(p-1)^2(\frac{k}{p}-p^{\frac{n-s}{2}-1})$}&	\Huge{$(p-1)^2(p^{n-s-1}+p^{\frac{n-s}{2}-1}-\frac{k}{p})$}\\\hline
				\Huge{	$(p-1)^2p^{n-2}-\epsilon_0(p^2-2p+2)p^{\frac{n+s}{2}-2}$}&	\Huge{$(p-1)^2(p^{n-s-1}-\frac{k}{p})$}&	\Huge{$(p-1)^2\frac{k}{p}$}\\\hline
		\end{tabular}}
	\end{table}
	\vspace{-1cm}
	\begin{table}[h]
		\centering
		\caption{The weight distribution of $\widetilde{\mathcal{C}}_{f}$ in Theorem 2 when $n+s$ is odd}
		\renewcommand\arraystretch{0.9}	
		{
			\begin{tabular}{|c|c|}
				\hline
				\scriptsize{Weight}&\scriptsize{Multiplicity}\\\hline
				\scriptsize{0}&\scriptsize{1}\\\hline
				\scriptsize{$(p-1)^2p^{n-2}$}&\scriptsize{$p^{n+1}-p^{n-s+1}+p^{n-s}-p$}\\\hline
				\scriptsize{$(p-1)(p^{n-1}-p^{n-2}+p^{\frac{n+s-3}{2}})$}&\scriptsize{$\frac{(p-1)}{2}(p^{n-s-1}+p^{\frac{n-s-1}{2}})$}\\\hline
				\scriptsize{$(p-1)(p^{n-1}-p^{n-2}-p^{\frac{n+s-3}{2}})$}&\scriptsize{$\frac{(p-1)}{2}(p^{n-s-1}-p^{\frac{n-s-1}{2}})$}\\\hline
				\scriptsize{$(p-1)^2p^{n-2}-p^{\frac{n+s-3}{2}}$}&\scriptsize{$\frac{(p-1)^2}{2}(p^{n-s-1}+p^{\frac{n-s-1}{2}})$}\\\hline
				\scriptsize{$(p-1)^2p^{n-2}+p^{\frac{n+s-3}{2}}$}&\scriptsize{$\frac{(p-1)^2}{2}(p^{n-s-1}-p^{\frac{n-s-1}{2}})$}\\\hline
				\scriptsize{$(p-1)p^{n-1}$}&\scriptsize{$p-1$}\\\hline
		\end{tabular}}
	\end{table}
	\vspace{-1cm}
	\begin{table}[h]
		\centering
		\caption{The weight distribution of $\mathcal{C}_{D_{f,0}\setminus\{0\}}$ in Theorem 3 when $n+s$ is even}
		\renewcommand\arraystretch{0.9}	
		\resizebox{\textwidth}{!}{
			\begin{tabular}{|c|c|c|}
				\hline
				{\footnotesize Weight} &{\footnotesize Multiplicity($0\in B_+(f)$)}&{\footnotesize Multiplicity($0\in B_-(f)$)}\\\hline
				{\footnotesize 0}&{\footnotesize 1}&{\footnotesize 1}\\
				\hline
				{\footnotesize $(p-1)(p^{n-2}+\epsilon_0(p-1)p^{\frac{n+s}{2}-2})$}&{\footnotesize $p^n-p^{n-s}$}&{\footnotesize $p^n-p^{n-s}$}\\
				\hline
				{\footnotesize $(p-1)p^{n-2}$}&{\footnotesize $\frac{k}{p}+(p-1)p^{\frac{n-s}{2}-1}-1$}&{\footnotesize $p^{n-s-1}-(p-1)p^{\frac{n-s}{2}-1}-\frac{k}{p}-1$}\\
				\hline
				{\footnotesize $(p-1)(p^{n-2}+2\epsilon_0(p-1)p^{\frac{n+s}{2}-2})$}&{\footnotesize $p^{n-s-1}-\frac{k}{p}$}&{\footnotesize $\frac{k}{p}$}\\
				\hline
				{\footnotesize $(p-1)(p^{n-2}+\epsilon_0p^{\frac{n+s}{2}-1})$}&{\footnotesize $(p-1)(\frac{k}{p}-p^{\frac{n-s}{2}-1})$}&{\footnotesize $(p-1)(p^{n-s-1}+p^{\frac{n-s}{2}-1}-\frac{k}{p})$}\\
				\hline
				{\footnotesize $(p-1)(p^{n-2}+\epsilon_0(p-2)p^{\frac{n+s}{2}-2})$}&{\footnotesize $ (p-1)(p^{n-s-1}-\frac{k}{p})$}&{\footnotesize $(p-1)\frac{k}{p}$}\\
				\hline
		\end{tabular}}
		
	\end{table}
	\vspace{-1cm}
	\begin{table}[h]
		\centering
		\caption{The weight distribution of $\mathcal{C}_{D_{f,0}\setminus\{0\}}$ in Theorem 3 when $n+s$ is odd}
		\renewcommand\arraystretch{1}	
		{
			\begin{tabular}{|c|c|}
				\hline
				{\scriptsize Weight}&{\scriptsize Multiplicity}\\ \hline
				{\scriptsize 0}&{\scriptsize 1}\\ \hline
				{\scriptsize $(p-1)p^{n-2}$}&{\scriptsize $p^n-(p-1)p^{n-s-1}-1$}\\ \hline
				{\scriptsize $(p-1)(p^{n-2}-p^{\frac{n+s-3}{2}})$}&{\scriptsize $\frac{(p-1)}{2}(p^{n-s-1}+p^{\frac{n-s-1}{2}})$}\\ \hline
				{\scriptsize $(p-1)(p^{n-2}+p^{\frac{n+s-3}{2}})$}&{\scriptsize $\frac{(p-1)}{2}(p^{n-s-1}-p^{\frac{n-s-1}{2}})$}\\ \hline
		\end{tabular}}
	\end{table}
	\vspace{-1cm}
	\begin{table}[h]
		\caption{The weight distributions of $\mathcal{C}_{D_{f,sq}}$ and $\mathcal{C}_{D_{f,nsq}}$  in Theorem 4 when $n+s$ is even}
		\renewcommand\arraystretch{1.8}	
		\resizebox{\textwidth}{!}{
			\begin{tabular}{|c|c|c|}
				\hline
				{\Large Weight} &{\Large Multiplicity($0\in B_+(f)$)}& {\Large Multiplicity($0\in B_-(f)$)}\\\hline
				{\Large 0}&{\Large 1}&{\Large 1}\\
				\hline
				{\Large $\frac{(p-1)^2}{2}(p^{n-2}-\epsilon_0p^{\frac{n+s}{2}-2})$}&{\Large $p^n-p^{n-s}$}&{\Large $p^n-p^{n-s}$}\\
				\hline
				{\Large $\frac{(p-1)^2}{2}p^{n-2}$}&{\Large $\frac{(p+1)}{2}\frac{k}{p}+\frac{(p-1)}{2}p^{\frac{n-s}{2}-1}-1$}&{\Large $\frac{(p+1)}{2}(p^{n-s-1}-\frac{k}{p})-\frac{(p-1)}{2}p^{\frac{n-s}{2}-1}-1$}\\
				\hline
				{\Large $\frac{(p-1)^2}{2}(p^{n-2}-2\epsilon_0p^{\frac{n+s}{2}-2} )$}&{\Large $\frac{(p+1)}{2}(p^{n-s-1}-\frac{k}{p})$}&{\Large $\frac{(p+1)}{2}\frac{k}{p}$}\\
				\hline
				{\Large $\frac{(p-1)}{2}(p^{n-1}-p^{n-2}-2\epsilon_0p^{\frac{n+s}{2}-1})$}&{\Large $\frac{(p-1)}{2}(\frac{k}{p}-p^{\frac{n-s}{2}-1})$}&{\Large $\frac{(p-1)}{2}(p^{n-s-1}-\frac{k}{p}+p^{\frac{n-s}{2}-1})$}\\
				\hline
				{\Large 	$\frac{(p-1)}{2}(p^{n-1}-p^{n-2}+2\epsilon_0p^{\frac{n+s}{2}-2})$}&{\Large $ \frac{(p-1)}{2}(p^{n-s-1}-\frac{k}{p})$}&{\Large $\frac{(p-1)}{2}\frac{k}{p}$}\\
				\hline
		\end{tabular}}
	\end{table}
	\vspace{-1cm}
	\begin{table}[h]
		\centering
		\caption{The weight distribution of $\mathcal{C}_{D_{f,sq}}$ in Theorem 4 when $n+s$ is odd and $0\in B_+(f)$ and the weight distribution of $\mathcal{C}_{D_{f,nsq}}$ in Theorem 4 when $n+s$ is odd and $0\in B_-(f)$}
		\renewcommand\arraystretch{2}	
		\resizebox{\textwidth}{!}{
			\begin{tabular}{|c|c|c|}
				\hline
				{\Large Weight} & { \Large Multiplicity of $\mathcal{C}_{D_{f,sq}}$($0\in B_+(f))$} &{\Large Multiplicity of $\mathcal{C}_{D_{f,nsq}}$($0\in B_-(f))$} \\ \hline
				{\Large 0}&{\Large 1}&{\Large 1}\\ \hline
				{ \Large$\frac{(p-1)^2}{2}(p^{n-2}+p^{\frac{n+s-3}{2}})$}&{\Large $p^n-p^{n-s}+\frac{(p-1)}{2}(p^{n-s-1}-p^{\frac{n-s-1}{2}})$}&{\Large $p^n-p^{n-s}+\frac{(p-1)}{2}(p^{n-s-1}-p^{\frac{n-s-1}{2}})$}\\ 
				\hline
				{\Large $\frac{(p-1)^2}{2}p^{n-2}$}&{\Large $\frac{k}{p}-1$}&{\Large $p^{n-s-1}-\frac{k}{p}-1$}\\ \hline
				{\Large $\frac{(p-1)^2}{2}(p^{n-2}+2p^{\frac{n+s-3}{2}})$}&{\Large $p^{n-s-1}-\frac{k}{p}$}&{\Large $\frac{k}{p}$}\\ \hline
				{\Large $\frac{(p-1)^2}{2}p^{n-2}+\frac{(p^2-1)}{2}p^{\frac{n+s-3}{2}}$}&{\Large $\frac{(p-1)}{2}(\frac{k}{p}+p^{\frac{n-s-1}{2}})$}&{\Large $\frac{(p-1)}{2}(p^{n-s-1}-\frac{k}{p}+p^{\frac{n-s-1}{2}})$}\\\hline
				{\Large $\frac{(p-1)^2}{2}p^{n-2}+\frac{(p-1)(p-3)}{2}p^{\frac{n+s-3}{2}}$}&{\Large $\frac{(p-1)}{2}(p^{n-s-1}-\frac{k}{p})$}&{\Large $\frac{(p-1)}{2}\frac{k}{p}$}\\\hline
		\end{tabular}}
	\end{table}
	\vspace{-1cm}
	\begin{table}[h]
		\caption{The weight distribution of $\mathcal{C}_{D_{f,sq}}$ in Theorem 4 when $n+s$ is odd and $0\in B_-(f)$ and the weight distribution of $\mathcal{C}_{D_{f,nsq}}$ in Theorem 4 when $n+s$ is odd and $0\in B_+(f)$}
		\renewcommand\arraystretch{2}	
		\resizebox{\textwidth}{!}{
			\begin{tabular}{|c|c|c|}
				\hline
				{\Large 	Weight} &{\Large Multiplicity of  $\mathcal{C}_{D_{f,sq}}$($0\in B_-(f)$) }&{\Large Multiplicity of $\mathcal{C}_{D_{f,nsq}}$($0\in B_+(f)$)}\\\hline
				{\Large 	0}&{\Large 1}&{\Large 1}\\ \hline
				{\Large $\frac{(p-1)^2}{2}(p^{n-2}-p^{\frac{n+s-3}{2}})$}&{\Large $p^n-p^{n-s}+\frac{(p-1)}{2}(p^{n-s-1}+p^{\frac{n-s-1}{2}})$}&{\Large $p^n-p^{n-s}+\frac{(p-1)}{2}(p^{n-s-1}+p^{\frac{n-s-1}{2}})$}\\
				\hline
				{\Large $\frac{(p-1)^2}{2}(p^{n-2}-2p^{\frac{n+s-3}{2}})$}&{\Large $\frac{k}{p}$}&{\Large $p^{n-s-1}-\frac{k}{p}$}\\ \hline
				{\Large $\frac{(p-1)^2}{2}p^{n-2}$}&{\Large $p^{n-s-1}-\frac{k}{p}-1$}&{\Large $\frac{k}{p}-1$}\\ \hline
				{\Large 	$\frac{(p-1)^2}{2}p^{n-2}-\frac{(p-1)(p-3)}{2}p^{\frac{n+s-3}{2}}$}&{\Large $\frac{(p-1)k}{2p}$}&{\Large $\frac{(p-1)}{2}(p^{n-s-1}-\frac{k}{p})$}\\\hline
				{\Large 	$\frac{(p-1)^2}{2}p^{n-2}-\frac{(p^2-1)}{2}p^{\frac{n+s-3}{2}}$}&{\Large $\frac{(p-1)}{2}(p^{n-s-1}-p^{\frac{n-s-1}{2}}-\frac{k}{p})$}&{\Large $\frac{(p-1)}{2}(\frac{k}{p}-p^{\frac{n-s-1}{2}})$}\\
				\hline
		\end{tabular}}
	\end{table}
	\vspace{-1cm}
	\begin{table}[h]
		\centering
		\caption{The weight distribution of {\scriptsize $\widetilde{\mathcal{C}}_{D_{f,0}\setminus\{0\}}$} in Theorem 5 when $n+s$ is even}
		\renewcommand\arraystretch{0.9}	
		\resizebox{\textwidth}{!}{
			\begin{tabular}{|c|c|c|}
				\hline
				{\scriptsize Weight} &{\scriptsize Multiplicity($0\in B_+(f)$)}&{\scriptsize Multiplicity($0\in B_-(f)$)}\\\hline
				{\scriptsize0}&{\scriptsize1}&{\scriptsize1}\\
				\hline
				{\scriptsize$p^{n-2}+\epsilon_0(p-1)p^{\frac{n+s}{2}-2}$}&{\scriptsize$p^n-p^{n-s}$}&{\scriptsize$p^n-p^{n-s}$}\\
				\hline
				{\scriptsize$p^{n-2}$}&{\scriptsize$\frac{k}{p}+(p-1)p^{\frac{n-s}{2}-1}-1$}&{\scriptsize$p^{n-s-1}-(p-1)p^{\frac{n-s}{2}-1}-\frac{k}{p}-1$}\\
				\hline
				{\scriptsize$p^{n-2}+2\epsilon_0(p-1)p^{\frac{n+s}{2}-2}$}&{\scriptsize$p^{n-s-1}-\frac{k}{p}$}&{\scriptsize$\frac{k}{p}$}\\
				\hline
				{\scriptsize$p^{n-2}+\epsilon_0p^{\frac{n+s}{2}-1}$}&{\scriptsize$(p-1)(\frac{k}{p}-p^{\frac{n-s}{2}-1})$}&{\scriptsize$(p-1)(p^{n-s-1}+p^{\frac{n-s}{2}-1}-\frac{k}{p})$}\\
				\hline
				{\scriptsize$p^{n-2}+\epsilon_0(p-2)p^{\frac{n+s}{2}-2}$}&{\scriptsize$ (p-1)(p^{n-s-1}-\frac{k}{p})$}&{\scriptsize$(p-1)\frac{k}{p}$}\\
				\hline
		\end{tabular}}
	\end{table}
	\vspace{-2cm}
	\begin{table}[h]
		\centering
		\caption{The weight distribution of {\scriptsize $\widetilde{\mathcal{C}}_{D_{f,0}\setminus\{0\}}$} in Theorem 5 when $n+s$ is odd}
		\renewcommand\arraystretch{0.9}	
		{
			\begin{tabular}{|c|c|}
				\hline
				{\scriptsize Weight}&{\scriptsize Multiplicity}\\ \hline
				{\scriptsize 0}&{\scriptsize 1}\\ \hline
				{\scriptsize $p^{n-2}$}&{\scriptsize $p^n-(p-1)p^{n-s-1}-1$}\\ \hline
				{\scriptsize $p^{n-2}-p^{\frac{n+s-3}{2}}$}&{\scriptsize $\frac{(p-1)}{2}(p^{n-s-1}+p^{\frac{n-s-1}{2}})$}\\ \hline
				{\scriptsize $p^{n-2}+p^{\frac{n+s-3}{2}}$}&{\scriptsize $\frac{(p-1)}{2}(p^{n-s-1}-p^{\frac{n-s-1}{2}})$}\\ \hline
		\end{tabular}}
	\end{table}
	\vspace{-1cm}
	\begin{table}[h]
		\caption{The weight distributions of {\scriptsize $\widetilde{\mathcal{C}}_{D_{f,sq}}$} and {\scriptsize $\widetilde{\mathcal{C}}_{D_{f,nsq}}$}  in Theorem 6 when $n+s$ is even}
		\renewcommand\arraystretch{1.1}	
		\resizebox{\textwidth}{!}{
			\begin{tabular}{|c|c|c|}
				\hline
				{\large Weight} &{\large Multiplicity($0\in B_+(f)$)}& {\large Multiplicity($0\in B_-(f)$)}\\\hline
				{\large 0}&{\large 1}&{\large 1}\\
				\hline
				{\large $\frac{(p-1)}{2}(p^{n-2}-\epsilon_0p^{\frac{n+s}{2}-2})$}&{\large $p^n-p^{n-s}$}&{\large $p^n-p^{n-s}$}\\
				\hline
				{\large $\frac{(p-1)}{2}p^{n-2}$}&{\large $\frac{(p+1)}{2}\frac{k}{p}+\frac{(p-1)}{2}p^{\frac{n-s}{2}-1}-1$}&{\large $\frac{(p+1)}{2}(p^{n-s-1}-\frac{k}{p})-\frac{(p-1)}{2}p^{\frac{n-s}{2}-1}-1$}\\
				\hline
				{\large $\frac{(p-1)}{2}(p^{n-2}-2\epsilon_0p^{\frac{n+s}{2}-2} )$}&{\large $\frac{(p+1)}{2}(p^{n-s-1}-\frac{k}{p})$}&{\large $\frac{(p+1)}{2}\frac{k}{p}$}\\
				\hline
				{\large $\frac{p^{n-1}-p^{n-2}-2\epsilon_0p^{\frac{n+s}{2}-1}}{2}$}&{\large $\frac{(p-1)}{2}(\frac{k}{p}-p^{\frac{n-s}{2}-1})$}&{\large $\frac{(p-1)}{2}(p^{n-s-1}-\frac{k}{p}+p^{\frac{n-s}{2}-1})$}\\
				\hline
				{\large	$\frac{p^{n-1}-p^{n-2}+2\epsilon_0p^{\frac{n+s}{2}-2}}{2}$}&{\large $ \frac{(p-1)}{2}(p^{n-s-1}-\frac{k}{p})$}&{\large $\frac{(p-1)}{2}\frac{k}{p}$}\\
				\hline
		\end{tabular}}
	\end{table}
	
	\begin{table}[h]
		\centering
		\caption{The weight distribution of {\scriptsize $\widetilde{\mathcal{C}}_{D_{f,sq}}$} in Theorem 6 when $n+s$ is odd and $0\in B_+(f)$ and the weight distribution of {\scriptsize $\widetilde{\mathcal{C}}_{D_{f,nsq}}$} in Theorem 6 when $n+s$ is odd and $0\in B_-(f)$}
		\renewcommand\arraystretch{1.2}	
		\resizebox{\textwidth}{!}{
			\begin{tabular}{|c|c|c|}
				\hline
				{\large Weight} & { \large Multiplicity of $\widetilde{\mathcal{C}}_{D_{f,sq}}$($0\in B_+(f))$} &{\large Multiplicity of $\widetilde{\mathcal{C}}_{D_{f,nsq}}$($0\in B_-(f))$} \\ \hline
				{\large 0}&{\large 1}&{\large 1}\\ \hline
				{ \large$\frac{(p-1)}{2}(p^{n-2}+p^{\frac{n+s-3}{2}})$}&{\large $p^n-p^{n-s}+\frac{(p-1)}{2}(p^{n-s-1}-p^{\frac{n-s-1}{2}})$}&{\large $p^n-p^{n-s}+\frac{(p-1)}{2}(p^{n-s-1}-p^{\frac{n-s-1}{2}})$}\\ 
				\hline
				{\large $\frac{(p-1)}{2}p^{n-2}$}&{\large $\frac{k}{p}-1$}&{\large $p^{n-s-1}-\frac{k}{p}-1$}\\ \hline
				{\large $\frac{(p-1)}{2}(p^{n-2}+2p^{\frac{n+s-3}{2}})$}&{\large $p^{n-s-1}-\frac{k}{p}$}&{\large $\frac{k}{p}$}\\ \hline
				{\large $\frac{(p-1)}{2}p^{n-2}+\frac{p+1}{2}p^{\frac{n+s-3}{2}}$}&{\large $\frac{(p-1)}{2}(\frac{k}{p}+p^{\frac{n-s-1}{2}})$}&{\large $\frac{(p-1)}{2}(p^{n-s-1}-\frac{k}{p}+p^{\frac{n-s-1}{2}})$}\\\hline
				{\large $\frac{(p-1)}{2}p^{n-2}+\frac{p-3}{2}p^{\frac{n+s-3}{2}}$}&{\large $\frac{(p-1)}{2}(p^{n-s-1}-\frac{k}{p})$}&{\large $\frac{(p-1)}{2}\frac{k}{p}$}\\\hline
		\end{tabular}}
	\end{table}
	\vspace{-15pt}
	\begin{table}[h]
		\caption{The weight distribution of {\scriptsize $\widetilde{\mathcal{C}}_{D_{f,sq}}$ }in Theorem 6 when $n+s$ is odd and $0\in B_-(f)$ and the weight distribution of {\scriptsize $\widetilde{\mathcal{C}}_{D_{f,nsq}}$} in Theorem 6 when $n+s$ is odd and $0\in B_+(f)$}
		\renewcommand\arraystretch{1.1}	
		\resizebox{\textwidth}{!}{
			\begin{tabular}{|c|c|c|}
				\hline
				{\large 	Weight} &{\large Multiplicity of  $\widetilde{\mathcal{C}}_{D_{f,sq}}$($0\in B_-(f)$) }&{\large Multiplicity of $\widetilde{\mathcal{C}}_{D_{f,nsq}}$($0\in B_+(f)$)}\\\hline
				{\large 	0}&{\large 1}&{\large 1}\\ \hline
				{\large $\frac{(p-1)}{2}(p^{n-2}-p^{\frac{n+s-3}{2}})$}&{\large $p^n-p^{n-s}+\frac{(p-1)}{2}(p^{n-s-1}+p^{\frac{n-s-1}{2}})$}&{\large $p^n-p^{n-s}+\frac{(p-1)}{2}(p^{n-s-1}+p^{\frac{n-s-1}{2}})$}\\
				\hline
				{\large $\frac{(p-1)}{2}(p^{n-2}-2p^{\frac{n+s-3}{2}})$}&{\large $\frac{k}{p}$}&{\large $p^{n-s-1}-\frac{k}{p}$}\\ \hline
				{\large $\frac{(p-1)}{2}p^{n-2}$}&{\large $p^{n-s-1}-\frac{k}{p}-1$}&{\large $\frac{k}{p}-1$}\\ \hline
				{\large 	$\frac{(p-1)}{2}p^{n-2}-\frac{p-3}{2}p^{\frac{n+s-3}{2}}$}&{\large $\frac{(p-1)k}{2p}$}&{\large $\frac{(p-1)}{2}(p^{n-s-1}-\frac{k}{p})$}\\\hline
				{\large 	$\frac{(p-1)}{2}p^{n-2}-\frac{p+1}{2}p^{\frac{n+s-3}{2}}$}&{\large $\frac{(p-1)}{2}(p^{n-s-1}-p^{\frac{n-s-1}{2}}-\frac{k}{p})$}&{\large $\frac{(p-1)}{2}(\frac{k}{p}-p^{\frac{n-s-1}{2}})$}\\
				\hline
		\end{tabular}}
	\end{table}
	\begin{table}[h]
		\centering
		\caption{ Optimal codes constructed by Theorems \ref{Th 2}, \ref{Th 5} and \ref{Th 6}}
		\renewcommand\arraystretch{0.8}	
		\resizebox{\textwidth}{!}{\begin{tabular}{|c|c|c|c|}
				\hline
				{ \scriptsize Conditions}&{\scriptsize Code}&{\scriptsize Parameters}&{\scriptsize Reference}\\
				\hline
				{\scriptsize $p=3,n=3,s=1,\epsilon_0=1,k=9$}&{\scriptsize $\widetilde{\mathcal{C}}_f$}&{\scriptsize $[12,4,6]$}&{\scriptsize Theorem \ref{Th 2}}\\
				\hline
				{\scriptsize $p=3, n=3,s=1,\epsilon_0=1$}&{\scriptsize ${\widetilde{\mathcal{C}}_f}^{\bot}$}&{\scriptsize $[12,8,3]$}&{\scriptsize Theorem \ref{Th 2}}\\
				\hline
				{\scriptsize $p=3, n=3,s=1,\epsilon_0=-1$}&{\scriptsize ${\widetilde{\mathcal{C}}_f}^{\bot}$}&{\scriptsize $[24,20,3]$}&{\scriptsize Theorem \ref{Th 2}}\\
				\hline
				{\scriptsize $p=3,n=4,s=0,\epsilon_0=1,k=81$}&{\scriptsize $\widetilde{\mathcal{C}}_f$}&{\scriptsize $[48,5,30]$}&{\scriptsize Theorem \ref{Th 2}}\\
				\hline
				{\scriptsize $p=3, n=4,s=0,\epsilon_0=1$}&{\scriptsize ${\widetilde{\mathcal{C}}_f}^{\bot}$}&{\scriptsize $[48,43,3]$}&{\scriptsize Theorem \ref{Th 2}}\\
				\hline
				{\scriptsize $p=3, n=4,s=0,\epsilon_0=-1$}&{\scriptsize ${\widetilde{\mathcal{C}}_f}^{\bot}$}&{\scriptsize $[60,55,3]$}&{\scriptsize Theorem \ref{Th 2}}\\
				\hline
				{\scriptsize $p=3, n=4,s=2,\epsilon_0=1$}&{\scriptsize ${\widetilde{\mathcal{C}}_f}^{\bot}$}&{\scriptsize $[36,31,3]$}&{\scriptsize Theorem \ref{Th 2}}\\
				\hline
				{\scriptsize $p=3,n=4,s=2,\epsilon_0=-1$}&{\scriptsize ${\widetilde{\mathcal{C}}_f}^{\bot}$}&{\scriptsize $[72,67,3]$}&{\scriptsize Theorem \ref{Th 2}}\\
				\hline
				{\scriptsize $p=3,n=5,s=1,\epsilon_0=1$}&{\scriptsize ${\widetilde{\mathcal{C}}_f}^{\bot}$}&{\scriptsize $[144,138,3]$}&{\scriptsize Theorem \ref{Th 2}}\\
				\hline
				{\scriptsize $p=3,n=5,s=1,\epsilon_0=-1$}&{\scriptsize ${\widetilde{\mathcal{C}}_f}^{\bot}$}&{\scriptsize $[180,174,3]$}&{\scriptsize Theorem \ref{Th 2}}\\
				\hline
				{\scriptsize $p=3,n=5,s=3,\epsilon_0=1$}&{\scriptsize ${\widetilde{\mathcal{C}}_f}^{\bot}$}&{\scriptsize $[108,102,3]$}&{\scriptsize Theorem \ref{Th 2}}\\
				\hline
				{\scriptsize $p=3,n=5,s=3,\epsilon_0=-1$}&{\scriptsize ${\widetilde{\mathcal{C}}_f}^{\bot}$}&{\scriptsize $[216,210,3]$}&{\scriptsize Theorem \ref{Th 2}}\\
				\hline
				{\scriptsize $p=5,n=3,s=1,\epsilon_0=1$}&{\scriptsize ${\widetilde{\mathcal{C}}_f}^{\bot}$}&{\scriptsize $[80,76,3]$}&{\scriptsize Theorem \ref{Th 2}}\\
				\hline
				{\scriptsize $p=5,n=3,s=1,\epsilon_0=-1$}&{\scriptsize ${\widetilde{\mathcal{C}}_f}^{\bot}$}&{\scriptsize $[120,116,3]$}&{\scriptsize Theorem \ref{Th 2}}\\
				\hline
				{\scriptsize $p=3,n=3,s=0 \ or\ s=2$}&{\scriptsize ${\widetilde{\mathcal{C}}_f}^{\bot}$}&{\scriptsize $[18,14,3]$}&{\scriptsize Theorem \ref{Th 2}}\\
				\hline
				{\scriptsize $p=3,n=4,s=1\ or\ s=3 $}&{\scriptsize ${\widetilde{\mathcal{C}}_f}^{\bot}$}&{\scriptsize $[54,49,3]$}&{\scriptsize Theorem \ref{Th 2}}\\
				\hline
				{\scriptsize $p=3,n=5,s=0\ or\ s=2\ or\ s=4$}&{\scriptsize ${\widetilde{\mathcal{C}}_f}^{\bot}$}&{\scriptsize $[162,156,3]$}&{\scriptsize Theorem \ref{Th 2}}\\
				\hline
				{\scriptsize $p=5,n=2,s=1$}&{\scriptsize ${\widetilde{\mathcal{C}}_f}^{\bot}$}&{\scriptsize $[20,17,3]$}&{\scriptsize Theorem \ref{Th 2}}\\
				\hline
		\end{tabular}}
	\end{table}
	\begin{table}[h]
		\centering
		\caption*{TABLE \uppercase\expandafter{\romannumeral15} continued}
		\resizebox{\textwidth}{0.7\linewidth}{\begin{tabular}{|c|c|c|c|}
				\hline
				{\scriptsize Conditions}&{\scriptsize Code}&{\scriptsize Parameters}&{\scriptsize Reference}\\
				\hline
					{\scriptsize $p=5,n=3,s=0\ or \ s=2 $}&{\scriptsize ${\widetilde{\mathcal{C}}_f}^{\bot}$}&{\scriptsize $[100,96,3]$}&{\scriptsize Theorem \ref{Th 2}}
				\\
				\hline
				{\scriptsize $p=7,n=2,s=1$}&{\scriptsize ${\widetilde{\mathcal{C}}_f}^{\bot}$}&{\scriptsize $[42,39,3]$}&{\scriptsize Theorem \ref{Th 2}}\\
				\hline
				
				{\scriptsize $p=3,n=4,s=0,\epsilon_0=1$}&{\scriptsize ${\widetilde{\mathcal{C}}_{D_{f,0}\setminus\{0\}}}^{\bot}$}&{\scriptsize $[16,12,3]$}&{\scriptsize Theorem \ref{Th 5}}\\		
				\hline
				{\scriptsize $p=3,n=4,s=0,\epsilon_0=-1,k=0$}&{\scriptsize ${\widetilde{\mathcal{C}}_{D_{f,0}\setminus\{0\}}}$}&{\scriptsize $[10,4,6]$}&{\scriptsize Theorem \ref{Th 5}}\\		
				\hline
				{\scriptsize $p=3,n=4,s=0,\epsilon_0=-1,k=0$}&{\scriptsize ${\widetilde{\mathcal{C}}_{D_{f,0}\setminus\{0\}}}^{\bot}$}&{\scriptsize $[10,6,4]$}&{\scriptsize Theorem \ref{Th 5}}\\		
				\hline
				{\scriptsize $p=3,n=5,s=1,\epsilon_0=-1$}&{\scriptsize ${\widetilde{\mathcal{C}}_{D_{f,0}\setminus\{0\}}}^{\bot}$}&{\scriptsize $[31,26,3]$}&{\scriptsize Theorem \ref{Th 5}}\\		
				\hline
				{\scriptsize $p=3,n=5,s=1,\epsilon_0=-1,k=0$}&{\scriptsize ${\widetilde{\mathcal{C}}_{D_{f,0}\setminus\{0\}}}$}&{\scriptsize $[31,5,18]$}&{\scriptsize Theorem \ref{Th 5}}\\		
				\hline
				{\scriptsize $p=3,n=5,s=1,\epsilon_0=1$}&{\scriptsize ${\widetilde{\mathcal{C}}_{D_{f,0}\setminus\{0\}}}^{\bot}$}&{\scriptsize $[49,44,3]$}&{\scriptsize Theorem \ref{Th 5}}\\
				\hline
				{\scriptsize $p=3,n=6,s=0,\epsilon_0=-1$}&{\tiny ${\widetilde{\mathcal{C}}_{D_{f,0}\setminus\{0\}}}^{\bot}$}&{\scriptsize $[112,106,3]$}&{\scriptsize Theorem \ref{Th 5}}\\		
				\hline
				{\scriptsize $p=3,n=6,s=0,\epsilon_0=-1,k=0$}&{\tiny ${\widetilde{\mathcal{C}}_{D_{f,0}\setminus\{0\}}}$}&{\scriptsize $[112,6,72]$}&{\scriptsize Theorem \ref{Th 5}}\\		
				\hline
				{\scriptsize $p=3,n=6,s=0,\epsilon_0=1$}&{\tiny ${\widetilde{\mathcal{C}}_{D_{f,0}\setminus\{0\}}}^{\bot}$}&{\scriptsize $[130,124,3]$}&{\scriptsize Theorem \ref{Th 5}}\\		
				\hline
				{\scriptsize 	$p=3,n=6,s=2,\epsilon_0=-1$}&{\tiny ${\widetilde{\mathcal{C}}_{D_{f,0}\setminus\{0\}}}^{\bot}$}&{\scriptsize $[94,88,3]$}&{\scriptsize Theorem \ref{Th 5}}\\
				\hline
				{\scriptsize $p=3,n=6,s=2,\epsilon_0=1$}&{\tiny ${\widetilde{\mathcal{C}}_{D_{f,0}\setminus\{0\}}}^{\bot}$}&{\scriptsize $[148,142,3]$}&{\scriptsize Theorem \ref{Th 5}}\\
				\hline
				{\scriptsize $p=5,n=4,s=0,\epsilon_0=1$}&{\tiny ${\widetilde{\mathcal{C}}_{D_{f,0}\setminus\{0\}}}^{\bot}$}&{\scriptsize $[36,32,3]$}&{\scriptsize Theorem \ref{Th 5}}\\		
				\hline
				{\scriptsize $p=5,n=4,s=0,\epsilon_0=-1,k=0$}&{\tiny ${\widetilde{\mathcal{C}}_{D_{f,0}\setminus\{0\}}}$}&{\scriptsize $[26,4,20]$}&{\scriptsize Theorem \ref{Th 5}}\\		
				\hline
				{\scriptsize $p=5,n=4,s=0,\epsilon_0=-1,k=0$}&{\tiny ${\widetilde{\mathcal{C}}_{D_{f,0}\setminus\{0\}}}^{\bot}$}&{\scriptsize $[26,22,4]$}&{\scriptsize Theorem \ref{Th 5}}\\		
				\hline
				
				{\scriptsize $p=3,n=3,s=0$}&{\tiny $\widetilde{\mathcal{C}}_{D_{f,0}\setminus\{0\}}$}&{\scriptsize $[4,3,2]$}&{\scriptsize Theorem \ref{Th 5}}\\
				\hline
				{\scriptsize $p=3,n=3,s=0$}&{\tiny ${\widetilde{\mathcal{C}}_{D_{f,0}\setminus\{0\}}}^{\bot}$}&{\scriptsize $[4,1,4]$}&{\scriptsize Theorem \ref{Th 5}}\\
				\hline
				{\scriptsize $p=3,n=4,s=1$}&{\tiny ${\widetilde{\mathcal{C}}_{D_{f,0}\setminus\{0\}}}^{\bot}$}&{\scriptsize $[13,9,3]$}&{\scriptsize Theorem \ref{Th 5}}\\
				\hline
				{\scriptsize 	$p=3,n=5,s=0\ or\ s=2$}&{\tiny ${\widetilde{\mathcal{C}}_{D_{f,0}\setminus\{0\}}}^{\bot}$}&{\scriptsize $[40,35,3]$}&{\scriptsize Theorem \ref{Th 5}}\\
				\hline
				{\scriptsize $p=3,n=6,s=1,\ or\ s=3$}&{\tiny ${\widetilde{\mathcal{C}}_{D_{f,0}\setminus\{0\}}}^{\bot}$}&{\scriptsize $[121,115,3]$}&{\scriptsize Theorem \ref{Th 5}}\\
				\hline
				{\scriptsize $p=5,n=3,s=0$}&{\tiny $\widetilde{\mathcal{C}}_{D_{f,0}\setminus\{0\}}$}&{\scriptsize $[6,3,4]$}&{\scriptsize Theorem \ref{Th 5}}\\
				\hline
				{\scriptsize $p=5,n=4,s=1$}&{\tiny ${\widetilde{\mathcal{C}}_{D_{f,0}\setminus\{0\}}}^{\bot}$}&{\scriptsize $[31,27,3]$}&{\scriptsize Theorem \ref{Th 5}}\\
				\hline
				{\scriptsize $p=7,n=3,s=0$}&{\tiny $\widetilde{\mathcal{C}}_{D_{f,0}\setminus\{0\}}$}&{\scriptsize $[8,3,6]$}&{\scriptsize Theorem \ref{Th 5}}\\
				\hline
				{\scriptsize $p=7,n=3,s=0$}&{\tiny ${\widetilde{\mathcal{C}}_{D_{f,0}\setminus\{0\}}}^{\bot}$}&{\scriptsize $[8,5,4]$}&{\scriptsize Theorem \ref{Th 5}}\\
				\hline
				
		\end{tabular}}
	\end{table}
	\begin{table}[h]
		\centering
		\caption*{TABLE \uppercase\expandafter{\romannumeral15} continued}
		\renewcommand\arraystretch{0.8}	
		\resizebox{\textwidth}{!}{\begin{tabular}{|c|c|c|c|}
				\hline
				{\scriptsize Conditions}&{\scriptsize Code}&{\scriptsize Parameters}&{\scriptsize Reference}\\
				\hline
				{\scriptsize 	$p=3,n=4,s=0,\epsilon_0=-1,k=0$}&{\tiny $\widetilde{\mathcal{C}}_{D_{f,sq}}$}&{\scriptsize $[15,4,9]$}&{\scriptsize Theorem \ref{Th 6}}\\
				\hline
				{\scriptsize $p=3,n=4,s=0,\epsilon_0=-1$}&{\tiny ${\widetilde{\mathcal{C}}_{D_{f,sq}}}^{\bot}$}&{\scriptsize $[15,11,3]$}&{\scriptsize Theorem \ref{Th 6}}\\
				\hline
				{\scriptsize 	$p=3,n=5,s=1,\epsilon_0=-1$}&{\tiny ${\widetilde{\mathcal{C}}_{D_{f,sq}}}^{\bot}$}&{\scriptsize $[45,40,3]$}&{\scriptsize Theorem \ref{Th 6}}\\
				\hline
				{\scriptsize $p=3,n=6,s=0,\epsilon_0=1$}&{\tiny ${\widetilde{\mathcal{C}}_{D_{f,sq}}}^{\bot}$}&{\scriptsize $[117,111,3]$}&{\scriptsize Theorem \ref{Th 6}}\\
				\hline
				{\scriptsize $p=3,n=6,s=0,\epsilon_0=-1$}&{\tiny ${\widetilde{\mathcal{C}}_{D_{f,sq}}}^{\bot}$}&{\scriptsize $[126,120,3]$}&{\scriptsize Theorem \ref{Th 6}}\\
				\hline
				{\scriptsize $p=3,n=6,s=0,\epsilon_0=-1,k=0$}&{\tiny $\widetilde{\mathcal{C}}_{D_{f,sq}}$}&{\scriptsize $[126,6,81]$}&{\scriptsize Theorem \ref{Th 6}}\\
				\hline
				{\scriptsize 	$p=3,n=6,s=2,\epsilon_0=-1$}&{\tiny ${\widetilde{\mathcal{C}}_{D_{f,sq}}}^{\bot}$}&{\scriptsize $[135,129,3]$}&{\scriptsize Theorem \ref{Th 6}}\\
				\hline
				{\scriptsize 	$p=5,n=4,s=0,\epsilon_0=1$}&{\tiny ${\widetilde{\mathcal{C}}_{D_{f,sq}}}^{\bot}$}&{\scriptsize $[60,56,3]$}&{\scriptsize Theorem \ref{Th 6}}\\
				\hline
				{\scriptsize 	$p=5,n=4,s=0,\epsilon_0=-1$}&{\tiny ${\widetilde{\mathcal{C}}_{D_{f,sq}}}^{\bot}$}&{\scriptsize $[65,61,3]$}&{\scriptsize Theorem \ref{Th 6}}\\
				\hline
				
				{\scriptsize $p=3,n=5,s=2,\epsilon_0=-1$}&{\tiny ${\widetilde{\mathcal{C}}_{D_{f,sq}}}^{\bot}$}&{\scriptsize $[27,22,3]$}&{\scriptsize Theorem \ref{Th 6}}\\
				\hline
				{\scriptsize $p=3,n=6,s=3,\epsilon_0=-1$}&{\tiny ${\widetilde{\mathcal{C}}_{D_{f,sq}}}^{\bot}$}&{\scriptsize $[81,75,3]$}&{\scriptsize Theorem \ref{Th 6}}\\
				\hline
				{\scriptsize $p=3,n=7,s=4,\epsilon_0=-1$}&{\tiny ${\widetilde{\mathcal{C}}_{D_{f,sq}}}^{\bot}$}&{\scriptsize $[243,236,3]$}&{\scriptsize Theorem \ref{Th 6}}\\
				\hline
				{\scriptsize $p=5,n=3,s=0,\epsilon_0=1$}&{\tiny ${\widetilde{\mathcal{C}}_{D_{f,sq}}}^{\bot}$}&{\scriptsize $[15,12,3]$}&{\scriptsize Theorem \ref{Th 6}}\\
				\hline
				{\scriptsize $p=5,n=3,s=0,\epsilon_0=-1,k=0$}&{\tiny $\widetilde{\mathcal{C}}_{D_{f,sq}}$}&{\scriptsize $[10,3,7]$}&{\scriptsize Theorem \ref{Th 6}}\\
				\hline
				{\scriptsize $p=5,n=3,s=0,\epsilon_0=1$}&{\tiny ${\widetilde{\mathcal{C}}_{D_{f,sq}}}^{\bot}$}&{\scriptsize $[10,7,3]$}&{\scriptsize Theorem \ref{Th 6}}\\
				\hline
				{\scriptsize $p=5,n=4,s=1,\epsilon_0=1$}&{\tiny ${\widetilde{\mathcal{C}}_{D_{f,sq}}}^{\bot}$}&{\scriptsize $[75,71,3]$}&{\scriptsize Theorem \ref{Th 6}}\\
				\hline
				{\scriptsize $p=5,n=4,s=1,\epsilon_0=-1$}&{\tiny ${\widetilde{\mathcal{C}}_{D_{f,sq}}}^{\bot}$}&{\scriptsize $[50,46,3]$}&{\scriptsize Theorem \ref{Th 6}}\\
				\hline	
				{\scriptsize $p=7,n=3,s=0,\epsilon_0=1$}&{\tiny ${\widetilde{\mathcal{C}}_{D_{f,sq}}}^{\bot}$}&{\scriptsize $[28,25,3]$}&{\scriptsize Theorem \ref{Th 6}}\\
				\hline
				{\scriptsize $p=7,n=3,s=0,\epsilon_0=-1,k=0$}&{\tiny $\widetilde{\mathcal{C}}_{D_{f,sq}}$}&{\scriptsize $[21,3,17]$}&{\scriptsize Theorem \ref{Th 6}}\\
				\hline
				{\scriptsize $p=7,n=3,s=0,\epsilon_0=1$}&{\tiny ${\widetilde{\mathcal{C}}_{D_{f,sq}}}^{\bot}$}&{\scriptsize $[21,18,3]$}&{\scriptsize Theorem \ref{Th 6}}\\
				\hline
				\end{tabular}}
			\end{table}
	
	\vspace{-5cm}
	\begin{table}[h]
		\centering
		\caption{Comparing our $7$-ary pure quantum codes given in Theorems \ref{Th 7} and \ref{Th 9} with that in \cite{Edel}}
		\renewcommand\arraystretch{0.85}	
		\resizebox{\textwidth}{!}{
			\begin{tabular}{|c|c|c|c|c|c|}
				\hline
				{\scriptsize Conditions}&{\scriptsize Our quantum codes}&{\scriptsize Rate}&{\scriptsize Quantum codes in \cite{Edel}}&{\scriptsize Rate}&{\scriptsize Reference}\\
				\hline
				{\scriptsize $n=3,s=1,\epsilon_0=1$}&{\scriptsize $[[252,247,3]]_7$}&{\scriptsize $0.980$}&{\scriptsize $[[255,237,3]]_7$}&{\scriptsize $0.929$}&{\scriptsize Theorem \ref{Th 7}}\\
				\hline
				{\scriptsize 	$n=3,s=1,\epsilon_0=-1$}&{\scriptsize $[[336,331,3]]_7$}&{\scriptsize $0.985$}&{\scriptsize $[[340,324,3]]_7$}&{\scriptsize $0.953$}&{\scriptsize Theorem \ref{Th 7}}\\
				\hline
				{\scriptsize $n=3,s=0$}&{\scriptsize $[[294,289,3]]_7$}&{\scriptsize $0.983$}&{\scriptsize $[[300,294,3]]_7$}&{\scriptsize $0.98$}&{\scriptsize Theorem \ref{Th 7}}\\
				\hline
				{\scriptsize 	$n=2,s=1$}&{\scriptsize $[[42,38,3]]_7$}&{\scriptsize $0.905$}&{\scriptsize $[[43,37,3]]_7$}&{\scriptsize $0.860$}&{\scriptsize Theorem \ref{Th 7}}\\
				\hline
				{\scriptsize $n=4,s=0,\epsilon_0=-1$}&{\scriptsize $[[175,169,3]]_7$}&{\scriptsize $0.966$}&{\scriptsize $[[180,170,3]]_7$}&{\scriptsize $0.944$}&{\scriptsize Theorem \ref{Th 9}}\\
				\hline
				{\scriptsize 	$n=3,s=0,\epsilon_0=1$}&{\scriptsize $[[28,23,3]]_7$}&{\scriptsize $0.821$}&{\scriptsize $[[30,24,3]]_7$}&{\scriptsize $0.8$}&{\scriptsize Theorem \ref{Th 9}}\\
				\hline
				{\scriptsize $n=3,s=0,\epsilon_0=-1$}&{\scriptsize $[[21,16,3]]_7$}&{\scriptsize $0.762$}&{\scriptsize $[[20,15,3]]_7$}&{\scriptsize $0.75$}&{\scriptsize Theorem \ref{Th 9}}\\
				\hline 
				{\scriptsize $n=4,s=1,\epsilon_0=1$}&{\scriptsize $[[196,190,3]]_7$}&{\scriptsize $0.969$}&{\scriptsize $[[200,192,3]]_7$}&{\scriptsize $0.96$}&{\scriptsize Theorem \ref{Th 9}}\\
				\hline
				{\scriptsize $n=4,s=1,\epsilon_0=-1$}&{\scriptsize $[[147,141,3]]_7$}&{\scriptsize $0.959$}&{\scriptsize $[[144,137,3]]_7$}&{\scriptsize $0.951$}&{\scriptsize Theorem \ref{Th 9}}\\
				\hline
		\end{tabular}}
	\end{table}
	\vspace{-5cm}
	\begin{table}[h]
		\centering
		\caption{Optimal ternary LCD codes constructed by Theorems \ref{Th 10}, \ref{Th 11} and \ref{Th 12}}
			\renewcommand\arraystretch{0.95}	
		\begin{tabular}{|c|c|c|}
			\hline
			{\scriptsize Conditions}&{\scriptsize Parameters}&{\scriptsize Reference}\\
			\hline
			{\scriptsize $ n=4,s=2,\epsilon_0=1$}&{\scriptsize $[41,36,3]$}&{\scriptsize Theorem \ref{Th 10}}\\
			\hline
			{\scriptsize $n=4,s=2,\epsilon_0=-1$}&{\scriptsize $[77,72,3]$}&{\scriptsize Theorem \ref{Th 10}}\\
			\hline
			{\scriptsize $n=5,s=1,\epsilon_0=1$}&{\scriptsize $[150,144,3]$}&{\scriptsize Theorem \ref{Th 10}}\\
			\hline
			{\scriptsize $n=5,s=1,\epsilon_0=-1$}&{\scriptsize $[186,180,3]$}&{\scriptsize Theorem \ref{Th 10}}\\
			\hline
			{\scriptsize $n=5,s=3,\epsilon_0=1$}&{\scriptsize $[114,108,3]$}&{\scriptsize Theorem \ref{Th 10}}\\
			\hline
			{\scriptsize $n=5,s=3,\epsilon_0=-1$}&{\scriptsize $[222,216,3]$}&{\scriptsize Theorem \ref{Th 10}}\\
			\hline
			{\scriptsize $n=3,s=2$}&{\scriptsize $[22,18,3]$}&{\scriptsize Theorem \ref{Th 10}}\\
			\hline
			{\scriptsize $n=4,s=1\ or\ s=3$}&{\scriptsize $[59,54,3]$}&{\scriptsize Theorem \ref{Th 10}}\\
			\hline
			{\scriptsize $n=5,s=0\ or\ s=2\ or\ s=4$}&{\scriptsize $[168,162,3]$}&{\scriptsize Theorem \ref{Th 10}}\\
			\hline
			{\scriptsize $n=5,s=1,\epsilon_0=-1$}&{\scriptsize $[36,31,3]$}&{\scriptsize Theorem \ref{Th 11}}\\		
			\hline
			{\scriptsize $n=5,s=1,\epsilon_0=1$}&{\scriptsize $[54,49,3]$}&{\scriptsize Theorem \ref{Th 11}}\\
			\hline
			{\scriptsize $n=6,s=0,\epsilon_0=-1$}&{\scriptsize $[118,112,3]$}&{\scriptsize Theorem \ref{Th 11}}\\
			\hline
			{\scriptsize $n=6,s=0,\epsilon_0=1$}&{\scriptsize $[136,130,3]$}&{\scriptsize Theorem \ref{Th 11}}\\
			\hline
			{\scriptsize $n=6,s=2,\epsilon_0=-1$}&{\scriptsize $[100,94,3]$}&{\scriptsize Theorem \ref{Th 11}}\\
			\hline
			{\scriptsize $n=6,s=2,\epsilon_0=1$}&{\scriptsize $[154,148,3]$}&{\scriptsize Theorem \ref{Th 11}}\\
			\hline
			{\scriptsize $n=4,s=1$}&{\scriptsize $[17,13,3]$}&{\scriptsize Theorem \ref{Th 11}}\\
			\hline
			{\scriptsize 	$n=5,s=0\ or\ s=2$}&{\scriptsize $[45,40,3]$}&{\scriptsize Theorem \ref{Th 11}}\\
			\hline
			{\scriptsize $n=6,s=1,\ or\ s=3$}&{\scriptsize $[127,121,3]$}&{\scriptsize Theorem \ref{Th 11}}\\
			\hline
			{\scriptsize 	$n=5,s=1,\epsilon_0=-1$}&{\scriptsize $[50,45,3]$}&{\scriptsize Theorem \ref{Th 12}}\\
			\hline
			{\scriptsize $n=6,s=0,\epsilon_0=1$}&{\scriptsize $[123,117,3]$}&{\scriptsize Theorem \ref{Th 12}}\\
			\hline
			{\scriptsize $n=6,s=0,\epsilon_0=-1$}&{\scriptsize $[132,126,3]$}&{\scriptsize Theorem \ref{Th 12}}\\
			\hline
			{\scriptsize 	$n=6,s=2,\epsilon_0=-1$}&{\scriptsize $[141,135,3]$}&{\scriptsize Theorem \ref{Th 12}}\\
			\hline
			{\scriptsize $n=5,s=2,\epsilon_0=-1$}&{\scriptsize $[32,27,3]$}&{\scriptsize Theorem \ref{Th 12}}\\
			\hline
			{\scriptsize $n=6,s=3,\epsilon_0=-1$}&{\scriptsize $[87,81,3]$}&{\scriptsize Theorem \ref{Th 12}}\\
			\hline
		\end{tabular}
	\end{table}

	\end{document}